\documentclass[10pt]{amsart}

\DeclareMathAlphabet{\mathpzc}{OT1}{pzc}{m}{it}

\usepackage{ytableau}
\usepackage{upgreek}
\usepackage{ marvosym }
\usepackage[hyperfootnotes=false]{hyperref}
\usepackage{enumerate}
\usepackage{tikz}
\usepackage{tikz-cd}
\usepackage{hyperref}
\usepackage{enumerate}
\usepackage[normalem]{ulem}
\usepackage{bm}
\usepackage{amsthm,amsfonts,amsmath,amscd,amssymb}
\usepackage{xcolor}
\usepackage{mathrsfs}
\usepackage{graphics,graphicx}
\usepackage{caption}
\let\oldmarginpar\marginpar
\renewcommand\marginpar[1]{\-\oldmarginpar[\raggedleft\footnotesize #1]%
{\raggedright\footnotesize #1}}
\theoremstyle{plain}
\newtheorem{thm}{Theorem}[section]
\newtheorem{theorem}{Theorem}[section]
\newtheorem{lemma}[thm]{Lemma}
\newtheorem{proposition}[thm]{Proposition}
\newtheorem{example}[thm]{Example}

\theoremstyle{definition}

\newtheorem{definition}[thm]{Definition}
\newtheorem{remark}[thm]{Remark}

\theoremstyle{remark}

\numberwithin{equation}{section}

\newcommand{\Lie}{\mathcal L} 
\newcommand{\g}{\mathfrak{g}}

\renewcommand{\P}{\mathbb{P}}

\newcommand{\R}{\mathbb{R}}

\newcommand{\SC}{\mathcal{C}}

\newcommand{\la}{\lambda}

\renewcommand{\a}{\alpha}

\newcommand{\dd}{\partial}

\newcommand{\sse}{\subseteq}
\newcommand{\lr}{\longrightarrow}

\newcommand{\Aut}{\operatorname{Aut}}

\newcommand{\SL}{\operatorname{SL}}
\newcommand{\GL}{\operatorname{GL}}

\newcommand{\wt}{\widetilde}

\newcommand{\Id}{\text{Id}}

\newcommand{\st}{\text{st}}

\newcommand{\ct}{{\mathcal T}}

\newcommand{\ce}{{\mathcal E}}

\def\Op{{\mathcal O}{\hspace{-.05mm}\mathpzc  p}\,}
\newcounter{daggerfootnote}

  \newcommand{\rpl}                         
{\mbox{$
\begin{picture}(12.7,8)(-.5,-1)
\put(0,0.2){$+$}
\put(3.7,2.8){\oval(6.7,6.7)[r]}
\end{picture}$}}

\newcommand{\lpl}                         
{\mbox{$
\begin{picture}(12.7,8)(-.5,-1)
\put(2,0.2){$+$}
\put(6,2.8){\oval(6.7,6.7)[l]}
\end{picture}$}}

\newcommand{\hh}{\hspace{.1mm}}

\def\sideremark#1{\ifvmode\leavevmode\fi\vadjust{\vbox to0pt{\vss
 \hbox to 0pt{\hskip\hsize\hskip1em
 \vbox{\hsize2cm\tiny\raggedright\pretolerance10000
  \noindent #1\hfill}\hss}\vbox to8pt{\vfil}\vss}}}

\newcommand{\nn}[1]{(\ref{#1})}

\DeclareMathOperator{\Vol}{Vol}
\DeclareMathOperator{\Zero}{Zero}

\DeclareMathOperator{\Span}{span}
\DeclareMathOperator{\End}{End}

\newcommand{\SP}{\mathscr{P}}
\newcommand{\SH}{\mathscr{H}}
\renewcommand{\sp}{\mathfrak{sp}}

\newcommand{\pdot}{{\textstyle\boldsymbol \cdot}\hspace{.05mm}}

\newcommand{\order}{{\mathcal O}}

\newcommand{\Cone}{{\mathcal C}}
\newcommand{\FZ}{{\mathfrak Z}}

\begin{document}

\title{
Dynamical Quantization of Contact Structures}
%
\author{Roger Casals}
\address{University of California Davis, Department of Mathematics, One Shields Avenue, Davis, CA 95616, USA}
\email{casals@math.ucdavis.edu}
\author{Gabriel Herczeg}
\address{Brown University, Department of Physics, Hope Street, Providence, RI 02912, USA}
\email{gabrielherczeg@brown.edu}
\author{Andrew Waldron}
\address{University of California Davis, Dept. of Mathematics and Center for Quantum Mathematics and Physics (QMAP), One Shields Avenue, Davis, CA 95616, USA}
\email{wally@math.ucdavis.edu}

\maketitle

\begin{abstract}\vspace{-1cm}
We construct a dynamical quantization for contact manifolds in terms of a flat connection acting on a Hilbert tractor bundle. We show that this contact quantization, which is independent of the choice of contact form, can be obtained by quantizing the Reeb dynamics of an ambient strict contact manifold equivariantly with respect to an ${\mathbb R}^+$-action. The contact quantization further determines a certain contact tractor connection whose parallel sections determine a distinguished choice of Reeb dynamics and their quantization. This relationship relies on tractor constructions from parabolic geometries and mirrors the tight relationship between Einstein metrics and conformal geometries. Finally, we construct in detail the dynamical quantization of the unique tight contact structure on the 3-sphere, where the Holstein-Primakoff transformation makes a surprising appearance.
\end{abstract}


\section{Introduction}\label{sec:intro}

The object of this article is the quantization of $(2n+1)$-dimensional contact manifolds $(Z,\xi)$. The dynamical quantization we present uses the differential geometry of tractor calculi to simultaneously quantize the contact Floer action functionals associated to every choice of contact 1-form for $\xi$.

\subsection{Contact Structures in Quantum Mechanics} The physical principle of diffeomorphism covariance puts time  on par with classical phase space configurations.
Once a choice of dynamics is specified,
 the $2n$-dimensional symplectic geometry underlying a classical phase space can be  enhanced to a $(2n+1)$-dimensional strict contact geometry~$(Z,\alpha)$, where $\alpha$ is a contact 1-form~\cite{ArnoldSing,ArGi}. 
The dynamical quantization of the {\it strict} contact manifold $(Z,\alpha)$ will be called a {\it strict quantization} and is defined directly below. This strict quantization depends on the choice of contact 1-form $\alpha$ in an intricate manner. In light of the development of symplectic field theory~\cite{EliashbergGiventalHofer00,Eliashberg_ICM98,Eliashberg_ICM07} and the foundational study~\cite{BEHWZ_SFT,Pardon_CH} of the contact Floer action functional $S:\SP(Z)\longrightarrow\R$, defined by
\begin{equation}\label{action}S_\alpha(\gamma)=\int_\gamma\alpha,\end{equation}
with $\SP(Z)$ a path space for $Z$, robust physical properties of a system can be studied in terms of the contact structure $(Z,\xi)$ itself, $\xi:=\ker\alpha$, independently of the choice of contact 1-form $\alpha$. Reeb dynamics, arising as the minimizing trajectories for the action functional $S_\a$, undergo drastic changes when differing choices of contact 1-form are made, {\it e.g.}~there are $C^\infty$-close ergodic perturbations of the periodic Reeb flow on the standard contact sphere $(S^{2n+1},\xi_\st)$ \cite{CasalsSpacil_ChernWeil}. Nevertheless, qualitative properties which hold for {\it all} possible Reeb dynamics associated to a contact structure $(Z,\xi)$ exist and are a central object of study in contact topology---for instance, the existence of at least one closed periodic orbit for {\it any} Reeb field, known as the Weinstein Conjecture, has been proven for all contact 3-manifolds, and a rather large class of higher-dimensional contact manifolds; see {\it e.g.} \cite{Hofer93,Taubes07,Oancea15_FreeLoop}. An intrinsic quantization of contact structures is thus important in order to understand these qualitatively stable properties of a quantum system associated to a contact phase-space, independent of a choice of contact 1-form. The main contribution of this article is such a quantization.

The problem of quantizing a contact structure $(Z,\xi)$, instead of a fixed choice of contact 1-form~$\alpha$, is challenging precisely because the associated dynamics change for each choice of contact 1-form. To resolve this tension, we first note that the relation between a strict contact geometry, {\it i.e.}~a choice of contact 1-form $\alpha$, and its underlying contact structure $\xi=\ker(\alpha)$, is analogous to the relation between a Riemannian geometry and its corresponding conformal geometry. Our results rely heavily on ideas used to treat Riemannian geometries covariantly with respect to conformal rescalings of the metric tensor~\cite{BEG}. 
Indeed, our covariant formulation of quantum mechanics  is in line with  mathematical methods in modern contact topology \cite{Pardon_CH, BourgeoisOancea_SHCH,BourgeoisOancea_SH}, whereby the theory is geometrically meaningful for a given contact 1-form $\alpha$, and yet the resulting construction is independent of the chosen $\alpha\in\Omega^1(Z)$ if $\xi=\ker(\a)$. The simplest general contact invariant of $(Z,\xi)$ is contact homology understood as a rigorous reformulation of the $S^1$-equivariant Morse homology of the Floer action functional~$S_\alpha$. The perturbative ground states---Reeb trajectories in $\SP(Z)$---depend on the choice of~$\alpha$, whereas the true ground states ought to be represented as contact homology classes, independent of $\alpha$, and thus independent of the physical choice of a time evolution.

\subsubsection*{Strict Dynamical Quantization} 

Let us first summarize the dynamical quantization of $(Z,\xi)$ when a fixed choice of contact 1-form $\alpha$ has been made, and thus the Reeb dynamics are also fixed. In the standard description of quantum mechanics, quantum states
 are elements of a Hilbert space~${\mathcal H}$.
Once dynamics are included, these states depend on classical parameters, including a time coordinate as measured by a classical laboratory clock. 
The time evolution of states is governed by the Schr\"odinger equation.
In a covariant approach, the classical laboratory and classical dynamics are modeled by a strict contact manifold $(Z^{2n+1},\alpha)$. The kernel of $\alpha$  determines a distribution $\xi$, and in turn the bundle of symplectic frames on $\xi$ produces an~$\operatorname{Sp}(2n)$ principal bundle. Given a Hilbert space~${\mathcal H}$ admitting a unitary, projective, metaplectic $\operatorname{Sp}(2n)$ representation, one has an associated vector bundle ${\mathcal H}Z \to Z$ with Hilbert space fibers ${\mathcal H}$. The inner product on~${\mathcal H}$ then gives a sesquilinear map on sections~$\psi,\chi\in \Gamma({\mathcal H}Z)$, denoted by $\langle \chi|\psi\rangle$, taking values in complex-valued smooth functions of~$Z$. 
The bundle of frames ${\mathscr M}\to Z$ preserving $\langle\cdot|\cdot\rangle$ is a principal bundle with structure group $U({\mathcal H})$ of unitary Hilbert transformations. We term the associated vector bundle ${\mathscr H}:= {\mathscr M}\times_{U({\mathcal H})} {\mathcal H}$,  the  {\it Hilbert bundle associated to $(Z,\xi,{\mathcal H})$} and recycle the notation $\langle \cdot|\cdot\rangle$ for the natural inner product on its section space. These sections may be viewed as a generalization of time-dependent wavefunctions, and thus~$\SH$ may be viewed as the bundle of wavefunctions.  The additional data of a flat connection on the Hilbert bundle~$\SH$ preserving $\langle\cdot|\cdot\rangle$, {\it i.e.} 
$$
d\langle \chi|\psi\rangle
=\langle \chi|\nabla \psi\rangle+
\langle \nabla \chi|\psi\rangle\, ,
\qquad
\forall\: \psi,\chi\in \Gamma({\mathscr H}) \, ,
$$
   allows generalized time evolution to be formulated in terms of parallel transport.
   Indeed, the Schr\"odinger equation is upgraded to a parallel transport condition on sections $\Psi\in\Gamma({\mathscr H})$,
$$
\nabla \Psi=0\, .
$$
Hence we call the data 
$$
(\SH,\nabla)\, ,\mbox{ where } \nabla^2=0\, ,
$$
a {\it quantum dynamical system} and $\nabla$ a {\it quantum connection}.

At this point, it is interesting to ask when a quantum dynamical system is the quantization of an underlying classical Reeb dynamics. In earlier work we showed how to construct a quantum dynamical system from Reeb dynamics~\cite{HW,HLW,LOW}. We formalize that construction in the following definition.

\begin{definition}\label{QMdef}
Let $(Z^{2n+1},\xi)$ be a contact manifold with a fixed a choice of contact 1-form  $\alpha$ where~$\xi=\ker(\alpha)$, ${\mathcal H}$ a projective Hilbert space representation of the symplectic group ${\rm Sp}(2n)$, and~$(\SH,\nabla^\hbar)$ a one-parameter family of quantum dynamical systems associated to this data. By definition, $\nabla^\hbar$ is a {\it dynamical quantization of the Reeb dynamics of $\alpha$} if  
$$\lim_{\hbar \to 0} i\hbar \nabla^\hbar \big|_{\hbar = 0} =\alpha.$$
\hfill$\blacksquare$\end{definition}

\noindent
By the above-displayed limit, we mean that for all $\Psi\in \Gamma(\SH)$ and $u\in \Gamma(TZ)$,
$$\lim_{\hbar \to 0} i\hbar \nabla^\hbar_u \Psi \big|_{\hbar = 0} =\alpha(u) \Psi.$$
Moreover, if $d^{\hat A^\hbar }$ is a connection form for $\nabla^\hbar$, it must have small $\hbar$-asymptotics
$$
d^{\hat A^\hbar }~\sim \frac{\alpha}{i\hbar} + \frac{\hat \kappa}{i\surd\hbar}+d+{\mathcal O}(\hbar^0)\, ,
$$
where, because $\nabla^\hbar$ is flat, the 1-form $\hat \kappa$ takes values in $\operatorname{End}(\Gamma(\SH))$ and obeys
$$\hat \kappa_u \circ \hat \kappa_v -
\hat \kappa_v \circ \hat \kappa_v=-id\alpha(u,v)
$$
for any $u,v\in \Gamma(TZ)$.
The {\it r\^ole} of the map $\hat \kappa$ is to calibrate classical quantities to quantum observables.

In Section~\ref{sec:StrictQuant}, we give a detailed description of quantum mechanics and 
quantization  of generic dynamical systems 
in terms of the above dynamical quantization  of strict contact manifolds. 

\subsection{Summary of Contributions} The central theme is the {\it contact quantization} of a contact manifold $(Z,\xi)$, as introduced below. The necessary quantum ingredients and techniques for this quantization are developed throughout the article, especially in Sections \ref{sec:StrictQuant} and \ref{sec:Quant}, and the quantization of the standard contact 3-sphere $(S^3,\xi_\st)$ is worked out in detail in Section \ref{S3}. The manuscript also contributes new results and constructions regarding the differential geometry of tractor bundles as they relate to contact structures; this is the focus of Section \ref{sec:Contract} and these geometric results are 
key for
our contact quantization in Section \ref{sec:Quant}.
 
Let us summarize the basic notions and thematic pillars appearing in our contact quantization. The data of a contact structure $\xi$ on a $(2n+1)$-manifold $Z$ can be encoded by a principal $\R^+$-bundle~${\mathcal C}$ over~$Z$, the symplectization of $(Z,\xi)$, see {\it e.g.} \cite{ArnoldGivental01,Geiges08}.
Tautologically, ${\mathcal C}$ is equipped with a 1-form $\lambda_\xi$ and the bundle of symplectic frames on $T{\mathcal C}$ is an $Sp(2n+2)$ principal bundle. The  projective metaplectic representation of $Sp(2n+2)$ acting on a Hilbert space ${\mathcal H}$  gives an embedding of $Sp(2n+2)$  into the group of unitary transformations
$U({\mathcal H})$. Hence this determines a $U({\mathcal H})$ principal bundle over~${\mathcal C}$; its associated vector bundle for such a choice of the  vector space ${\mathcal H}$ will be denoted~${\SH}(\xi)$. Since the metaplectic representation of $Sp(2n+2)$ acting on a Hilbert space ${\mathcal H}$ induces a unitary representation of the Heisenberg algebra (see Section~\ref{sec:metaplectic}), the bundle~$\SH(\xi)$ comes equipped with a {\it Heisenberg map} $s$ that maps contractors to unitary endomorphisms of the section space $\Gamma(\SH(\xi))$, see Section~\ref{HBC}. The map $s$ can viewed as a contact analog of 
 Clifford multiplication and symplectic spinors~\cite{Kostant}.

Before introducing the dynamical quantization of a contact manifold $(Z,\xi)$, in Definition \ref{ContQMdef} below, we first describe precisely the notion of a symplectic quantization on the symplectization of $(Z,\xi)$:

\begin{definition}\label{QMcontdef} Let $(Z,\xi)$ be a contact manifold. A one-parameter family
of quantum dynamical systems $(\SH(\xi),{ \nabla}^\hbar)$ is said to be a {\it symplectic quantization} of the contact structure $\xi$ if
$$
\lim_{\hbar \to 0} i\hbar {\nabla}^\hbar \big|_{\hbar = 0} =
\lambda_\xi\, .
$$
\hfill $\blacksquare$	
\end{definition}

\noindent
The quantization in Definition \ref{QMcontdef} takes place on the symplectic manifold $({\mathcal C},d\la_\xi)$. To ensure that this descends to a quantization along the contact manifold $Z$, we further impose an $\R^+$-equivariance property on~$\nabla^\hbar$, as follows. Given a contact manifold $(Z,\xi)$, we let $(\nabla^\hbar,\SH(\xi))$ be  a symplectic quantization of~$\xi$ and consider the vector field $X$ infinitesimally generating the $\R^+$-action on the total space $\mathcal{C}$ of the principal $\R^+$-bundle $({\mathcal C},d\la_\xi)\lr(Z,\xi)$. Suppose that $\nabla^\hbar$ obeys the $\R^+$-equivariance condition
\begin{equation}\label{RPE}
\Big[
\nabla^0_X+2\hbar \frac{d}{d\hbar}, \nabla^\hbar\Big]=0\, ,
\end{equation}
where $\nabla^0:=\frac12 \lim_{\hbar\to 0}\frac{d^2 (\hbar \nabla^\hbar)}{(d\surd \hbar)^2}\,
$ is the constant term of $\nabla^\hbar$ expanded as a Laurent series around~${\hbar=0}$. 
\begin{definition} 
\label{ContQMdef}

Let $\SH(\FZ)$ be the bundle  over~$Z$ obtained by quotienting $\SH(\xi)$ by parallel transport with respect to  $\nabla^0$ along the $\R^+$-action, and $\bm \nabla^\hbar$ be the connection  on $\SH(\FZ)$
obtained by restricting~$\nabla^\hbar$ to  $Z$, viewed as the space of $\R^+$-rays in ${\mathcal C}$. By definition, the pair $(\SH(\FZ),\bm \nabla^\hbar)$ is said to be a {\it contact quantization} of~$(Z,\xi)$.
\hfill $\blacksquare$	
\end{definition}

The contact quantization of Definition~\ref{ContQMdef}, which we also referred to as a {\it dynamical quantization}, combines 
a classical  differential geometry construction with physical BRST quantization. The geometric side of this story relies on tractor calculi for parabolic geometries \cite{CapGov}; an excellent reference for parabolic geometries is~\cite{SlovCap}. In particular, a  contact structure is locally equivalent to
\begin{equation}\label{5grading}(\R^{2n-1},\xi_\st)\cong \left(\frac{\sp(2n+2,\R)}{\R\oplus\sp(2n,\R)\oplus\R^{2n+2}\oplus\R},\sp_{-1}\right)\, .
\end{equation}
This is the flat model associated to the contact 2-grading of the Lie algebra $$\sp(2n+2,\R)\cong\mathfrak{sp}_{-2}\oplus\mathfrak{sp}_{-1}\oplus\mathfrak{sp}_{0}\oplus\mathfrak{sp}_{1}\oplus\mathfrak{sp}_{2},$$
where the parabolic algebra is $(\R\oplus\sp(2n,\R))_0\oplus(\R^{2n+2})_1\oplus(\R)_2$, with $\sp(2n,\R)$ being the semisimple part of $\mathfrak{sp}_0$ and its center $\mathfrak{z}(\mathfrak{sp}_0)=\R$ generated by the grading element. Our quantization relies upon the tractor bundle associated to a curved analog of the above contact Klein geometry. 
This is a rank~$(2n+2)$ vector bundle associated to a $P:=(\R^+\otimes \operatorname{Sp}(2n))\ltimes {\mathbb R}^{2n}\ltimes \R^+$
principal bundle over~$Z$. The details are given in Definition~\ref{def:contractor}; to differentiate it from  tractor bundles for other structures, we call this the {\it contractor bundle} and denote it by $\FZ$. 
There are various ways to handle the bundle~$\FZ$ which will appear in this article. In analogy with the Fefferman-Graham construction~\cite{FeffermanGraham_Conformal} for conformal metrics, we
approach this by introducing an ambient manifold $M$, two dimensions higher than the underlying contact manifold $(Z,\xi)$:

\begin{definition}\label{def:ambientspace}
Let $(Z^{2n+1},\xi)$ be a contact manifold, and  $(M^{2n+3},A)$ a $(2n+3)$-dimensional strict contact manifold, $A\in \Gamma(T^*\!M)$ being a contact 1-form. The data $(M,A,x,i)$, where $x:\R^+\times M\to M$ is a proper free $\R^+$-action generated by a vector field $X\in \Gamma(TM)$, and $i:{\SC}\longrightarrow M$ is an $\R^+$-equivariant embedding, is said to be an ambient contact manifold for $(Z,\xi)$ if the following conditions are satisfied:

\begin{enumerate}[(i)]

	\item\label{poundtime}
	${\mathcal L}_X A=2A$,\\[.1mm]
	\item\label{isquaretominusone} $\Cone=\Zero(A(X)),$
	\\
	\item\label{third} For any contact form $\alpha$ on $(Z,\xi)$, $\exists\sigma \in \Gamma(\SC)$ such that $\alpha = (i \circ \sigma )^{*}(A)$.\label{compatibility}
\hfill $\blacksquare$	
	\end{enumerate}
\end{definition}

For many applications to the differential geometry of contact structures, especially in a parabolic geometry context~\cite{Fox,GNW}, a~$(2n+2)$-dimensional ambient space given by the symplectization ${\mathcal C}$ is employed. This is a specialization of the Thomas space for projective geometries. Indeed it has already been utilized in Definition~\ref{RPE} above. That said, the ambient contact manifold is particularly useful for us, precisely because it is also a contact manifold. Note that this is in line with the $(d+2)$-dimensional Fefferman--Graham manifold used to handle $d$-dimensional conformal geometries. Section~\ref{sec:Contract} will explore the geometric interplay between the contact geometry of $(Z,\xi)$ and that of an associated ambient contact manifold. In particular, our Theorem~\ref{thm:equivalence}
will establish an equivalence between the tractor bundle $\FZ\to Z$ and the ambient contact distribution $\Xi=\ker A$ restricted to the image $i(\mathcal C)$ of the symplectization in the ambient strict contact manifold~$(M,A)$.

The utility of tractor bundles relies in part on the existence of a canonical class of Cartan connections acting on them; in our context, we refer to these connections as {\it contractor connections}; see Definition~\ref{ssec:ContractConnections}. In particular, a contact projective structure ({\it i.e.} a family of paths tangent to a contact distribution and geodesic for some affine connection)
corresponds to a regular Cartan connection which gives a  canonical
contractor connection upon imposing a normality condition (or equivalently in this a context, a torsion free-type condition), see~\cite{Fox}.  For the purposes of our dynamical quantization, only the regularity condition is needed. 
In Definition~\ref{ssec:ContractConnections} we will give an equivalent {\it ambient} definition of a contractor connection and Theorem~\ref{existenceambientconnection} will establish an equivalence between a class of ambient projective connections and ambient contractor connections; these may be viewed as $(2n+3)$-dimensional contact analogs of the symplectization results of~\cite{Fox,GNW}.

These geometric developments and constructions set the stage for relating contact quantization of~$(Z,\xi)$ with strict contact quantization of an associated ambient $(M,A)$. The following result, proved in Section \ref{sec:Quant} using the techniques from Section \ref{sec:Contract}, relates $\R^+$-equivariant quantizations of ambient contact manifolds (see Definition~\ref{strictequi}) to the contact quantization of a contact manifold, as laid out in Definition~\nn{QMcontdef}. 

\begin{theorem}\label{biggy}
Let $(Z,\xi)$ be a contact manifold, ${\mathscr H}(\FZ)\to Z$ a contractor Hilbert bundle and~$(M,A,x,i)$ an ambient contact manifold for~$(Z,\xi)$.
Suppose that the Hilbert bundle~$\SH(
 \ker A)$ determines the contractor Hilbert bundle ${\mathscr H}(\FZ)$.
Then, 
 an $\R^+$-equivariant strict quantum connection~$\nabla^\hbar$ on a Hilbert bundle $\SH(
 \ker A)$
 determines a contact quantization $({\mathscr H}(\FZ),{\bm \nabla}^\hbar)$ of the contact manifold $(Z,\xi)$. 
\end{theorem}

The necessary ingredients and notions featured in the statement of Theorem \ref{biggy} are developed in Sections \ref{sec:Contract} and \ref{sec:Quant}. In short, a contact quantization determines a contractor connection~$\nabla^\FZ$ on the bundle $\FZ$ and Theorem \ref{biggy} shows how to construct such a contractor connection. In general, given any parabolic geometry and a tractor connection determined by a particular representation of the parabolic subgroup, it is interesting to consider the consequences of the existence of a parallel tractor $I\in \Gamma(\FZ)$, {\it i.e.} a tractor satisfying
$$
\nabla^\FZ I =0\, .
$$
In the contact setting, a solution to this condition determines a distinguished contact form, and thus its associated Reeb dynamics. By imposing that the connection $\bm \nabla^\hbar$ leaves the kernel of $s(I)$ invariant,~$s$ being the Heisenberg map, the constraint
$$
s(I)=0
$$
can be imposed on the section space $\Gamma(\SH(\FZ))$. This then determines a quantization of the aforementioned distinguished Reeb dynamics. This procedure is detailed in Section~\ref{parallel}.

Finally, given a contact structure, 
it is possible to find a formal quantization $\nabla^\hbar$ given as formal power series in $\hbar$ that satisfies Definition~\ref{QMdef} order by order in~$\hbar$. This is not surprising in light of the Kontsevich formality theorem~\cite{Kont,Kont1}; see also example~\ref{proto} . In contrast, the problem of going beyond formality is more challenging. We propose a solution for a class of contact structures with enhanced symmetries, detailing the dynamical quantization of the unique tight contact structure in the 3-sphere. Indeed, in Section~\ref{S3} we show that for $(S^3,\xi_\st)$, it is possible to go beyond formality and we will determine both a strict quantum connection $\nabla^\hbar$ and a contact quantum connection $\bm \nabla^\hbar$. The latter determines a contractor connection that admits a parallel scale tractor, and this allows us to give an explicit example of the relationship discussed above between the quantization of contact structures and Reeb dynamics. In fact, our solution exhibits an interesting spectrum shortening phenomenon: on a discrete set of values of the parameter $\hbar$, the Hilbert space fibers given by ${\mathcal H}=L^2({\mathbb R})$ can be consistently truncated to unitary ${\mathfrak{su}}(2)$ representations. In particular, we see a well-known relationship between the harmonic oscillator and spin states, due to Holstein and Primakoff, surprisingly appear in the contact quantization of $(S^3,\xi_\st)$.  

\subsection{Related Systems and Works}
 It might be relevant to emphasize that several families of contact structures arise naturally in mathematical physics, to which we can apply the above quantization. First, the space of geodesics in a smooth manifold $Q$ has a canonical contact structure when viewed as the space of oriented contact elements $(\P(T^*Q),\xi_\st)$. The contact manifold $(\P(T^*Q),\xi_\st)$ is not canonically endowed with a contact $1$-form unless we additionally consider the Liouville structure from the bulk $(T^*Q,\la_\st)$, and thus a (non-strict) contact quantization is in general required.

Second, the quantization of the symplectic coadjoint orbits, as studied by \cite{Kirillov_OrbithMethodI,Kirillov_OrbitMethodII}, becomes a contact geometry problem in the nilpotent orbits. Indeed, by the Jacobson-Morosov Theorem \cite[Section 3.3]{CollingwoodMcGovern_Nilpotent} the nilpotent coadjoint orbits in $\g^*$ admit a symplectic dilation, and thus the projectivized coadjoint orbit in $\P(\g^*)$ is naturally a contact manifold. These contact structures and their infinitesimal automorphisms are studied in \cite{Yamaguchi}. In general, local Lie algebras \cite{Kirillov_LocalLie,Lichnerowicz_JacobiManifolds} are foliated by leaves with contact structures and locally conformal symplectic structures.

Third, in the context of General Relativity, the space of light rays in a globally hyperbolic Lorentz spacetime is a contact structure, and by the contact topology proof of the Low Conjecture \cite{ChernovNemirovski_Causality,ChernovRudyak_Causality}, we now know that events are casually related if and only if their Legendrian fibers are contact linked. It is a natural step forward to study the quantizations of Definitions~\ref{QMdef}
and~\ref{RPE}
with respect to Legendrian submanifolds. The quantization of Legendrian submanifolds in this context and its comparison to the contact augmentation category shall be explored in the future. 
In addition to the aforementioned causality in General Relativity, it was shown in \cite{Lewandowski:1990cg} that realizable CR structures can be naturally associated to algebraically special solutions of Einstein's equations in four spacetime dimensions, which, by virtue of the Goldberg-Sachs theorem \cite{goldberg2009republication} are guaranteed to admit a shear-free, null geodesic congruence. The CR manifold is realized as the leaf space of this congruence, and in the generic case where the congruence has non-vanishing twist, the manifold is contact. Finally, contact structures arising in the thermodynamics of AdS black holes were recently studied in \cite{ghosh2019contact}.

Regarding the scientific context, it is worth noting that~\cite{BiswasDey,Kashiwara} discuss the deformation quantization of complex contact structures in terms of a Moyal-Weyl star product and microlocal sheaf theory, and~\cite{Polesello_EquivDirac,Rajeev} address quantization for contact structures via the virtual representation given by the index of the contact Dirac operator and as a bracket deformation, respectively. It would be enlightening to compare these algebraic constructions with the geometry-based quantizations of Definitions~\ref{QMdef} and~\ref{QMcontdef}, 
which follow from the fact that maximal non-integrability of the contact structure forces a BRST quantization with a maximal set of second class constraints.

Finally, the quantization of a strict contact structure yields the Schr\"odinger equation as the operator governing the wavefunctions (see Secton~\ref{ssec:PhysicalStates});  it would be  desirable to compare it to the path integral formulation of quantum mechanics with action $S_\a$. It ought to be noticed that from the perspective of \cite{Witten_PathIntegral}, the path integral on the (complexified) loop space of the symplectization of a contact structure~$(Z,\xi)$ simplifies when integrated along a canonical coisotropic brane \cite{GukovOrlov_BCC}, which coincides with the base of our Hilbert bundle $(\SH,\nabla)$. 
Note that along these lines, the path integral quantization of the Wilson line operator for quantum connection has been constructed in~\cite{HLW}.\\

{\bf Structure}: The manuscript is organized as follows. Section \ref{sec:StrictQuant} describes the BRST quantization of the action functional associated to a contact form. A discussion on the charge, symmetries and physical states is also included in that section. Section \ref{sec:Contract} develops the tractor formalism in the context of contact structures, with a view towards the dynamical quantization of contact structures. This third section is differential geometric in nature, geometrically valuable on its own, and also offers both an intrinsic and an extrinsic constructions of the contractor bundles and proves their equivalence. Section \ref{sec:Quant} constructs the dynamical quantization of a contact structure $(Z,\xi)$. The results and constructions developed in Sections \ref{sec:StrictQuant} and \ref{sec:Contract} are  crucial to those of Section \ref{sec:Quant}. Finally, Section \ref{S3} discusses in detail the dynamical quantization of the standard contact 3-sphere $(S^3,\xi_\st)$, unraveling an unexpected connection with the Holstein--Primakoff formalism.\\

{\bf Acknowledgements}: We thank Sean Curry, Emanuele Latini and Olindo Corradini for discussions. We are grateful to Eugene Skvortsov for pointing out the Holstein--Primakoff mechanism. R.~Casals is supported by the NSF grant DMS-1841913, the NSF CAREER grant DMS-1942363 and the Alfred P. Sloan Foundation. A.~Waldron is supported by a Simons Foundation Collaboration Grant for Mathematicians ID 317562.

\section{Strict Dynamical Quantization}\label{sec:StrictQuant}

Let $(Z,\xi)$ be a $(2n+1)$-dimensional contact manifold equipped with a fixed choice of contact 1-form~$\alpha\in\Omega^1(Z)$. The object of this section is to quantize the  Reeb dynamics~$\alpha$. Reeb dynamics
and, hence its quantization, depend on the  choice of contact form $\alpha$, and thus is not a priori an invariant of the contact structure $\xi=\ker(\alpha)$. Sections~\ref{sec:Contract} and~\ref{sec:Quant} are devoted to the quantization of  
 the contact manifold $(Z,\xi)$, rather than the strict contact structure~$(Z,\alpha)$ studied here, and use the results from the present section.

\subsection{Contact Action}\label{ssec:ContactAction} Reeb dynamics are governed by critical points of the Floer action functional
$$S_\alpha:L_{z_0}(Z)\lr\R,\qquad S_\alpha(\gamma)=\int_\gamma\alpha\, ,$$
where $L_{z_0}(Z)$
is the space of unparameterized paths in $Z$ passing through each choice of~$z_0\in Z$. 
Rather than a
path integral  approach,
we will apply canonical Hamiltonian methods in order to quantize this classical dynamical system.

Critical points of $S_\alpha$ are given by paths $\gamma$ whose tangent spaces lie in the kernel of $d\alpha$. 
A canonical choice of parameterization is that in which the tangent vector equals the {\it Reeb vector} $\rho$, which is uniquely defined by
$$
\alpha(\rho)=1\, ,\qquad 
{\mathcal L}_\rho \alpha =0\, .
$$
The Legendre transform between the Lagrangian $L$, determined by the contact form, and the corresponding Hamiltonian of the action $S_\alpha$  is highly degenerate. Indeed, the Lagrangian is maximally singular, since in  local coordinates $(q^1,\ldots, q^{2n+1})$ for some patch in  $Z$,
$$L(q^1,\ldots,q^{2n+1},\dot{q}^1,
\ldots,\dot{q}^{2n+1})=
\alpha_i(q)\, \dot{q}{\hh}^i\, ,$$
where the dot denotes the derivative with respect to any parameter along the path $\gamma$. Thus
$$\frac{\dd^2 L}{\dd{\dot{q}^i}\dd{\dot{q}^j}}=0\, ,
$$
and so the Hamiltonian formulation on the symplectic manifold $(T^*\!Z,d\la_\st)$ must be constrained to the $(2n+1)$-dimensional submanifold defined by $$\la_\st=\pi^*\alpha\, ,$$ where $\la_\st$ is the standard Liouville structure on the cotangent bundle $T^*Z$ and $\pi:T^*\!Z\lr Z$ is the canonical projection to the zero section.
In local coordinates $(p_i,q^i)$ for $T^*\!Z$, the above condition reads
$$
p_i=\alpha_i(q)\, .
$$

\subsection{Hamiltonian Formulation}\label{ssec:HamPersp} Let us continue with the quantization of the action functional $S_\alpha$, see \cite{HenneauxTeitelboim_QuantizationGauge} for details on quantization procedures, and the manuscript \cite{GrigorievLyakhovich_BRST} for a case more tailored to our context. The $(2n+1)$-dimensional  constraint submanifold discussed above,
$$\{\la_\st-\pi^*\alpha=0\}\sse (T^*\!Z,\la_\st),$$
is  the image of the contact manifold $(Z,\alpha)$ under the contact form $\alpha\in\Omega^1(Z)$, understood as a section of $T^*\!Z$. The constraint submanifold can be 
cut out as the zero locus of  a~$T^*\!Z$-valued function on~$T^*\!Z$, which we denote by ~$\phi$. 
Now, recall  the existence of a  vector bundle splitting $TZ=\xi\oplus {\rm span}(\rho)$; this corresponds to a reduction of the structure group $\GL(2n+1,\R)$ of $TZ$ to the subgroup $\operatorname{Sp}(2n)\times\mathbb R^+$. Hence, noting that the kernel of interior multiplication by the  Reeb vector field gives a subbundle of~$T^*\!Z$ isomorphic to the codistribution $\xi^* = T^*\! Z / \alpha$,  we canonically get a projection $\phi^o:T^*\!Z \to \xi^*$.
 In  local coordinates  this amounts to viewing $\phi$ as~$(2n+1)$ constraint functions 
$\phi_1,\ldots,\phi_{2n+1}\in C^\infty(T^*\!Z,\R)$,
with $\phi_i=p_i-\alpha_i(q)$,
whose Poisson brackets  satisfy $$\{\phi_i,\phi_j\}:=d\la_\st(X_{\phi_i},X_{\phi_j})=\pi^*d\alpha(X_{\phi_i},X_{\phi_j}),\quad 1\leq i,j\leq 2n+1\, ,$$
where $X_{\phi_i}$ denote the corresponding Hamiltonian vector fields. 

Since the Levi 2-form $\varphi:=d\alpha$ has maximal rank, 
these $(2n+1)$ constraints separate into a  first-class constraint $\phi_{2n+1}$ that generates the 
reparameterization gauge symmetry of the functional $S_\alpha$,
and $2n$ second-class constraints $\{\phi_1,\ldots,\phi_{2n}\}$ corresponding to the function $\phi^o$. 
Note that the Hamiltonian vector field of the first-class constraint~$\phi_{2n+1}$
obeys  $\pi_* X_{\phi_{2n+1}}=\rho$.

Typically, second class constraints are handled by replacing Poisson brackets by the Dirac bracket; the latter is constructed such that the constraints do not evolve (a useful review is~\cite{Bojowald}).
For our purposes, however, it is propitious to instead replace the symplectic manifold $T^*\!Z$ with a larger symplectic manifold~$P$, on which the second class constraints are promoted to mutually Poisson-commuting first class constraints; this is in line with ~\cite{Fradkin,Batalin}. These generate gauge invariances such that the dynamics on this extended phase space are gauge equivalent to the original system. In order to achieve that, we need to   construct a  suitable coisotropic submanifold~$N$ of~$P$. 

First, the extended phase space we consider is the Whitney vector bundle sum
$$P=T^*\!Z\oplus\xi^*\, ,$$
where $P$ has a symplectic structure 
$$
\Omega=\omega_{\rm st, T^*\!Z} + \omega_{\rm st, {\mathbb R}^{2n}}
$$
given by the sum of closed $2$-forms constructed as follows: the two forgetful maps that project out the first or second summand of the Whitney sum can be used to pull-back forms on the manifolds~$T^*\!Z$ and~$\xi^*$ to~$P$.
Then $\omega_{\rm st, \tiny T^*\!Z}$ is the pull-back of the standard symplectic 2-form $d\lambda_{\rm st}$ on $T^*\!Z$. The 2-form~$\omega_{\rm st, {\mathbb R}^{2n}}$ 
is a smooth choice of a closed, non-degenerate antisymmetric bilinear form making each fiber of~$\xi^*$ into a symplectic vector space.
This may be described by the data of a  {\it classical calibration map}~$\kappa$, which is defined as any linear map 
$$
\kappa: \Gamma(T\xi) \to \Gamma(\xi)\, ,
$$
such that $\kappa$ has rank $2n$ and is an isomorphism for vertical vectors $ \kappa: \Gamma(T_V\xi) \stackrel\cong\longrightarrow \Gamma(\xi)$ (so $T_V\xi$ is the subbundle of $T\xi$ of vectors tangent to the fibers of $\xi$), and such that
for any pair of sections~$u,v\in \Gamma(\xi)$, the 2-form $J\in \Omega^2 P$ defined, using the Levi bracket,  by
$$J(u,v)=(d\alpha)(\kappa(u),\kappa(v))=-\alpha([\kappa(u),\kappa(v)])$$
is closed.
Then $\omega_{\rm st, {\mathbb R}^{2n}}$ is 
 the pullback of $J$ (which may be viewed as a choice of symplectic form on each ${\mathbb R}^{2n}$ fiber of $\xi$) to $P$ along the forgetful map. 
 
\begin{remark}
Physically, one may imagine trying to locally model Reeb dynamics by evolution in  the trivial symplectic manifold ${\mathbb R}^{2n}$, and $P$ as gluing together these local models. The map $\kappa$ calibrates
generalized positions and momenta in ${\mathbb R}^{2n}$ to those in $\xi^*$. 
It may always be constructed locally by introducing frames for the distribution $\xi$, and therefore also whenever $Z$ is parallelizable. We leave aside a study of possible obstructions to the global existence of classical calibration maps, and note that in any case the existence of global measuring devices is an idealization that is seldom realizable for physical systems.\hfill$\blacksquare$
\end{remark}

To construct the required coisotropic submanifold $N$ of the symplectic manifold $(P,\Omega)$, we first introduce a 1-form $A$ on the (total space of the) manifold $\xi^*$ with the horizontal property that~$A$ annihilates vertical vectors in $T\xi^*$
({\it i.e.} vectors $v$ tangent to curves in $\pi_o^{-1}(z)$ for any $z\in Z$). Moreover we require that, restricted to the image of the zero section in $\xi^*$ (which is isomorphic to $Z$ itself),
the form $A$ then coincides with the contact form $\alpha$. Let us denote by $A_{\xi^*}$ the pullback of $A$ to $P$ under the second forgetful map.

Then we define~$N$ to be the submanifold
$$
N=
\big\{\la_{\st,T^*\!Z}
- A_{\xi^*}
=0
\big\}\sse (P,\Omega)\, ,
$$
which is the zero locus of a function $\Phi:P\to P$, {\it i.e.}   $N=\Phi^{-1}(Z)$. The submanifold $N$ is coisotropic when $A$ obeys a Cartan--Maurer equation $$\{ \la_{\st,T^*\!Z},  A_{\xi^*}\}_{_\Omega} -  \{   A_{\xi^*}, A_{\xi^*}\}_{_\Omega}=0\, ,$$ where $\{\pdot,\pdot\}_{_\Omega}$ denotes the Poisson brackets of $\Omega$. This takes a more familiar form upon choosing local coordinates $(q^i,p_i dq^i)$ for $T^*\!Z$ and $s_a$ for ${\mathbb R}^{2n}$.
Then, using the classical calibration map $\kappa$, points in~$\xi^*$ are labeled by $(q^i, s_a e_i{}^a(q) dq^i)$ where the $2n$ one-forms $e_i{}^a dq^i$ are in the codistribution $\xi^*$. They determine the map $\kappa$  and therefore have maximal rank. Then the 1-form $A_{\xi^*}$ is given by $A_i(q,s)dq^i$ which must obey $A_i(s)=\alpha_i(q)+\omega_i(q,s)$. Thus, the first class constraint function $\Phi$ is given by 
$$
\Phi = dq^i \big(p_i - \alpha_i(q) - \omega_i(q,s)\big)\, ,
$$
while the non-zero Poisson brackets are $\{p_i,q^j\}_{_\Omega}=\delta^i_j$ and $\{s_a,s_b\}=J_{ab}$.
Denoting the horizontal exterior derivative $d:= dq^i\frac{\partial}{\partial q^i}$ and $A:=A_i(q,s)dq^i$, the Cartan--Maurer equation expressing
that the
first class constraints commute, is given by 
$$
dA - \frac12\big\{ A,A\big\}_{_\Omega}=0\, .
$$
The $s$-independent part of the above equation stipulates that $d\alpha = \frac12 \{\omega,\omega\}\big|_{s=0}$. This can be solved by requiring that the leading term in $\omega$ is
determined by the calibration map and that, in addition, the 1-forms $e_i{}^a:=\partial\omega_i/\partial s_a  \big|_{s=0}$, where $1\leq a\leq 2n$, obey
$$
\frac 12\, 
 e^a\wedge e^b J_{ab}= d\alpha=:\varphi\, .
$$  
In the case that the {\it soldering forms} $e^a$ are closed, the Cartan--Maurer equation is solved by $\omega=s_a e^a$ and the first class constraint submanifold $N$ is given by 
the intersection of the graph of $\phi^o$ in $P$ and the submanifold defined by $\lambda_{\st}(\rho)=1$ (pulled back to $P$ by the first forgetful map).

In conclusion, the dynamics on the extended phase space $(P,\Omega)$ is governed by the extended action
$$\wt S_\alpha(\Gamma):=\int_{\Gamma}\Big[\lambda_{\st,T^*\!Z}+A_{\xi^*}\Big],\quad \Gamma \in L_{p_0}(P)\, .$$
The gauge transformations generated by the first class constraints  and the gauge fixings required
such that the dynamics of the extended action $\widetilde S_\alpha$ reproduce those of $S_\alpha$ have been  computed in~\cite{HW}. This extended action is the starting point for Hamiltonian BRST quantization.

\subsection{BRST Charge}\label{ssec:BRSTComplex} The homological quantization of the extended action $\wt S_\alpha$ proceeds in two  steps:
\begin{itemize}
	\item[(i)] A construction of the space of classical solutions modulo gauge symmetries, or {\it reduced phase space}, in terms of derived geometry.
	\item[(ii)] A canonical quantization.
\end{itemize}

For the first step, 
one constructs  a  Hamiltonian BRST complex giving the
homological resolution of  the  reduced phase space. This construction takes place on a suitable jet bundle, in terms of the derived intersection of the equations of motion and the derived quotient by the gauge symmetries. The resulting complex  combines the derived critical locus of $\widetilde S_\alpha$, modeled, in negative degree, on the BV-complex of polyvector fields  and, in positive degree,  on the (twisted) Chevalley-Eilenberg complex of the Lie algebra generated by the gauge symmetries---for systems where the constraint algebra is first class but not Lie, see ~\cite{Letterbomb}.

In the case of the contact action $S_\alpha$, the reformulation in Subsection \ref{ssec:HamPersp} of the constrained Hamiltonian system in terms of $(2n+1)$ first-class constraints yields an abelian Lie algebra represented on the algebra of smooth functions $C^\infty(P)$. The  submanifold $N$ generates  a coisotropic ideal $I$. The Koszul resolution of the quotient $C^\infty(P)/I$ is performed by  introducing ghost number one variables $\{c^i\, |\, 1\leq i\leq 2n+1\}$ that give a  basis for the graded vector space $V$ and  its exterior algebra in terms of products of ghosts. Ghost momenta $\{b_i\, |\, 1\leq i\leq 2n+1\}$---assigned ghost number minus one---yield a basis for the dual  vector space~$\check{V}$ and its exterior algebra. Together, the ghosts and their momenta  are used to construct a  module graded by ghost number. This is   a dg-algebra with a ghost number one differential given by the nilpotent  BRST charge $\Omega$. In our case, the BRST charge reads
$$\Omega_{BRST}=
c^i\Phi_{i}.$$
This acts on functions of the {\it BRST  extended phase space}, which depend on (coordinates in) the extended phase space $(P,\Omega)$, the ghosts and ghost momenta.
At the level of sheaves, this gives a homological resolution of the reduced phase space. Since we are primarily interested in proceeding to a quantization, a detailed analysis of the space of sections is not necessary here and we proceed with the second step (ii).

The second step in the homological quantization of the extended action $\wt S_\alpha$ is  a canonical quantization of the
BRST  extended phase space. For this, we   identify the Grassmann algebra of the  ghost variables with the exterior bundle of differential forms on~$Z$ (see~\cite{Remorse}). This allows us to view the quantized BRST charge ~$\widehat \Omega:=\widehat\Omega_{BRST}$ as a differential 1-form and, because it is nilpotent $\widehat \Omega^2=0$, 
 as a flat connection. 
For the quantization of the  
 second summand $\xi^*$  in the phase space $T^*\!Z\oplus\xi^*$, one may choose  classical fiber coordinates $(s_1,\ldots, s_{2n})$.
  Then a  model to keep in mind is   a fiberwise choice of polarization in~$\xi^*$
 with vertical coordinates $(s_1,\ldots , s_n)$ and the Weyl representation of the Heisenberg algebra 
$$\mathfrak{heis}_n={\rm span}\{s_1,\ldots,s_n,\tfrac\hbar i\dd_{s_1},\ldots,\tfrac \hbar i\dd_{s_n}, \hbar\}\, .$$ 
To be precise, the bundle of frames on the  distribution $\xi$ is a principal ${\rm Sp}(2n)$-bundle ${\mathcal M}Z\to Z$ over the contact manifold ~$Z$. Then we pick an (infinite dimensional)  Hilbert space  ${\mathcal H}$ with a unitary, projective representation of~$ {\rm Sp}(2n)$. For the case ${\mathcal H}=L^2({\mathbb R}^n)$, the latter is provided by the metaplectic representation, which we discuss in detail in Section~\ref{sec:metaplectic} \!.
As discussed in the introduction,
this data determines the  Hilbert bundle ${\mathcal M}Z\times_{{\rm Sp}(2n)}{\mathcal H}=:{\mathcal H}Z\to Z$
associated to ${\mathcal M}Z$. In addition, the bundle of (complete orthonormal) frames on ${\mathcal H}Z$ is a principal $U({\mathcal H})$-bundle~${\mathscr M}$. Again, we have an associated Hilbert bundle $\SH:={\mathscr M}\times_{U({\mathcal H})} {\mathcal H}$. 

By the Stone--von Neumann theorem, the Hilbert space ${\mathcal H}$ is a representation of the  Heisenberg algebra, defined up to unitary equivalence. Hence, the bundle $\SH$ is
 equipped with the symplectic analog $s$ of Clifford multiplication, that we refer to as the {\it Heisenberg map}
$$
s:\Gamma(\xi\otimes {\SH})\to \Gamma(\SH) \, ,
$$
where
\begin{equation}\label{sdoesthis}
s(u)\circ s(v)-s(v)\circ s(u)=-i\varphi(u,v)  {\rm Id}\, ,\quad \forall u,v \in \Gamma(\xi)\, .
\end{equation}
Since the Hilbert bundle
is endowed with a hermitean 
vector bundle metric whose restriction to $\pi^{-1}(z)\times \pi^{-1}(z)$ is an inner product on the Hilbert space ${\mathcal H}$, we have a sesquilinear map on sections  $(\Psi,\Phi)\in\Gamma({\mathscr H})\times \Gamma({\mathscr H})\to C^\infty_{\mathbb C} Z$
given by
$
\langle \Psi(z),\Phi(z)\rangle
$.
This defines an adjoint $\dagger$ on $\operatorname{End}\big(\Gamma({\mathscr H})\big)$ such that $s^\dagger (u)= s(u)$. Locally, trivializing the bundle $\SH$, the map $s$ gives a 1-form $\hat \kappa$ taking values in endomorphisms of~${\mathcal H}$, which we call a {\it quantum  calibration map}, since it may be viewed as calibrating classical observables on~$Z$  with their quantum Hilbert space counterparts. Fiberwise, the map $\hat \kappa$
gives a representation of the Heisenberg algebra where we have made rescalings such that the central element~$\hbar$
acts by the identity.

Now, given a (local) choice of soldering forms $e^a$ subject to 
$$\frac12 J_{ab} e^a\wedge e^b = d\alpha\, ,$$
the quantum calibration map may be locally expressed as
$$
\hat \kappa = 
e^a \hat s_a\, ,
$$
where the hermitean operators~$\hat s^a=\hat s^a{}^\dagger$  on ${\mathcal H}$  obey $$[\hat s_a,\hat s_b]=i\cdot j_{ab}\, ,$$
acting on the Hilbert space~${\mathcal H}$.
These may be viewed as the quantization of the fiber coordinates $s^a$.

\begin{remark} Ontologically,  any quantum system obtained from a classical one via some quantization procedure
must enjoy a flow with respect to $\hbar$. This is indeed the case in our context, as we can define the operator
\begin{equation}
{\sf gr}=2 \hbar \frac{\partial }{\partial \hbar} 
\, ,
\end{equation}
that measures this flow. Then we may declare that an operator $\widehat O^\hbar\in\End\big(\Gamma(\SH)\big)$
has {\it grade $k$} if
$$
\big[{\sf gr},\widehat O^\hbar\hh\big]=k\hh  \widehat O^\hbar\, .
$$
More generally, when $\widehat O^\hbar$ has a Laurent series about $\surd \hbar = 0$, we call the coefficient $O^{(k)}$ of
$\hbar^{k/2}$ the {\it grade~$k$ part} of $\widehat O^\hbar$.\hfill$\blacksquare$
\end{remark}

\medskip

Given the data $(\SH, \hat \kappa)$, the quantization of $(Z,\alpha)$ is completed by specifying a 
 quantum BRST charge $\widehat\Omega$. This is  an ${\rm End}({\mathcal H})$-valued 1-form differential $d^{\hat A}:=d+\hat A$ which in our case reads 
$$\frac{i}\hbar\, \widehat\Omega=\frac{\alpha\, {\rm Id}}{i\hbar}+\frac{\hat \kappa}{i\surd\hbar}+d+\hat \omega,$$
where  the ${\rm End}({\mathcal H})$-valued 1-form $\hat \omega$ must be chosen both such that
$\widehat \Omega$ is nilpotent and regular in~$\surd \hbar$ around $\hbar =0$ (so as not to disturb the leading classical 
$\alpha/(i\hbar)$ behavior).
In order to have a global quantization, we also ask that there exist a one-parameter family of quantum dynamical systems---{\it i.e.} flat connections~$\nabla^\hbar$ on~$\SH$---for each of which the $\widehat \Omega$ displayed above (locally) defines a connection form.
We may then  view~$\nabla^\hbar$ as a degree one endomorphism of the {\it BRST Hilbert space} ${\mathcal H}_{\rm BRST}:=\Gamma({\mathscr H}\otimes \wedge Z)$.

In conclusion, this defines a dynamical quantization of the Reeb dynamics of~$\alpha$ in the sense of Definition~\ref{QMdef} \!. 
When $\hat \omega$ is given as a formal  expansion order by order in  $\surd\hbar$, we call~$\nabla^\hbar$ a {\it formal quantum connection}. We also use the moniker {\it local quantum connection} when a solution for~$\nabla^\hbar$ is only given on a patch of~$Z$.

\begin{remark}
When this is clear from context, we drop the superscript $\hbar$ on the quantum connections~$\nabla^\hbar$. Note also that in~\cite{HLW}, a quantum Darboux theorem is established, stating that, locally, any two formal quantum connections are equivalent.\hfill$\blacksquare$
\end{remark}

\subsection{Illustrative Examples and Computations}\label{ssec:Examples} 

In this subsection, we provide detailed examples illustrating the construction described in Subsection \ref{ssec:BRSTComplex} above, in Examples \ref{ex:Contact1} and \ref{proto}, and also set up the an archetypal example for the quantization of  contact geometries  in Example \ref{ex:Contact2} and formalize this in Definition \ref{strictequi} below. First, starting from the data $(\alpha, \hat \kappa)$, the flatness condition
$$
\nabla^2 = 0\, ,
$$
can be always be solved locally and  formally; a local solution is given in the following example, while formality is addressed in Example~\ref{proto} below.

\begin{example}\label{ex:Contact1}
Let us choose local Darboux coordinates in which we have
$$\alpha=
p_{{\bar a}} d x^{{\bar a}} - dt.$$ 
A local solution with $\hat \omega=0$ and $\psi( \sigma^1,\ldots, \sigma^n)\in{\mathcal H}=L^2({\mathbb R}^n)$ is given by
 \begin{equation}\label{darb}i\hbar\,  d^{\hat A} = 
 dt\Big(i\hbar \frac\partial{\partial t}-1\Big)+
dx^{{\bar a}}\Big(p_{{\bar a}} 
+i\hbar
\big(\frac{\partial}{\partial \sigma^{{\bar a}}} +\frac{\partial}{\partial x^{{\bar a}}}\big)\Big)  +  dp_{{\bar a}}\big(\sigma^{{\bar a}} +i\hbar\frac\partial{\partial p_{{\bar a}}}\big)
\, .\end{equation} 
Defining rescaled variables $\sigma^{\bar a}=:\surd \hbar \hh s^{\bar a}$, the 
quantum calibration map is given by
$$
\hat \kappa = 
dp_{{\bar a}} s^{{\bar a}} + i dx^{{\bar a}} \frac{\partial}{\partial s^{{\bar a}}}
= 
e^a \hat s_a. 
$$
It is closed because $e^a=(dp_{\bar a},dx^{\bar a})$ obey $de^a = 0$. Note that by twisting with an ${\rm Sp}(2n)$ gauge transformation $U=\exp(\frac{1}{2i\hbar} u_{ab}\hat s^a \hat s^b)$  taking values in the projective metaplectic representation, it is possible to $($locally$)$  reach a
 quantum connection form $U\circ d^{\hat A} \circ U^{-1}$ with any choice of frames~$e^a$. \hfill$\blacksquare$\end{example}

The following example demonstrates how
to solve for a formal quantum connection given the data of a contact form and symplectic frames for the distribution $\xi$.

\begin{example}
\label{proto}
Let $(Z,\alpha)$
be a strict contact manifold and
$\{e^a\}$ a set of local frames for the codistribution $\xi^*$ obeying
$$
\frac12 J_{ab} e^a\wedge e^b = d\alpha\, .
$$
Note that when 
$Z$ is parallelizable
these frames exist globally and this example extends to global formal quantizations.
Now, for a Weyl ordered quantization in which functions are
replaced by Weyl ordered operators, we make the following
formal ansatz
for $\hat \omega$:
$$\hat \omega=-i\sum_
{k\geq 0} \hbar^{\frac k2}\, 
\omega^{a_1\cdots a_k} \hat s_{a_1}\cdots \hat s_{a_k}\, ,
$$ 
where the 1-forms $\omega^{a_1\ldots a_k}$ are totally symmetric in the labels $a_1,\ldots, a_k$. 
Formal flatness of $\nabla^\hbar$ then imposes a system of equations, which can be obtained by working order by order in the grading:
\begin{eqnarray*}
&d\alpha - \tfrac 12 e^a \wedge e_a=0\, ,&\\[1mm]
&d^\omega e^a \,=0\,  ,&
\\[1mm]
&F^{ab} + \omega^{abc}\wedge e_c = 0\, , &\\
&d^\omega\omega^{abc} + \omega^{abcd}\wedge e_d = 0 \, ,&\\
&d^\omega\omega^{abcd} 
+
\omega^{(ab|e}\wedge \omega_e{}^{cd)}
+ \omega^{abcde}\wedge e_e = 0 \,, &\\[-1mm]
&\vdots&
\end{eqnarray*}
In the above we have denoted  $d^\omega X^{ab...}:=dX^{ab...}+\omega^a{}_e X^{ab...}+
\omega^b{}_e X^{ae...}+
\cdots$ and $F^{ab}:=d\omega^{ab}+\omega^a{}_c \wedge\omega^{cb}$, and indices are raised and lowered using $J_{ab}$ with the convention $J_{ab}X^b=: X_a$. The first relation in the above system is the one used to define the frames $e^a$; this is thus solved. Now, we aim at algebraically solving, order by order,  the remaining equations for the 1-forms $\omega^{ab},\omega^{abc},\ldots$.
In order to solve the second equation, {\it i.e.} solve $de^a + \omega^{ab} \wedge e_b = 0$, we note that the relation
$\frac12 j_{ab}e^a\wedge e^b = d\alpha$ implies that $e_a\wedge de^a = 0$.
Thus, we may expand
$$
\quad-de^a =  \tau^{ab} \alpha\wedge e_b + \tau^{abc} e_b \wedge e_c\, ,
$$
where the functions $\tau^{ab}$
and $\tau^{abc}$ obey
$\tau^{ab}=\tau^{ba}$ and $\tau^{abc}=\tau^{bac}$ and $\tau^{abc}=-\tau^{acb}$.
Similarly, we may also expand
$\omega^{ab}=w^{ab} \alpha + w^{abc} e_c$ so that 
$$
\omega^{ab}\wedge e_b = w^{ab} \alpha\wedge e_b + w^{abc} e_b\wedge e_c
\, ,$$ 
where the functions $w^{ab}$ and $w^{abc}$ obey
$w^{ab}=w^{ba}$, 
$w^{abc}=w^{bac}$.
 Thus, we need to solve the system of equations
 $$
 w^{ab}=\tau^{ab} \, ,\qquad w^{a[bc]}=\tau^{abc}\, ,
 $$
 which
 always has a solution for $w^{ab}$ and $w^{abc}$. 
 Schematically, in a Young tableaux notation, we  have
$$
\omega^{ab}=
\ytableausetup
{mathmode,centertableaux,boxsize=.7em}
\ydiagram{2}+
\ydiagram{2}\otimes\ydiagram{1}
\ni 
\ydiagram{2}+
\ydiagram{2,1}=-de^a\,.
$$
This solves the second equation. In order to address the next order, {\it i.e.} the third equation, we note that 
$
d^\omega e^a = 0$ implies $F^{ab}\wedge e_b=0$. Using the same diagrammatic scheme this says that the 2-form~$F^{ab}$  decomposes as
$$
F^{ab}=
\ytableausetup
{mathmode,centertableaux,boxsize=.7em}
\ydiagram{3}+
\ydiagram{3,1}\, ,
$$
and note that for the 1-form $\omega^{abc}$ we have
$$
\omega^{abc}= 
\ytableausetup
{mathmode,centertableaux,boxsize=.7em}
\ydiagram{3}+
\ydiagram{3}\otimes\ydiagram{1}
\ni 
\ydiagram{3}+
\ydiagram{3,1}\, .
$$
In general, analogs of the  identities $e_a de^a = 0 = e_a \wedge F^{ab}$ continue to hold at every order; this is a  consequence of the Bianchi identity $[\nabla,\nabla^2]=0$ applied to the lowest order non-vanishing term in the curvature~$\nabla^2$. For instance, at the next order
one has
$
e_a\wedge d^\omega \omega^{abc} = 0\, ,
$
and thus the decompositions
$$
\omega^{abcd} = 
\ytableausetup
{mathmode,centertableaux,boxsize=.7em}
\ydiagram{4}+
\ydiagram{4}\otimes\ydiagram{1}
\ni 
\ydiagram{4}+
\ydiagram{4,1}=-d^\omega \omega^{abc}\,.
$$
This pattern continues to all higher orders and gives a formal solution for $\nabla$.
In Example~\ref{Hamsys} below, we will apply this formal quantization to Hamiltonian systems.\hfill$\blacksquare$\end{example}

An important consequence of the flatness condition is that a quantum connection  induces a connection on the distribution $\xi$, denoted $\nabla^\xi$, as follows.
We first expand the flatness condition for $\nabla$ in a Laurent series about $\surd \hbar = 0$, and at grade $-2$ one finds that the grade $-1$ part of $\nabla$ obeys the Heisenberg map Relation~\nn{sdoesthis}, so we have $\nabla^{(-1)}=s$.
Moreover, at grade $-1$ the flatness condition gives the equation
\begin{equation}\label{employ-me}
\nabla^{(0)} \circ s+ s \circ  \nabla^{(0)}=0\, . 
\end{equation}
Thus, if $u \in \Gamma(\xi)$, by linearity of $\nabla$ we may define $\nabla^\xi u$ by
\begin{equation}\label{extract}
[\nabla^{(0)},s(u)]=s(\nabla^\xi u)\, .
\end{equation}
This is well-defined because the maximal rank condition on the Levi 2-form $\varphi$ ensures that $s$ is injective.  
Also, the connection $\nabla^\xi$ obeys an analog of a  torsion-free condition; namely, for~$u,v\in \Gamma(\xi)$,
\begin{equation}\label{torsionlike}
\nabla^\xi_u v - \nabla^\xi_v u = [u,v] - \rho \alpha([u,v])=:{\mathcal L}^\xi_u v\in \Gamma(\xi)\, .
\end{equation}
To see that the above holds, one computes $s$ of the left hand side of the above display and then employs Equation~\nn{employ-me}.
In the following Example \ref{ex:Contact2}, the grade zero operator $\nabla^{(0)}$ is simply the exterior derivative on ${\mathbb R}^3$, which induces a flat connection on the distribution for ${\mathbb R}^3$ equipped with its standard contact form.

\begin{example}\label{ex:Contact2}
Let $(x,y,z)\in\R^3$ be Cartesian coordinates and $(r,\theta,z)\in\R^3$ cylindrical coordinates such that $(r,\theta)$ are polar coordinates for the $(x,y)$-plane. Consider the strict contact structure 
$({\mathbb R}^3,\alpha)$ with $\alpha:=\frac12 r^2 d\theta-dz$ and the Hilbert space ${\mathcal H}=L^2({\mathbb R})$, with wavefunctions $\psi(s)\in{\mathcal H}$.  Then
a quantum connection form is 
$$
d^{\hat A} = \frac{\alpha}{i\hbar} + \frac{\hat \kappa}{i\surd\hbar} + d\, ,
$$
where the calibration is given by $\hat \kappa = 
s_- dr + s_+ rd\theta\, , $ with
$$
s_-:=i\cos\theta  \, \frac{\partial}{\partial s}+ 
\sin\theta \, s\, ,
\qquad
s_+:=i \sin\theta   \,  \frac{\partial}{\partial s} -
\cos\theta \, s\, .
$$
The vector field $u=u^-  \frac{\partial}{\partial r}
+\frac{u^+}r \big(
\frac{\partial}{\partial \theta}+\frac {r^2}2\frac{\partial}{\partial z}
\big)
$ is in the distribution $\xi$, for $u^\pm \in C^\infty({\mathbb R}^3)$.
Since we have that $d s_- =  -d\theta\hh s_+$ and $d s_+=   d\theta\, s_-$, it follows that
$$
d \hat \kappa(u) = \hat \kappa (\nabla^\xi u)=
\left[d+
\begin{pmatrix}
0&  d\theta  
\\
 -d\theta &  0
\end{pmatrix} 
\right]\begin{pmatrix}
u^+\\ u^-
\end{pmatrix}\, ,
$$
where we are employing a matrix notation in which   $\hat \kappa(u) 
=u^+ s_++u^-s_-
=:  \begin{pmatrix} u^+\\ u^-\end{pmatrix}$. For later use, we observe that the identity operator on the distribution can be obtained by acting with the connection~$\nabla^\xi$ on the Euler vector field $r\frac\partial{\partial r}$; {\it i.e.}  
$$
\nabla^\xi \Big(r\frac\partial{\partial r}\Big)=\operatorname{Id}:\Gamma(\xi)\to \Gamma(\xi)\, .
$$
\hfill$\blacksquare$\end{example}

Let us now develop an example which illustrates the interplay between  the quantization of a given strict contact manifold $(Z,\alpha)$ and an {\it ambient} strict contact space $(M,A)$. This theme will be further developed in 
Section \ref{sec:Contract}.

\begin{example}\label{getmaxy}
Let $(Z,\alpha)$ be a strict contact structure and $(\mu,\theta)$ be Cartesian coordinates for $\R^2$. Consider the $(\dim Z+2)$-dimensional  strict contact structure~$(M,A)$, where $M:={\mathbb R}^2\times Z$ 
and the contact 1-form is given by
\begin{equation}\label{hereisA}
A:=e^{2\mu} \alpha + d\theta,\mbox{ and thus }
dA = e^{2\mu}(d\alpha+2d\mu\wedge \alpha )\, .
\end{equation}
This is the contactization of the symplectization of $(Z,\alpha)$, as $e^{2\mu}\alpha$ is a Liouville form for the symplectization $\mathcal{C}\to Z$, under the identification $\mathcal{C}\cong\R_\mu\times Z$.

In general, the contact 1-form $A\in\Omega^1(M)$ is unchanged under the replacement $\mu\to \mu - \varpi$ and $\alpha \to e^{2\varpi}\alpha$. In particular, there exists a contactomorphism generated by the vector field $X:=\frac{\partial }{\partial \mu}+2 \theta \frac{\partial }{\partial \theta}$, which satisfies ${\mathcal L}_X  A =2A$. The corresponding flow, {\it i.e.} $\R^+$-action $x$, is given by $(\mu,\theta,z)\stackrel{x_\lambda}\longmapsto 
(\mu+\log \lambda, \lambda^2 \theta,z)$.

In order to quantize the strict contact manifold $(M,A)$, we introduce frames 
\begin{equation}\label{whoframedrogerrabbit}E^A
=(2e^{2\mu}\alpha,e^\mu e^a,d\mu )\Rightarrow dA = 
\frac12 J_{AB} E^A \wedge E^B
\, ,
\end{equation}
where the index runs as $A=(+,a,-)$, $1\leq a\leq 2n+1$, and 
$$
\big(J_{AB}\big)=
\begin{pmatrix}
0&0&-1\\
0& j_{ab}&0\\1&0&0
\end{pmatrix}\, .
$$
Note that ${\mathcal L}_X E^A = \operatorname{diag}(2,1,0)^A{}_B E^B$. Now, consider the $($connection$)$ matrix $\Omega^A_B$ of 1-forms given by
$$
\Omega^A{}_B := \left(\begin{array}{ccc}   -d\mu & e^{\mu}e_b &  2e^{2\mu}\alpha\\  -e^{-\mu}P^a  & \omega^a{}_b & e^\mu e^a\\ \frac12e^{-2\mu}Q & e^{-\mu}P_b & d\mu  \end{array} \right)\qquad (\mbox{this satisfies }\Omega_{AB} = J_{AC} \Omega^C{}_B= \Omega_{BA}),
$$
where  $Q$, $P^a\!$,  and  $\omega^{ab}$ obey the algebraic system of equations
\begin{equation}\label{flatter}
 e_a\wedge P^a = \alpha\wedge Q \, ,\qquad  d^\omega  e^a  =- 2\alpha\wedge P^a\, .\end{equation}
Here we define 
$d^\omega v^a := d v^a + \omega^a{}_b v^b$, 
$v_a =j_{ab} v^b$, and require $\omega^{ab}=\omega^{ba}$. We also denote $X^A:=E^A(X)=(0,0,1)$. These relations were chosen such that the connection matrix $\Omega^A{}_B$ obeys 
\begin{eqnarray*}
d^\Omega X^A &:=& \, d X^A + \Omega^A{}_B X^A = E^A\, ,\\[1mm]
 F^{A}{}_B X^B\! &:=&  \!(d \Omega^A{}_B + \Omega^A{}_C\wedge\Omega^C{}_B)  X^B=0\, ,
\end{eqnarray*}
and in turn,
$$
d^\Omega E^A := d E^A + \Omega^A{}_B E^B=0\, .
$$
Altogether, the above conditions imply that the connection 
\begin{equation}\label{nab0}
d^{\hat A^0}=\frac{A}{i\hbar} + \frac{E^A \hat S_A}{i\surd\hbar} + d + \frac1{2i} \, \Omega^{AB} \hat S_A \hat S_B
\end{equation}
 on the Hilbert bundle ${\SH}\to M$ obeys
 $$
 (d^{\hat A^0})^2 = \frac{1}{2i} \, F^{AB} \hat S_A \hat S_B\, ,
 $$
 where $\hat S_A\in \operatorname{End}{\SH}$ satisfy $\hat S_A\circ \hat S_B - \hat S_B \hat S_A =i J_{AB}\operatorname{Id}$. This connection may have curvature and thus $d^{\hat A^0}$ is not, in general, a quantum connection form. In order  to construct a formal quantum connection when the curvature $F^{AB}\neq 0$, higher order terms are required to achieve flatness. The properties of these higher order terms are interesting;
 observe that
 \begin{equation*}\label{equi}
 {\mathcal L}_X A = 2A\, ,\qquad 
  {\mathcal L}_X (E^A \otimes E_A)  = 2 E^A\otimes E_A\, ,  \qquad
 {\mathcal L}_X (\Omega^{AB} E_A\otimes E_B) = 2 \Omega^{AB} E_A \otimes E_B\, .
 \end{equation*}
 Moreover, the modified grading operator $\widetilde {\sf gr}={\sf gr}-\frac 1{i} \hat S_+\hat S_-$ obeys
 $$
 \widetilde {\sf gr}\circ \hat S_A - \hat S_A\circ \widetilde {\sf gr}=\hat S_B \operatorname{diag}(-1,0,1)^B{}_A\, .
 $$
 Hence we conclude that ${\sf Gr}:={\mathcal L}_X+\widetilde{\sf gr}$ obeys the equation
$$
{\sf Gr}\circ d^{\hat A^0} - d^{\hat A^0}\circ {\sf Gr}=0\, .
$$
Thus, one can  search for a formal quantum connection $d^{\hat A} = d^{\hat A^0}+\widehat \Omega^1$ that commutes with ${\sf Gr}$.
A~useful ansatz for such  solutions is given by
\begin{equation}\label{notme}
 \widehat \Omega^1=
 \sum
 _{
 \ell\geq 1
}
\frac{\hbar^{\frac\ell 2}}{i\hh k!}\, 
\Omega^{A_1\ldots A_\ell} \hat S_{A_1}\cdots \hat S_{A_\ell}
=-\big(\widehat \Omega^1\big)^\dagger
\, ,\qquad
({\mathcal L}_X -2) (\Omega^{A_1\ldots A_k} E_{A_1}\otimes\cdots \otimes E_{A_\ell}) = 0\, .\end{equation}
The tensors $\Omega^{A_1\ldots A_\ell} E_{A_1}\otimes\cdots \otimes E_{A_\ell}$  have definite homogeneity and 
take values in~$\Gamma(\otimes^\ell\Xi^*)$. In the language of the upcoming Section~\ref{sec:Contract}, these are contractors.\hfill$\blacksquare$\end{example}

Quantum connections associated to contact forms 
in the kernel of
$
{\mathcal L}_X - 2\, ,
$
where $X$ is a vector field generating an $\R^+$-action, will play a distinguished {\it r\^ole} in the study of quantized contact structures presented in Section~\ref{sec:Quant}. In preparation for that, we make the following definition.

\begin{definition}\label{strictequi}
Let $(M,A,x)$ be a strict contact structure with an $\R^+$-action $x$ generated by a vector field $X$ that obeys ${\mathcal L}_X A= 2A$. A quantization $\nabla$ of the Reeb dynamics of $A$ is said to be {\it $\R^+$-equivariant} if, acting on sections of the Hilbert bundle over $M$, the relation 
\begin{equation}\label{eqgr}
[\nabla^{(0)}_X+{\sf gr},\nabla_U]=0\,
\end{equation}
holds for all homogeneous vector fields $U\in \Gamma(TM)\cap \ker {\mathcal L}_X$.\hfill$\blacksquare$
\end{definition}

\noindent
Note that  Definition \ref{strictequi} formalizes various features of Example~\ref{getmaxy}
above. At leading order in the grading, the condition above imposes the homogeneity condition ${\mathcal L}_X A = 2A$. At grade $-1$ one finds the condition $d_U^\Omega X^A=E^A(U)$, since $
[d^{\Omega}_X, E^A(U) \hat S_A]
=(d^\Omega_X E^A(U)) \hat S_A =
(\nabla^\Omega_U E^A(X))\hat S_A
$, where we used Equation~\nn{torsionlike} and $[X,U]=0$. At grade zero, we learn that $\iota_X$ acting on the curvature of~$\nabla^\Omega$ vanishes. This condition will be commensurate with 
those necessary for $\nabla^\Omega$ to be an ambient tractor connection, as will be discussed in Section~\ref{ssec:ambientconnection}. At higher grades, the definition imposes the conditions listed in the last display of the example. A second detailed example of an $\R^+$-equivariant quantization is given in Section~\ref{S3}. Along similar lines to the preceding discussion, it is also interesting to study  contact forms in the kernel of ${\mathcal L}_u$ for vector fields~$u$ associated to conserved charges. This is the content of the next Subsection~\ref{symm}.

\subsection{Symmetries and Charges}\label{symm}

Let $(Z,\alpha)$ be a strict contact manifold and $\rho$ its associated Reeb vector field. For every vector field $u\in \Gamma(TZ)$ such that
$$
{\mathcal L}_u \alpha = d\beta\, ,
$$
for some $\beta\in C^\infty Z$,
the quantity
$
Q_u:=\alpha(u)-\beta
$
is a {\it conserved charge}, meaning that it is preserved by Reeb dynamics:
$$
{\mathcal L}_\rho Q_u = 0\, .
$$
Let us call a vector field $u$ 
such that ${\mathcal L}_u \alpha$ is exact a {\it strict contact symmetry}.
In fact, a converse of the above statement holds; the following is the strict contact analog of the Noether theorem:
\begin{proposition}
For every strict contact symmetry $u$, there exists a corresponding conserved charge~$Q_u$, and conversely for every 
conserved charge~$Q$, there exists a strict contact symmetry.
\end{proposition}

\begin{proof}
Given a strict symmetry $u$,   the function
$$Q_u=\alpha(u) -\beta\, ,$$
where $d\beta = {\mathcal L}_u \alpha$, is a conserved charge because
$$
{\mathcal L}_\rho Q_u = 
(\iota_\rho \circ d \circ \iota_u)\,  \alpha - \iota_\rho d\beta = \iota_\rho \, \big(
{\mathcal L}_u \alpha + \iota_u \varphi
-d\beta
\big) = 0\, .
$$ 
The second step used Cartan's formula, while the last used $\iota_\rho \varphi = 0$.

For the converse statement, 
given $Q$ that is conserved, 
take $u$ to be any solution of
$$
\iota_u \varphi = dQ\, .
$$
Note that $d Q$ is in  the image of the map $u\mapsto \iota_u \varphi$ because $\varphi$ has maximal rank and $\iota_\rho d Q = {\mathcal L}_\rho Q=0$. 
We remark that $u$ is unique up to the addition of a term $f \rho$ for any $f\in C^\infty Z$. To complete the proof, note that
$$
{\mathcal L}_u \alpha=
\iota_u \varphi + d( \alpha(u) )
=d(Q+\alpha(u))\, .
$$ 
\end{proof}

An important question in the study of quantizations of dynamical systems is how classical conserved charges are promoted to quantum operators that, in a suitable sense,  commute with quantum dynamics. In our context, this means searching for a section $\widehat Q_u\in \Gamma({\mathscr M})$ of the adjoint (or principal $U({\mathcal H})$) bundle ${\mathscr M}$,
that is parallel with respect to $\nabla$, {\it i.e.} 
$$
\nabla\circ \widehat Q_u = \widehat Q_u\circ \nabla\, ,
$$
and is obtained as the image $\widehat Q_u = \widehat \phi(Q_u)$ of a mapping on functions in the kernel of ${\mathcal L}_\rho$.
Such a map is called a {\it quantization map}. This is a generalization of the quantization of symplectic manifolds developed by Fedosov~\cite{Fed}, where one studies parallel sections of a Weyl bundle over a symplectic base with respect to a suitable connection. 
Indeed, as shown in~\cite{HLW}, given $f\in \operatorname{ker}{\mathcal L}_\rho\subset C^\infty Z[[\surd\hbar]]$, and a formal quantum connection~$\nabla$, it is possible to assign a  unique section $\widehat \phi(f)$ of $\hbar^{-2}\End({\mathscr H})[[\hbar]]$ that obeys 
$$
\nabla\circ \widehat \phi(f) = \widehat \phi(f) \circ \nabla\,, \quad {\mathcal L}_\rho \widehat \phi(f)= 0\,  .
$$
The {\it quantization map} $\widehat \phi$ is an isomorphism~\cite{HLW}, and defines an associative star product on functions~$f,g\in \operatorname{ker}{\mathcal L}_\rho$, according to 
$$
f\star g = \hat \phi^{-1} (\hat \phi(f)\circ \hat \phi(g))\in \ker {\mathcal L}_\rho\, . 
$$

In many situations, the charge $\widehat Q_u$ associated to a symmetry $u$ is easy to compute. Firstly, note that for any vector field $v$ on $Z$, the pair of operators on the BRST Hilbert space given by $\nabla$ and~${\mathcal L}^\nabla_v:=\iota_v \circ \nabla + \nabla \circ \iota_v$ commute. In addition, the grade $-2$ part of the latter acts as multiplication by $Q_v= \alpha(v)$. Remembering that the difference of two connections on ${\mathscr H}$ gives a section of the adjoint bundle ${\mathscr M}$, we 
have reformulated the problem as a search for a second connection $\nabla'$ such that 
$
\nabla'_u
$
and $\nabla$ commute. In particular, writing 
 $u = \frac\partial{\partial t}$, if the connection form $d^{\hat A}$ is independent of $t$ in some gauge, then we can simply take $d^{\hat A'}=d$ and $\widehat Q_{u} = \iota_u \circ (\nabla-\nabla') + (\nabla-\nabla') \circ \iota_u$. 
 This construction is carried out explicitly in the next example.

\begin{example}\label{chargeextraction}
Consider Cartesian coordinates $(p,q,t)\in{\mathbb R}^3$ and the 1-form
$$
\alpha: = pdq - H(p,q)dt\, ,
$$
where $H:\R^2\to\R$ is a polynomial in $p$ and $q$; $\alpha$ is contact iff $H$ is a non-vanishing function. Then, a quantum connection form, acting on forms and taking values in $L^2({\mathbb R})\ni \psi(s)$, is given by
$$
d^{\hat A}=dt\Big[\frac{\partial }{\partial t} +i \big[H\big(p-i\frac{\partial }{\partial s},q+s\big)\big]_W\Big]+
dp\Big[\frac{\partial }{\partial p}
+i s\Big]
+
dq
\Big[-ip+\frac{\partial }{\partial q}
-\frac{\partial }{\partial s}\Big]\, ,
$$
where we are setting $\hbar = 1$, see~\cite{HW} for details. Here the subscript $W$ denotes Weyl ordering with respect to the symbols $s$ and $\frac\partial{\partial s}$.
The vector field~$\frac{\partial}{\partial t}$ gives a strict contact symmetry.
The classical charge is $Q_{\!\frac{\partial}{\partial t}}=-H(p,q)$ while the quantum  charge reads
$$
\widehat{Q}_{\!\frac{\partial}{\partial t}}=i \big[H\big(p-i\frac{\partial }{\partial s},q+s\big)\big]_W\, .
$$
It is readily seen that this anti-hermitean operator commutes with $d^{\hat A}$, as required.
\hfill$\blacksquare$\end{example}

The next two subsections  conclude this Section 
by briefly discussing how the above strict quantization recovers the Schr\"odinger equation in quantum mechanics, and how to obtain physical correlators from the quantum connection. 

\subsection{Physical States}\label{ssec:PhysicalStates} The  analysis of Subsection \ref{ssec:BRSTComplex} establishes that BRST quantization of the action functional $S_\alpha$ produces a Hilbert bundle $\SH$ over the contact manifold~$Z$ endowed with a flat connection $\nabla$. The true physical quantum states are given by certain elements in the cohomology of the BRST charge. Moreover, the space of chains consists of $\SH$-valued differential forms on the initial contact manifold $Z$. In particular, in degree zero, quantum wavefunctions~$\Psi\in \Gamma(\SH)$ are defined by the partial differential equation
$$\nabla\Psi=0\, ,$$
This recovers the Schr\"odinger equation for Hamiltonian dynamics on ${\mathbb R}^{2n}$ as follows.

\begin{example}
\label{Hamsys}
Consider the strict Darboux ball $(\R^{2n+1},\la_\st-H(q,p,t)dt)$, where $(q,p)\in\R^{2n}$ are Cartesian coordinates and $H\in C^\infty(Z)$ is the contact Hamiltonian to be quantized. Also, we choose the Hilbert space $L^2({\mathbb R}^n)\ni \Psi(s^{\bar a})$.
A formal solution for  the quantized BRST charge $\widehat\Omega=-i \nabla$ is then given by the connection form $d^{\hat A}=d+\hat A$ with
$${\hat A}=\frac{\alpha}{i\hbar}+\frac{1}{i\surd\hbar}
\Big(\big(dp_{\bar a}+\frac{\dd H}{\dd {q^{\bar a}}}dt\big)s^{\bar a}+i\big(dq^{\bar a}-\frac{\dd H}{\dd {p_{\bar a}}}dt\big)\frac{\dd}{\dd{s^{\bar a}}}\Big)\!
-i{dt}\sum_{k\geq2}\frac{\hbar^{\frac k2-1}}{k!}\frac{\dd^k H}{\dd w_{a_1}\ldots \dd w_{a_k}}\, \hat{s}_{a_1}\!\cdots\hat{s}_{a_k},$$	
where $w_a=(q,p)\in\R^{2n}$ and $\hat s_a=(s^{\bar a},\frac1 i \frac\partial{\partial s^{\bar a}})$. In particular, the BRST charge is a finite sum for a polynomial Hamiltonian $H\in C^\infty(Z)$. The parallel condition $$\nabla\Psi=0$$ 
on wavefunctions $\Psi  \in \Gamma(\SH)$ decomposes into three systems of equations, proportional to  $dq^i$, $dp_i$ or $dt$. Upon solving the first two of these, the Schr\"odinger equation with Hamiltonian given by  the canonical Weyl-ordered quantization of $H$ is precisely this last of these equations $($see also ~\cite{HW}$)$.
Indeed, when  the classical Hamiltonian $H$ is independent of the time coordinate $t$,  a quantized Hamiltonian operator that is Hermitean and commutes with $d^{\hat A}$ can be obtained in the way
 discussed in Example~\ref{chargeextraction}.
\hfill $\blacksquare$
\end{example}

\subsection{Correlators}

Finally, let us sketch how physical correlators and probabilities are obtained from the quantum connection. Let \Aries\ be a section of the BRST Hilbert space and $\nabla = d + \hat A$. Then, a formal solution to the parallel condition $\nabla\mbox{\Aries}  =0$ is given in terms of the line operator
$$\mbox{\Aries}(z) = P \exp\Big(-\int^z \hat A\Big)\,  \mbox{\Aries}_0\, .
$$
In the above, $P$ denotes path ordering along any path in $Z$ ending at $z$ and $\mbox{\Aries}_0$ is any closed $\Gamma(\SH)$-valued differential form, used as an initial condition. The above amounts to parallel transport of the data $\mbox{\Aries}_0$ at the initial point of this path to $z\in Z$. Intuitively, for physical correlators, we imagine preparing a (normalized) state $|\Psi_i\rangle\in {\mathcal H}_{z_i}$ at a given $z_i\in Z$, and wish to compute its overlap with some other (normalized) state $|\Psi_f\rangle$ prepared at $z_f$. In this case, the transition probability $P_{fi}$ of measuring this final state given the initial one is
$$
P_{fi}=\Big|\Big\langle \Psi_f\Big| P \exp\Big(-\int^{z_f}_{z_i} \hat A\Big)\Psi_i\Big\rangle \Big|^2\, .
$$
{\it A priori}, the above prescription depends on a choice of path between $z_i$ and $z_f$ whenever the flat connection~$\nabla$ has non-trivial holonomy. In that case, we must modify the section space of the Hilbert bundle by suitably  quotienting by the holonomy operator induced by $\nabla$. {\it E.g.}, this explains how our quantization detects the difference between quantization on a line versus a circle. Ths holonomy  quotient  is described in more detail in~\cite{HLW}.

\section{Contractor Bundles and Their Connections}\label{sec:Contract}

In this section we develop aspects of  tractor  calculus~\cite{BEG} (see also~\cite{Capvak})  for contact structures~\cite{Fox,GNW} necessary for their quantization.
Section~\ref{sec:Quant} details the quantization of contact structures $(Z,\xi)$ and relies both on the results of the present section and the previous section on strict quantization.

\subsection{Symplectization Cone and Densities}
\label{ssec:symplectization} Let $(Z,\xi)$ be a $(2n+1)$-dimensional co-oriented contact manifold. The contact structure $\xi$ 
gives rise to an equivalence class of contact 1-forms $\bm \alpha= [\alpha]=[\Omega^2 \alpha]$, where $\Omega\in C^\infty (Z,\R^+)$ is a smooth, positive, real-valued function; we therefore often replace the label $\xi$ by $\bm \alpha$. This class of 1-forms $\bm \alpha$ defines a rank one subbundle
$$\Cone:=\{v^*\in\Gamma(T^*\!Z):v^*(\xi)=0\}\sse T^*\!Z,$$
over $Z$, {\it i.e.}   $\Cone$ is the space of covectors on $Z$ vanishing along the hyperplane $\xi\sse TZ$. The fibers of $\Cone$ at a point $p\in Z$ are given by the possible values of $\alpha_p$ for the contact forms $\alpha$ in the class $\bm \alpha$. This is an $\R^+$-principal bundle over $Z$ with $\R^+$-action $\alpha_p\mapsto \Omega^2_p \alpha_p$. The quotient of the cone $\Cone$ by this $\R^+$-action is contactomorphic to $(Z,\xi)$. One of the central themes of our dynamical quantization for $(Z,\xi)$ is presenting the contact distribution $\xi$ as a {\it conformal} geometry, {\it i.e.} by understanding a contact distribution $\xi$ as an equivalence class of (contact) 1-forms {\it up to} positive scaling.

The standard symplectic structure $(T^*\!Z,\omega_\st)$ restricts to a symplectic structure $(\Cone,\omega_\st)$ in $\Cone$ which, as above, we refer to as the {\it symplectization} of the contact manifold $(Z,\xi)$; see also our earlier discussion in the introduction. The central feature of the bundle $\Cone$ is that its sections are in bijection with contact 1-forms. The $\R^+$-action filters the space of sections in terms of eigenspaces, according to the weight of the representation. In order to keep track of these tensors, let us introduce the corresponding density bundles.

\begin{definition}
The density bundle $\ce Z[w]$ of weight $w\in\R$ is the associated rank-one vector bundle determined by the  $\R^+$-principal bundle $\Cone$ and the $\R^+$-representation $t\mapsto \Omega^{-w} t\in {\mathbb R}$.
\hfill $\blacksquare$
\end{definition}

\noindent
A smooth section $\bm f$ of the density bundle $\ce Z[w]$ is a  equivalence class
$$\bm f=[\alpha,f]=[\Omega^2 \alpha, \Omega^w f]\, ,$$
where $f\in C^\infty Z$. Conversely, the data of a function $f$ and a choice of $\alpha\in \bm \alpha=[\alpha]$ defines a section $[\alpha,f]\in \Gamma(\ce Z[w])$ for each $w\in\R$. In fact, once a contact form $\alpha\in\Gamma(\Cone)$ is fixed, sections of $\ce Z[w]$ 
are in bijection with functions on ${\mathcal C}$ with homogeneity $w$ 
with respect to the $\R^+$-action, {\it i.e.}   functions $f\in C^\infty \Cone$ obeying
$$
{\mathcal L}_Xf
=wf\, ,
$$ 
where $X$ is the vector field associated to the $\R^+$-action. In general, given any tensor $z\in\Gamma({\mathbb T}Z)$  and a choice of contact form $\alpha$, the same maneuver as above defines a section $[\alpha,z]$ of ${\mathbb T}Z[w]:=\ce Z[w]\otimes {\mathbb T}Z$. Generically, we denote  $BZ[w]:=BZ\otimes \ce Z[w]$, where $BZ$ is an arbitrary vector bundle over $Z$.

Let $v\in \Gamma(T{\mathcal C})$ be a tangent vector which is nowhere parallel to the $\R^+$-infinitesimal generator~$X$. Then we shall say that $v$ is {\it in the distribution}  when $v\in\pi^*(\xi)=\ker \lambda_\xi$, so that $\pi_*(v)\in \Gamma(\xi)$, where~$\pi$ is the projection map for the bundle $\pi:{\mathcal C}\to Z$. Similarly to functions, vectors $v$ on ${\mathcal C}$ in the distribution obeying the homogeneity condition
$$
{\mathcal L}_X v=- v\, $$
are in bijection with sections of $\xi[-1]$. These identifications shall be made implicitly from now on without comment; we will also employ a bold notation for tensor densities of non-vanishing weight, and an unbolded notation for their corresponding tensors with fixed homogeneity on ${\Cone}$.

 Recall that the choice of a contact form $\alpha\in \bm\alpha$, determines the Reeb vector field $\rho^\alpha\in\Gamma(TZ)$, uniquely defined by the two conditions
$$
d\alpha(\rho^\alpha)=0,\qquad \alpha(\rho^\alpha)=1\, .
$$
The $2$-form $\varphi:=d\alpha$ will be referred to as the Levi form. The equivalence class  $\bm \alpha$ of the contact form itself defines the set of Reeb vector fields
$\rho^{\bm \alpha}=\{\rho^\alpha|\alpha \in \bm \alpha\}$.
The codistribution 
$\xi^*$  is the quotient bundle given pointwise by the cokernel of $\bm \alpha$, so $[\omega]=[\omega+\nu \alpha]\in \Gamma(\xi^*)$.
From this data, we can define the projection of the gradient operator $d^\xi$ to the codistribution $\xi^*$ by
$$d^\xi:C^\infty Z\longrightarrow \Gamma(\xi^*),\quad f \longmapsto [d f]:=d^\xi f\, .$$
Upon choosing a  contact form $\alpha\in \bm \alpha$, we obtain a direct sum decomposition $T^*\!Z=\Span(\alpha) \oplus \xi^*_\alpha$ and the projection reads
$$
C^\infty Z\longrightarrow \Gamma(\xi^*_\alpha),\quad f \longmapsto 
d^\alpha f:=d f - ({\mathcal L}_{\rho^\alpha} f)\alpha,
$$
where $\xi^*_\alpha:=\ker \iota_{\rho^\alpha}$. Finally, given any vector $v\in \Gamma(\xi)$, we define $v^\flat\in \xi^*$ by $
v^\flat = \varphi(v,\pdot)$. Similarly, given $\omega\in \xi^*_\alpha$, we define $\omega^\sharp\in \Gamma(\xi)$ by the unique solution to $\varphi(\omega^\sharp,\pdot) = \omega$. Note that the musical maps satisfy $(\omega^\flat)^\sharp = \omega \mbox{ and } (v^\sharp)^\flat = v$. We are now ready to introduce one of the geometric ingredients for the quantization of a contact manifold, the contractor bundle of $(Z,\xi)$.

\subsection{Contractor Bundles}\label{cbundles}
A Cartan geometry is the data $(G,H,X)$ of a Lie group $G$, a subgroup~$H\sse G$ and a manifold $X$ equipped with a $G$-connection that identifies the tangent spaces of $X$ with 
the tangent space $T_{1\cdot H}(G/H)$ of the model geometry of the coset $G/H$. This structure is referred to as a parabolic geometry in the case that $H\sse G$ is a parabolic subgroup. Tractor bundles~\cite{BEG,CapGov} are natural vector bundles associated to parabolic geometries that can be endowed with canonical classes of linear connections, corresponding to their underlying Cartan connections. These bundles are extremely useful for the construction of invariant operators and allow representation theory to be applied to the solution of differential geometric problems associated to the corresponding geometry (see for example~\cite{Capvak}). Instances of this include  results in conformal geometry, projective structures, and CR-structures~\cite{BEG}.

The group of contactomorphisms of $(Z,\xi)$ is  a regular, infinite-dimensional Fr\'echet Lie group, in striking contrast to the finite-dimensional group of conformal isometries. Note that dimension $d$ conformal geometries 
are in direct correspondence with the  Lie group $SO(d+1,1)$, the parabolic subgroup stabilizing a lightlike ray, and a distinguished normal Cartan connection.
This is not the  case for a contact structure $(Z,\xi)$, and indeed the appearance of functional dimensionality  is reflected by  additional choices required to determine distinguished Cartan connections. For example, in the case $G=\operatorname{Sp}(2n+2)$, a normality condition on the Cartan connection yields a contact projective structure, for which  a projective connection with Legendrian geodesics starting in the contact distribution $\xi$ needs to be added to the data~\cite{Fox}. The constructions presented in this manuscript are not constrained by the normality condition, a feature which dovetails with the 
space of BRST quantizations for a given contact structure~$(Z,\xi)$.

A central object of geometric interest in this article is defined as follows:

\begin{definition}\label{def:contractor}
	Let $(Z,\bm \alpha)$ be a contact manifold. The {\it contractor bundle}~$\FZ\longrightarrow Z$ is an equivalence class of direct sum bundles
	$$
	\FZ = \Big(\bigsqcup_{\alpha\in \bm \alpha} \FZ^\alpha \Big)\Big/ \!\sim\, ,
	$$
	where for any contact form $\alpha\in\bm \alpha$, $\FZ^\alpha$ in the above disjoint union is the vector bundle
	\begin{equation}\label{direct_sum}
	\FZ^\alpha=\ce Z[1] \oplus \xi[-1] \oplus \ce Z[-1],
	\end{equation}
	and the defining equivalence relation is given by:
	\begin{equation}
	\nonumber
	\FZ^{\alpha}\ni (\bm v^+,\bm v,\bm v^-) \sim (\bm v^+,\bm v-\bm v^+\bm \Upsilon^\sharp,\bm v^-+\bm \Upsilon(\bm v)-\tfrac12\bm v^+\bm \chi)\in \FZ^{\Omega^2\alpha}\, .
	\end{equation}
	In the above expression, $\Omega\in C^\infty Z$ is a positive function and 
	$$\bm \Upsilon^\sharp :=[\alpha,(d^\alpha \log \Omega)^\sharp]\in \Gamma(\xi[-2]),\qquad \bm \chi:=[\alpha,{\mathcal L}_{\rho^{ \alpha}}\! \log \Omega]\in \Gamma(\ce Z[-2]),$$
	$$\bm \Upsilon(\bm v)=[\alpha,\iota_v d\log \Omega]\in \Gamma(\ce Z[-1]),$$
	where $\bm v=[\alpha,v]$ defines the vector $v$.
	\hfill $\blacksquare$
	\end{definition}

\noindent The equivalence relation in Definition~\ref{def:contractor} captures the change of the $1$-jet $j^1f$ of a function $f\in C^\infty Z$ under the rescaling $\alpha\mapsto\Omega^2\alpha$. The contact form $\alpha$ specifies a given decomposition of $df$ with respect to  $\ker\alpha$ and the Reeb vector field $\rho_\alpha$, and Definition \ref{def:contractor} precisely establishes the transformation of this decomposition with respect to the choice of contact form for the distribution $\xi=\ker(\alpha)$.

The unipotent lower triangular structure of the equivalence relation in Definition~\ref{def:contractor} is the contact geometric avatar of the parabolic subgroups appearing in Cartan geometries~\cite{Capvak,Currygo}. Indeed, let~$P$ be the
 parabolic subgroup $P=(\R^+\otimes \operatorname{Sp}(2n))\ltimes {\mathbb R}^{2n}\ltimes \R^+$ in $\operatorname{Sp}(2n+2)$ obtained from the stabilizer of a ray in~${\mathbb R}^{2n+2}$ under the canonical action of $\operatorname{Sp}(2n+2)$ on the space of rays. 
 Also, let ${\mathbb V}$ be the fundamental representation of $\operatorname{Sp}(2n+2)$.
Then the  contractor bundle $\FZ$ can also be obtained by considering an associated vector bundle ${\mathscr P}\times_{P}{\mathbb V}$ where ${\mathscr P}$
is a $P$ principal bundle over $Z$.
Topologically, the contractor bundle~$\FZ$ is a semi-direct sum of the two line bundles $\ce Z[1]$, $\ce Z[-1]$ and the contact distribution $\xi[-1]$, {\it i.e.}   $\FZ\cong\ce Z[1]\lpl\xi[-1]\lpl \ce Z[-1]$.
An abstract index notation $V^A$ will be used for sections of the contractor bundle, where implicitly $1\leq A\leq (2n+2)=\mbox{rk}(\FZ)$. These sections will be referred to as {\it contractors}. Given a choice of $\alpha\in\bm \alpha$, the section $V^A$ is determined by the triple of sections $\bm v^+\in\Gamma(\ce Z[1])$, $\bm v\in \Gamma(\xi[-1])$, $\bm v^-\in \Gamma(\ce Z[-1])$, so we will also employ the matrix notation
$$
\Gamma(\FZ)\ni V^A\stackrel\alpha=\begin{pmatrix}
\bm v^+\\\bm  v\\\bm v^-
\end{pmatrix}\in \Gamma(\FZ^\alpha)\, .
$$
Also, given a representative $\alpha$, we will refer to  $\bm v^+$, $\bm v$ and $\bm v^-$ as the top, middle and bottom components of the contractor $V^A$.

\begin{remark}
In this article we employ four equivalent definitions for  the contractor bundle. The first two---an equivalence class of direct sum bundles and an associated vector bundle to a principal~$P$ bundle---are described directly above. The  other two are in terms of $(2n+2)$- and $(2n+3)$-dimensional ambient manifolds.
For the first of these latter two, the contractor bundle is identified with the tangent bundle to the symplectization ${\mathcal C}$ of $Z$ suitably quotiented by the natural $\R^+$-action; this construction is detailed in the articles~\cite{Fox,GNW}. 
The latter is based on a $(2n+3)$-dimensional strict contact structure equipped with a homothetic $\R^+$-action and the details are given in Section~\ref{ambicontractious}.\hfill$\blacksquare$
\end{remark}

Tensor products of the contractor bundle $\FZ$ with density line bundles give rise to weighted contractor bundles $\FZ[w]$.
 The former, weight zero, contactor bundle with fundamental representation ${\mathbb V}$ will be termed the {\it standard} contractor bundle $\FZ=\FZ[0]$, and similarly, the standard {\it cocontractor bundle} is the dual:
$$\FZ^*:=\ce Z[-1]\rpl \xi^*[1]\rpl \ce Z[1],$$
whose sections are denoted by a lower index $W_A$. For a choice of $\alpha\in \bm \alpha$, these are given by
$$
W_A\stackrel\alpha=\begin{pmatrix}\bm w^- & \bm w & -\bm w^+\end{pmatrix}\, ,
$$
where $w^-\in \Gamma(\ce Z[-1])$, $w\in \Gamma(\xi^*[1])$ and $w^+\in \Gamma(\ce Z[1])$.
Repeated indices will denote contractions, so that $W_A V^A \stackrel\alpha = \bm w^- \bm v^+ + \iota_{\bm v} \bm  w- \bm w^+ \bm v^- $ denotes a function in $C^\infty(Z)$.

\subsection{Tensor Contractor Bundles and the Canonical Contractor}\label{ssec:contractorbundle} Tensor contractor bundles are  bundles built from the contractor bundle $\FZ$ via standard tensor constructions, including exterior and symmetric powers. In particular, there is a distinguished section
$$J_{AB}\in\Gamma\left(\bigwedge^2 \hh\FZ^*\right),\quad \Gamma(\FZ^*[w])\ni J_{AB} V^B \stackrel\alpha = \begin{pmatrix}\bm v^- & \bm v^\flat & -\bm v^+\end{pmatrix}
$$ of the second exterior power $\bigwedge^2 \hh\FZ^*$, defined by its action on a tractor vector $V^A\in \Gamma(\FZ[w])$. Hence, if $V^A\in \Gamma(\ct Z[w])$ and $U^B\in \Gamma(\ct Z[w'])$ then
\begin{equation}\label{JUV}
J_{AB} V^A U^B \stackrel\alpha= \bm u^+ \bm v^- - \bm u^- \bm v^+ + \bm\varphi(\bm u,\bm v)\in \Gamma(\ce Z[w+w']), 
\end{equation}
where the skew tensor  $\bm \varphi\in \bigwedge^2 \xi^*[2]$ is tautologically defined by $[\alpha;d\alpha]$. The contractor tensor $J_{AB}$ will be called the  {\it contractor symplectic form}. Its inverse is defined by
$$
J^{AB}\in \Gamma\left(\bigwedge^2 \FZ\right),\quad J_{AB} J^{BC}= \delta_A^C\, ,
$$
where $\delta_A^C$ is the abstract index notation for the identity endomorphism of the standard tractor bundle. We will also employ an index raising and lowering  notation for the map between contractors and cocontractors given  by the contractor symplectic form, so that $V_A:=J_{AB}V^B$
and $W^A:=J^{AB}W_B$. The geometry of the contractor symplectic form extends, for each $\alpha\in \bm\alpha$,  the conformally symplectic structure in the contact structure $\xi$ by adding a 2-dimensional symplectic vector space spanned by the pair~$(\bm v^-,\bm v^+)$.\\

The lower triangular equivalence defining the contractor bundle $\FZ$, as in Definition \ref{def:contractor}, canonically provides a non-vanishing section. Indeed, this section $X^A$ of $\FZ[1]$, referred to as the {\it canonical contractor}, can be defined by 
$$
X^A\stackrel\alpha =\begin{pmatrix}0\\0\\1\end{pmatrix}\, .
$$
This section gives the bundle inclusion $\ce Z[-1]\lr \FZ$ at the level of sections by
$$\bm \mu \longmapsto \bm \mu X^A,\mbox{ where }\bm \mu\in\Gamma(\ce Z[-1]) \mbox{ and } \bm \mu X^A\in \Gamma(\FZ)\, .$$
Similarly, this injection yields a surjective vector bundle map $\FZ \longrightarrow \ce Z[1]$, which reads 
$$V^A \mapsto  J_{AB} X^A V^B,\mbox{ where }V^A\in\Gamma(\FZ) \mbox{ and } J_{AB} X^A V^B \in \Gamma(\ce Z[1]).$$
Together, these two maps give the composition series
$$
\ce Z[-1]\longrightarrow \FZ \longrightarrow \ce Z[1]\, .
$$
The contact structure $\xi$ is the kernel of the above surjection quotiented by the inclusion map.

Finally, the contractor analog of the gradient operator is the {\it contact $D$-operator}, which is the map 
$$D_A:\Gamma(\ce Z[w])\to \Gamma(\FZ^*[w-1])$$
given by  
\begin{equation}\label{DDef}
\Gamma(\ce Z[w])\ni \bm \nu=[\alpha,\nu] \mapsto 
D_A \bm \nu\stackrel\alpha=\begin{pmatrix}
\tfrac12 {\mathcal L}_{\rho^\alpha}\nu &d^\alpha \nu &w\nu  
\end{pmatrix}\in \Gamma(\FZ^{{}^*}{\!}^\alpha [w-1])\, .
\end{equation}

\subsection{Scales and Contractor Connections}\label{ssec:ContractConnections} 
The BRST charge  introduced for the strict quantization of $(Z,\alpha)$ in Section \ref{sec:StrictQuant} for a fixed $\alpha \in \bm \alpha$, depends on a choice of connection on the hyperplane distribution~$\xi$. These connections shall be a primary ingredient in the quantization scheme presented in Section \ref{sec:Quant}. Therefore, we now study connections for the contractor bundle. This is the subject of the present subsection as well as Subsection \ref{ssec:ambientconnection}.

The contractor bundle introduced in Definition \ref{def:contractor} geometrically encodes the hyperplane distribution given by the contact structure $\xi$ as well as  the possible choices of contact form. A strictly positive section $\bm \tau\in \Gamma(\ce Z[1])$ is termed a {\it true scale}; it uniquely determines the contact form 
$
\alpha_{\bm \tau}
$
such that~$\bm \tau = [\alpha_{\bm \tau},1]$. By viewing $\alpha_{\bm \tau}$ as a 1-form valued weight zero density, we will also write~$\alpha_{\bm \tau} = \bm \alpha/ \bm \tau^2$. In general, the term {\it scale} will be used to label sections of~$\ce Z[1]$ that are not identically zero. For the choice $\alpha=\alpha_{\bm \tau}$ of contact form defined by $\bm \tau$, the weight~$-2$ density~$\bm \chi$ in Definition~\ref{def:contractor}
equals~$\bm \tau^{-2}{\mathcal L}_{\rho^\alpha} \log \Omega$.
Given a true scale, instead of labeling a standard tractor by a triple of density-valued objects, such as $\bm v^+\in \ce Z[1]$, $\bm v\in \xi[-1]$ and $\bm v^-\in \ce Z[-1]$, we may instead employ a pair of functions and a vector in  the distribution by setting $v^+=\bm v^+/\bm \tau$, $v=\bm \tau \bm v$ and $v^-=\bm \tau \bm v^-$. In that labeling, the defining equivalence relation of Definition~\ref{def:contractor} becomes
$$(v^+,v,v^-)\sim \big(\Omega v^+,\Omega^{-1}(v-v^+\Upsilon^\sharp),\Omega^{-1}(v^-+\Upsilon(v)-\frac 12 v^+ \chi)\big).$$
Our terminology for this is {\it computing in a scale} and we use an unbolded notation when performing this maneuver. This concludes our discussion on scales.

 Let us now define a contractor connection for the contractor bundle $\FZ$:

\begin{definition}\label{def:contractorconnection}
Let $\nabla^{\FZ}$ be a connection on the contractor vector bundle $\FZ\lr(Z,\xi)$. A connection~$\nabla^{\FZ}$ is said to be a {\it contractor connection} if the contractor
$\nabla^\FZ_z \big(X^A/\bm \tau)\in \Gamma(\FZ)$
obeys 
\begin{equation}\label{iamaconnection}
\nabla^\FZ_z \big(X^A/\bm \tau)\stackrel {\alpha_{\bm \tau}}= \begin{pmatrix}
2\alpha_{\bm \tau}(z)\\
z_{\bm \tau}\\
0
\end{pmatrix}\:\quad \mbox{ and } \quad\:
\nabla^\FZ_z J_{AB}=0\, ,
\end{equation}
for all sections $z\in \Gamma(TZ)$. In the above, the parallel condition for the contractor symplectic form is obtained by extending $\nabla^\FZ$ to act on weight zero cocontractors and higher (weight zero) tensor contractors by imposing the Leibniz property. The entry $z_{\bm \tau}\in \Gamma(\xi)$ denotes the projection of the section $z$ to the distribution determined by $\alpha_{\bm \tau}$, which is given by $z_{\bm \tau} := z- \alpha_{\bm \tau}(z)\rho^{\alpha_{\bm \tau}}$.
\hfill
$\blacksquare$
\end{definition}

 The first condition in Definition \ref{def:contractorconnection} is the non-degeneracy condition for a Cartan connection, whereas the second condition ensures that $\nabla^{\FZ}$ is an $\frak{sp}(2n+2)$-valued connection, {\it i.e.}   parallel transport integrates to fiberwise linear symplectomorphisms. To ease notation, the superscript ${\FZ}$ will be suppressed, writing $\nabla^{\FZ}=\nabla$, when it is clear from context that we are manipulating a contractor connection.

\noindent Note that upon decomposing the section $z_{\bm \tau}=v+\rho^{\alpha_{\bm \tau}} v^+$ in  terms of a vector field $v\in \Gamma(\xi)$ and a function $v^+$,  we may define contractors $Y^A\in \Gamma(\FZ[-1])$ and $Z^A_\pdot \in \Gamma(\xi^* \otimes \FZ)$ by
$$
\nabla_z \big(X^A/\bm \tau)
=2v^+Y^A+ Z_v^A\, .
$$
Then the tractors $(Y^A, Z^A_\pdot,X^A)$ determine the direct sum decomposition of the bundle $\FZ^{\alpha_{\bm \tau}}$ given in Equation~\nn{direct_sum} via
$$
\Gamma(\ct Z)\ni V^A\stackrel{\alpha_{\bm \tau}}=
\begin{pmatrix}v^+\\ v \\ v^-\end{pmatrix}=v^+ Y^A + Z_v^A + X^A v^-\in
\Gamma(\FZ^{\alpha_{\bm \tau}})\, .
$$
Observe that the contractor symplectic form can be written in terms of bilinears in the tractors $(Y^A,Z^A_\pdot, X^A)$ because the right hand side of Equation~\nn{JUV} can be re-expressed as 
$$
(u^+ Y^A + Z_u^A + X^A u^-)(v^+ Y_A + Z_{vA} + X_A v^-)=J_{AB}V^A U^B\, .
$$

\newcommand{\Rho}{P}

Let $\bm \tau$ be a true scale with corresponding contact form $ \alpha$. Then a contractor connection is  determined by the  data
$(\alpha,\nabla^\alpha,\Rho,Q)$,
where $\nabla^\alpha$ is itself a connection on the contact distribution $\xi$ that preserves the Levi 2-form $\varphi=d\alpha$, $\Rho\in \Gamma(T^*\!Z\otimes \xi)$ and $Q$ is a 1-form on $Z$.
That this data determines~$\nabla^\FZ$ can be seen by writing out the two
Conditions~\nn{iamaconnection} in Definition~\ref{def:contractorconnection}.  In these terms, the contractor connection acts on a standard tractor $V^A\stackrel\alpha=(v^+,v,v^-)$ 
according to
\begin{equation}\label{nablaV}
\nabla^\FZ_z \begin{pmatrix}
v^+\\[1mm]v\\[1mm]v^-
\end{pmatrix}
=
\begin{pmatrix}
{\mathcal L}_z v^+ + \varphi(v,z) +2v^-\alpha(z)\\[1mm]
\nabla_z^\alpha v-v^+\Rho_z +v^- z
\\[1mm]{\mathcal L}_z v^-+\frac12Q_z v^++\varphi(v,\Rho_z)
\end{pmatrix}\,.
\end{equation}

\begin{remark}
As discussed earlier, in the study of parabolic geometries, a normality condition on the curvature of a Cartan connection (given by the vanishing of the Kostant differential), is often required~\cite{Fox,Capnik,Capnik1}.
In this article we need the freedom of not imposing such a condition. Indeed, this is required in order to model basic physical systems such as the quantum harmonic oscillator. In particular, the Cartan connection $\nabla^\FZ$ will in general not be a normal Cartan connection.\hfill$\blacksquare$
\end{remark}

Contractors that are parallel with respect to a contractor connection play a distinguished {\it r\^ole}. These are neatly linked to the contact $D$-operator introduced in Subsection \ref{ssec:contractorbundle} above, as follows:

\begin{lemma}\label{parascale}
	Let $\nabla^\FZ$ be a contractor connection and $I^A$ a standard contractor subject to the parallel condition $\nabla^\FZ I^A  =0$. Then 
	\begin{equation}\label{scaletractor}
	I^A = D^A \bm \sigma
	\end{equation}
	where $\bm \sigma\in \Gamma(\ce Z[1])$ is given by $\bm \sigma = X_A I^A\, .$
\end{lemma}

\begin{proof}
	Applying Equation~\nn{nablaV}  to a parallel contractor $I^A\stackrel\alpha=(\sigma, m,\kappa)$
	imposes the condition
	$$
	d \sigma+\varphi(m,\pdot) +2\kappa\alpha=0\, .$$
	Contracting the above with the Reeb vector field $\rho^\alpha$ determines the component $\kappa$ as
	$
	\kappa = -\frac12{\mathcal L}_{\rho^\alpha} \sigma
	$
	and in turn 
	$
	m=(d\sigma)^\sharp
	$.
	Thus we have 
		$$
	I^A\stackrel\alpha=
	\begin{pmatrix}
	-\sigma\\
	(d\sigma)^\sharp\\
	\frac12 {\mathcal L}_{\rho^\alpha} \sigma
	\end{pmatrix}\, .
	$$
	Let $\bm \sigma =[\alpha,\sigma]\in \Gamma(\ce Z[1])$, since $X_A I^A=\bm \sigma$ it now suffices to observe that $J_{AB}I^B$ agrees with $D_A\bm \sigma$ as defined in Equation~\nn{DDef}.	
\end{proof}

\noindent Parallel contractors are thus $D$-exact sections. This is the contact geometric analog of the 
conformal
 Einstein equation~\cite{BEG} (see also~\cite[Theorem 3.11]{Currygo
}) in conformal geometry.

In summary, this subsection presented an intrinsic presentation of the contractor bundle $\FZ$, and of its sections and connections, in terms of the geometry of~$(Z,\xi)$. Now, drawing inspiration from the Fefferman-Graham construction \cite{Capgo,Gope} for a conformal class of metrics, we shall develop an extrinsic treatment of the contractor bundle $\FZ$ and its geometric structures, which will be used in our program for quantizing a contact structure $(Z,\xi)$. 

\subsection{Ambient Contractors}\label{ambicontractious} In this subsection we develop a novel ambient realization of the contractor bundles $\FZ\longrightarrow(Z,\xi)$ based on a $(2n+3)$-dimensional contact manifold $(M,\Xi)$. Its primary motivation is to extend our strict quantization results of Section \ref{sec:StrictQuant}, that depended on a preferred choice of contact form $\alpha$ for $\xi=\ker\alpha$, to the setting of contact manifolds, where the quantization of $(Z,\xi)$, with no distinguished contact form being chosen, will be achieved. This ambient approach, realizing contractors on an {\it ambient contact manifold}, is also of interest in its own right. Unlike the ambient approaches based on the projective Thomas space employed in~\cite{Fox,GNW}, it is  a direct analog of the Fefferman--Graham ambient space for conformal geometries in which an ambient  contact 1-form parallels the ambient metric.

Let $(Z^{2n+1},\xi)$ be a $(2n+1)$-dimensional contact manifold and let $\pi\!:\! \SC \!\to\! Z$ denote the $\R^+$-principal bundle of positive contact forms with kernel $\xi$, as introduced in Subsection \ref{ssec:symplectization}. The ambient construction of the contractor bundle $\FZ$ is based on the geometry of the following contact manifold:

\begin{definition}\label{def:ambientspace}
Let $(Z^{2n+1},\xi)$ be a contact manifold. A $(2n+3)$-dimensional strict contact manifold $(M^{2n+3},A)$, where $A\in \Gamma(T^*\!M)$ is a contact 1-form,  equipped with a proper free $\R^+$-action ${x:\R^+\times M\to M}$, infinitesimally generated by the vector field $X\in \Gamma(TM)$, and an $\R^+$-equivariant embedding $i:{\SC}\longrightarrow M$ is said to be an {\it ambient contact manifold for} $(Z,\xi)$ if the following three conditions are satisfied:

\begin{enumerate}[(i)]

	\item\label{poundtime}
	${\mathcal L}_X A=2A$,\\[.1mm]
	\item\label{isquaretominusone} $\Cone=\Zero(A(X)),$
	\\[.1mm]
	\item\label{third} For any contact form $\alpha$ on $(Z,\xi)$, $\exists\sigma \in \Gamma(\SC)$ such that $\alpha = (i \circ \sigma )^{*}(A)$.\label{compatibility}
\hfill $\blacksquare$	
	\end{enumerate}

\end{definition}

\noindent 

Condition~(\ref{poundtime}) ensures that the action $x$ induces an endomorphism of $\Gamma(\Xi)$ by pushing forward, where $\Xi=\ker(A)\sse TM$. Condition~(\ref{isquaretominusone}) ensures that $\Cone$
is in bijection with the data $(Z,\bm \alpha)$. Note also that the function $A(X)$ is a defining function for the submanifold $\Cone$. Indeed, along $\Cone$ we have that $d\big(A(X)\big)=2A-X^\flat\neq 0$ is non-vanishing. Here, the musical map $\flat$ is defined ambiently, given the choice of contact form $A$, in the same way as discussed earlier for any contact manifold. To ease notation, we shall implicitly assume the existence of the $\R^+$-equivariant embedding $i$ and denote by~$\Cone$ the image $i(\Cone)\sse (M,A)$, and implicitly pull-back structures to $\Cone$ without explicitly using $i^*$.

\begin{example}
The manifold $M={\mathbb R}^2\times Z$ from Example~\ref{getmaxy}, equipped with contact form 
$
A=e^{2\mu} \alpha + d\theta$
is an ambient contact manifold for the pair $(Z,\xi)$, $\xi=\ker(\alpha)$.  The $\mathbb R^+$-action is given by
$$
(\mu,\theta,z)\stackrel{x_\lambda}\longmapsto
(\mu+\log\lambda,\lambda^2 \theta,z)\, , 
$$
and the submanifold ${\mathcal C}=\{\theta=0\}$.
The section $\sigma$ of ${\mathcal C}$ defined by $z\mapsto (\mu,z)=(\varpi(z),z)$ obeys $$\sigma^*(A)=e^{2\varpi(z)}\alpha\, ,$$ 
for any $\varpi\in C^\infty(Z)$.
\hfill$\blacksquare$\end{example}

\noindent The data in Definition \ref{def:ambientspace} suffices to define the ambient analog of the contractor bundle as follows: 

\begin{definition}\label{def:ambientcontractor}
Given an ambient contact manifold $(M,A,x)$ with contact distribution $\Xi=\ker A$, the {\it ambient contractor bundle}
	is the rank-$(2n+2)$ vector bundle $
	\Xi|_{\mathcal C}/\!\sim$
	given by the restriction of~$\Xi$ to $\mathcal C$ quotiented by the equivalence
	$$
	V_P \sim  \big(\lambda^{-1}dx_\lambda (V)\big)_{x_\lambda(P)}\, , \quad V_P\in \Xi|_{\Cone}\, ,
	$$
	where $dx_\lambda$ is the differential of the $\R^+$-action, for 
	some $\lambda\in\R^+$.
	\hfill $\blacksquare$
\end{definition}

The space of sections of the ambient contractor bundle $\FZ$ is in bijection with elements of ~$\Gamma(\Xi)$ with homogeneity minus one and support along ${\mathcal C}$. These amount to equivalence classes of ambient sections of the distribution $\Xi$ 
subject to the differential condition
\begin{equation}\label{homog}
{\mathcal L}_X V= -V\, ,
\end{equation}
considered up to the equivalence 
\begin{equation}\label{supp}
\Gamma(\Xi) \ni V \sim V + A(X) U\, ,
\end{equation}
for any $U\in \Gamma(\Xi)$ of homogeneity minus three. The equivalence relation imposed in Definition \ref{def:ambientcontractor} is essentially stating that $\FZ$ is well-defined as a bundle over the image of any section of $\Cone\longrightarrow Z$, and thus~$Z$ itself. The quantization of $(Z,\xi)$ will be focused in the germ defined by the hypersurface~$\SC\sse M$, and thus equivalence up to the vanishing of $A(X)$ is sufficient.

\begin{remark}
It is possible to enlarge this space and include ambient vectors that lie in the distribution only along $\Cone$. Similarly, given a vector $V$ that obeys $\Lie_X V = -V$ along $\Xi$, we may always add terms proportional to $A(X)$ so that $\Lie_X V + V=0$ everywhere in $M$.\hfill$\blacksquare$
\end{remark}
\noindent We will refer to an ambient vector obeying these conditions as an {\it ambient contractor}.

\subsection{Ambient Scales} Homogeneity one positive functions $\tau\in C^\infty(M)$ will be referred to as {\it ambient scales}. 
Let $R$ be the Reeb vector of the ambient contact form $A$.
From an ambient scale $\tau$, we may build the ambient 1-form:
\begin{equation}\label{why}
Y_\tau := d\log\tau  - \tfrac12 
({\mathcal L}_R\log\tau)\, 
d(A(X)) \, ,
\end{equation}
which  satisfies the following three properties: 
\begin{enumerate}[(i)]
	\item 
	\label{1}
	$Y_\tau\in \Gamma(\Xi^*_A)$, {\it  i.e.}  $Y_\tau(R) = 0$.
	\item \label{2} $Y_\tau$ has homogeneity zero: ${\mathcal L}_X Y_\tau=0$.
	\item $Y_\tau$ obeys a normalization condition along the ambient symplectization of $Z$,
	$$
	Y_\tau(X)\big|_{\Cone}=1\, .
	$$
	\item
	\label{4} For any ambient scale $\tau' = \tau + \mathcal{O}(A(X))$, we have $Y_{\tau'} = Y_\tau +  \mathcal{O}(A(X)).$ 
	
\end{enumerate}
Property~(\ref{1}) holds because  
$Y_\tau$ obeys
\begin{equation}\label{oldplus}
Y_{\tau}=d_A \log \tau + \frac 12 X^\flat \Lie_R \log \tau\, .
\end{equation}
Property~(\ref{2}) follows from the homogeneity requirements ${\mathcal L}_X\tau= \tau$, and ${\mathcal L}_XA= 2A$, the latter of which implies ${\mathcal L}_X R= -2R$ (see Lemma \ref{lem:Xprime} for details). 
The remaining properties are easily checked. Property~(\ref{4}) ensures that $Y_{\tau'}|_{\Cone}=Y_{\tau}|_{\Cone}
$ whenever ${\tau'}|_{\Cone}={\tau}|_{\Cone}$. Hence the  ambient cocontractor 
$[Y_{\tau}]$, defined analogously to an ambient contractor,
is determined by a choice of $\alpha\in\bm \alpha$. To ease notation, the ambient scale subscript $\tau$ in $Y_\tau$ will not be written if it is clear from context.

Finally, note that because the function $\tau$ obeys the homogeneity condition ${\mathcal L}_X\tau=\tau>0$, the hypersurface in $\Cone$ determined by $\tau=1$ is everywhere transverse to $X$, and therefore defines a section of the symplectization ${\mathcal C}$, {\it i.e.} a choice of contact form $\alpha\in \bm \alpha$. Thus, a choice of positive, homogeneity one function $\tau$ on the ambient manifold~$(M,A)$ determines a contact form on $(Z,\xi)$ and a corresponding section $\sigma$ of ${\mathcal C}$. Also, the $(2n+1)$-dimensional submanifold $Z^\alpha:=\Cone\cap \Zero(\tau-1)=\sigma (Z)$ is diffeomorphic to $Z$. In short, ambient scales are the geometric objects that allow us to recover the contact geometry  $(Z,\xi)$, including the choice of any contact form.

\begin{example}
Returning to Example~\ref{getmaxy}, the function $\tau=e^{\mu} \in C^\infty({\mathbb R}^2\times Z)$ has homogeneity one. Pulling back the 1-form $A$ to the  submanifold $\{\theta=0\}\cap \{\mu=0\}$ gives the contact form $\alpha$. Note that the Reeb vector of $A$ is $R=\frac\partial{\partial \theta}$ and that
$Y_\tau=d\mu$, 
which has homogeneity zero. Moreover, $\iota_R Y_\tau=0$ and so $Y_\tau\in \Gamma(\Xi^*_A)$. It is  not difficult to check that $Y_\tau$ is unchanged when a function of $\theta$  is added to~$\tau$, and in turn deduce that condition~(iv) of the above list holds. \hfill$\blacksquare$\end{example}

\subsection{The Contractor Equivalence} We now establish the equivalence between the contractor bundle~$\FZ$, as constructed in Definition \ref{def:contractor}, and its ambient realization~$\Xi|_{\mathcal C}/\!\sim$, introduced in Definition~\ref{def:ambientcontractor}. The central result of this subsection is the following isomorphism:

\begin{theorem}\label{thm:equivalence}
Let $(Z,\xi)$ be a contact manifold and $(M,A,x)$ an ambient contact manifold for $(Z,\xi)$. Then there exists an isomorphism $\FZ\cong\Xi|_{\Cone}/\!\sim$ of vector bundles over $Z$.

\noindent In fact, given an ambient scale $\tau$, and thus $\alpha \in \bm \alpha$, the map
$$\Big(\Xi|_{\Cone}/ \!\sim\Big) \ni
[V]\mapsto
(\bm v^+,\bm v,\bm v^-)\in \ct^\alpha Z=\ce Z[1] \oplus \xi[-1] \oplus \ce Z[-1]\, ,
$$
defined by the following quantities, restricted to the cone $\Cone$,
\begin{equation}\label{topmiddlebottom}
v^+=\Phi(V,X),\quad
v=V-Y_\tau(V)\cdot X -\Phi(V,X)\cdot Y_\tau^\sharp
\, ,
\quad
v^-=Y_\tau(V),
\end{equation}
is an isomorphism of vector bundles.\hfill$\blacksquare$
\end{theorem}

\noindent The proof of Theorem \ref{thm:equivalence} shall constitute the remainder of the present subsection. The statement can be proved in  two steps. First, establish the isomorphism for a choice of contact form~$\alpha$, {\it i.e.}   compare the vector bundles $\FZ^\alpha$ and $\Xi|_{Z^\alpha}$. Second, prove that the rescaling transformations induced by a rescaling of $\alpha$ coincide for these vector bundles. \\

\noindent
{\bf First Step.} Let us fix a contact form $\alpha\in \bm\alpha$ and start by showing that the ambient contractor bundle $\Xi|_{Z^\alpha}$ is isomorphic to the (intrinsic) tractor bundle $\FZ^\alpha=\ce Z[1] \oplus \xi[-1] \oplus \ce Z[-1]$. Note that, as stated above, the smooth submanifold $Z^\alpha$ is diffeomorphic to $Z$, as $Z^\alpha$ is a transverse slice of the cone over~$Z$. In this proof we leave the choice of scale $\tau$ associated to $\alpha$ implicit and often write $Y_\tau=Y$.

The filtered isomorphism in the statement of Theorem \ref{thm:equivalence} relies on
 the fact that a representative~$V$ for a standard ambient tractor decomposes as
\begin{equation}\label{decompose}
V=\tilde v^+ Y^\sharp + \tilde v + \tilde v^- X+A(X) U\, ,
\end{equation}
where $\tilde v^+=\Phi(V,X)$, $ \tilde v^-=Y(V)$, $U\in \Gamma(\Xi)$ has homogeneity minus three, $\tilde v$ satisfies $\Phi(\tilde v,X)=0$, and $Y(\tilde v)=0$. This decomposition uses the (dual of the) 1-form $Y_\tau$.

Note that the musical map $\sharp$ is defined ambiently, given the choice of contact form $A$, in the way discussed earlier for any contact manifold. Also, notice that $Y^\sharp$ is not parallel to $X$ along $\Cone$ because $\Phi(Y^\sharp,X)\neq 0$ there.

Restrictions of smooth functions of the ambient space $M^{2n+3}$ to the cone ${\mathcal C}$, of a given homogeneity~$w$, correspond to sections of the weight $w$ density bundles $\Gamma(\ce Z[w])\lr Z$. The functions $\tilde v^+$ and~$\tilde v^-$ in Equation~\nn{decompose}
have definite homogeneities one and minus one, respectively.
Upon restriction to $\mathcal C$, they define functions  with the same homogeneities with respect to the $\R^+$-action.
Hence they define
sections of~$\ce Z[1]$ and~$\ce Z[-1]$, that we shall respectively denote $v^+$ and $v^-$. These are the components bijectively mapping to $\ce Z[1]$ and $\ce Z[-1]$ in the statement of Theorem \ref{thm:equivalence}. Notice that the map is a bijection, up to fixing $\tilde v$, due to the fact that the Levi form $\Phi$ is non-degenerate.

Let us now conclude that the component $\tilde v$ can always be determined, and uniquely. That is, let us describe precisely which subbundle of the ambient contractor bundle corresponds isomorphically to the $\xi[-1]$ piece. It suffices to show that the vector $\tilde{v}$ in the decomposition above is the pushforward of a unique vector $v \in \Gamma(\xi)$. For this, we note that for any $v \in \Gamma(\xi)$, we have 
$$A\big(\sigma_* v\big) = \left(\sigma^*A\right)(v) = \alpha(v) = 0\, ,$$
where  we used Item~\eqref{third} of  the  ambient contact manifold Definition~\ref{def:ambientspace}. 

Hence $\sigma_* v \in \ker A = \Xi$. 
Next, notice that
\begin{equation}\label{iamacuteequation}
X^\flat = \Phi(X,\pdot) = 
\iota_X d A = {\mathcal L}_X A - d \big(A(X)\big) = 2A - d \big(A(X)\big)\, . 
\end{equation}
\noindent Now we compute 
\begin{equation}\label{daxi}
X^\flat\big(\sigma_* v\big) = 2A\big(\sigma_* v\big) - d\big(A(X)\big)\big(\sigma_* v\big) = - d\big(A(X)\big)\big(\sigma_* v\big) = 0\, .
\end{equation}
The last equality holds because $A(X)$ is a defining function for $\mathcal C$ and the ambient vector $\sigma_* v$ is tangent to $\Cone$.
Thus $\sigma_* v \in \ker X^\flat$. Finally, we can compute
$$
Y\big(\sigma_* v\big) = d\log\tau\big(\sigma_* v\big) -\tfrac 12 (\mathcal{L}_R\log\tau)  \big(d\big(A(X)\big)\big)\big(\sigma_* v\big) = d\log\tau\big(\sigma_* v\big) = 0\, .
$$
For the first equality we used Equation~\nn{daxi}, while
the last equality holds because $\sigma_* v$ is tangent to~$\sigma(Z)$ and this ambient submanifold is a constant locus of $\tau$.
Therefore, $\sigma_* v \in \ker Y$. 
Orchestrating, we have that along $Z_\alpha$, 
$$
\sigma_* v \in \ker A \cap\ker X^\flat \cap\ker Y.
$$
Remembering that  $\sigma$  defines an embedding, it follows that  $\big(\!\ker X^\flat \cap \ker A \cap \ker Y)|_{Z^\alpha}$ is in bijection with $\Gamma(\xi)$.
Vector fields along $Z^\alpha$ can be pushed forward along the $\R^+$-action to homogeneous vector fields on $\Cone$. Hence, we have an isomorphism between sections of $\xi$ and homogeneous vector fields on~$\Cone$ subject to the triple kernel condition above. 
Remembering that $V$ has homogeneity minus one,  we may use this result to identify the
restriction of~$\tilde v$ to $\Cone$ in Equation~\nn{decompose} with a section $\bm v$ of $\xi[-1]$. This concludes the first step in proving Theorem \ref{thm:equivalence}.\\

\noindent
{\bf Second Step.} Now we show that ambient contractors transform in the same manner as (intrinsic) contractors, {\it i.e.}   those defined directly in terms of direct sums of density bundles as in Equation~\eqref{direct_sum}, or equivalently, that the $\R^+$-action used to quotient $\Xi$ for the definition of $\FZ$  induces the equivalence relation in Definition~\ref{def:contractor}. First, the component $v^+$ of the (intrinsic) contractor bundle in Definition~\ref{def:contractor} is invariant under rescaling of the contact form $\alpha$, and thus it suffices to study the transformations of the components $v^-$ and $v$ under a different choice of ambient scale $\hat\tau=\tilde{\Omega}\hh\hh\cdot\tau$, where $\tilde{\Omega}\in C^\infty(M)$ is a positive homogeneous function on the ambient space $M$.

 For that, we note that the isomorphism in Theorem \ref{thm:equivalence} relies on the associated ambient $1$-form $Y_\tau$, defining the component $v^-$. Let us first study $v^-$, and thus $Y_\tau$, under rescaling. Let
 $$V=\tilde v^+ Y^\sharp + \tilde v + \tilde v^- X+A(X) U$$
 be an ambient contractor, $\tau$ an ambient scale and $\hat\tau=\tilde{\Omega}\hh\hh\cdot\tau$ a rescaling as above. Then the transformation of $Y_\tau$ under rescaling reads
\begin{multline*}
\qquad Y_{\hat \tau}(V) = Y_{ \tau}(V)+
\big((d^A+\tfrac12 X^\flat \Lie_R) \log\tilde\Omega\big)(V)\\
=Y_{ \tau}(V)
 + (d\log\tilde\Omega)(\tilde v)+\big((d\log\tilde\Omega)(Y^\sharp)-\tfrac 12 \Lie_R \log \tilde\Omega\big)\tilde v^+ + {\mathcal O}(A(X))\, .
 \qquad\qquad
 \end{multline*}
In order to obtain the transformation law for $v^-$ of Definition \ref{def:contractor}, we must restrict to the cone $\Cone$. This eliminates the ${\mathcal O}(A(X))$ factor and the first summand $Y_{ \tau}(V)$ is naturally identified with $\tilde v^-$. In addition, the term $(d\log\tilde\Omega)(\tilde v)$ restricts to $(d\log\Omega)(v)$ as desired, and it suffices to study the term
$(d\log\tilde\Omega)(Y^\sharp)-\tfrac 12 \Lie_R \log \tilde\Omega$.
For that, we compute along the cone as follows:
$$(d\log \tilde \Omega)(Y^\sharp)=
(d^A\log \tilde \Omega)(Y^\sharp)
=\tfrac12 (d^A\log \tilde \Omega)(R-
\tilde\rho  )
=-\tfrac12 \iota_{\tilde\rho}\, d^A\log \tilde \Omega
=
- \tfrac12 \Lie_{\tilde\rho} \log \tilde \Omega
+\tfrac12 \Lie_R\log\tilde\Omega\, ,
$$
where the vector field $\tilde\rho$ is defined as
$$\tilde{\rho} = \lambda^{-2}(x_\lambda\circ \sigma)_*\rho,$$
with $\rho$ being the Reeb vector field associated to the scale $\tau$ and $\sigma \in \Gamma(\mathcal C)$ denotes the section also associated with $\tau$. In the above equality, we have used $\tilde\rho = R-2Y^\sharp$, which can be readily verified. In conclusion, the term $\tfrac12 \Lie_R\log\tilde\Omega$ in $(d\log \tilde \Omega)(Y^\sharp)$ disappears when inserted in $(d\log\tilde\Omega)(Y^\sharp)-\tfrac 12 \Lie_R \log \tilde\Omega$, and we are left with the coefficient
$$- \tfrac12 \Lie_{\tilde\rho} \log \tilde \Omega$$
for the $\tilde v^+$ component. Thus, at this point, the rescaling reads
$$
\hat v^-=v^--\frac12 v^+ \tau^{-2} \, {\mathcal L}_{\sigma_*\rho^\alpha} \log\Omega +(d\log\Omega)(v)\, , 
$$
and it suffices to notice that $\tilde\rho$ is related to $\rho$ by
$$
\tilde{\rho} = \lambda^{-2}(x_\lambda \circ \sigma)_*(\rho^\alpha)=\tilde \tau^{-2}(x_\lambda\circ \sigma)_*(\rho^\alpha)\, .$$ Hence, in terms of densities along $Z$, it follows that
\begin{equation}\label{botslot}
\Gamma(\ce Z[-1])\ni\bm{\hat  v}^-= \bm v^--\frac12 \bm v^+ \bm \tau^{-2} \, {\mathcal L}_{\rho^\alpha} \log\Omega +\bm \Upsilon(\bm v)\, ,
\end{equation}
where the weight $-2$ density $\bm \tau^{-2} \, {\mathcal L}_{\rho^\alpha} \log\Omega$
is precisely the $\bm \chi$ introduced in Definition~\ref{def:contractor}. This concludes the proof for the $v^-$-component. It remains to argue for the $v$-component, which we do now.

For that, first observe that for any ambient contractors $U$ and $V$, the ambient function
$
\Phi(U,V)
$
is homogeneous. It therefore defines a homogeneous function along $\Cone$ and thus a weight zero density along $Z$. Inserting decomposition~\nn{decompose} in $\Phi(U,V)$ for a choice of ambient scale $\tilde \tau$
corresponding to~$\alpha\in \bm \alpha$
gives  
$$
\Phi(U,V)=\tilde u^+_{\tilde \tau}\tilde v^-_{\tilde \tau} -\tilde v^+_{\tilde \tau}\tilde u^-_{\tilde \tau} +\Phi(\tilde u_{\tilde \tau},\tilde v_{\tilde \tau}) \label{phi(u,v)}
.$$
Now we remember that, restricted to $\Cone$, the triplet $(\tilde v^+_{\tilde \tau},\tilde v,\tilde v^0)$
corresponds to a section $(\bm v^+,\bm v,\bm v^-)$ of $\ce Z[1]\oplus \xi[-1]\oplus \ce Z[-1]$. Moreover, using arguments similar to above, it can be verified that the homogeneity zero function $\Phi(\tilde u_{\tilde \tau},\tilde v_{\tilde \tau})$ corresponds to $\bm \varphi(\bm u,\bm v)$ introduced in Equation~\nn{JUV}. The quantity $\Phi(U,V)$ is independent of any choice of $\alpha$. Thus, denoting by hats the triple of sections obtained for another $\hat \alpha \in \bm \alpha$, we have
$$
\hat{\bm u}^+
\hat{\bm v}^-
-\hat{\bm v}^+
\hat{\bm u}^-
+\bm \varphi( \hat {\bm u},\hat {\bm v})=
{\bm u}^+
{\bm v}^-
-{\bm v}^+
{\bm u}^-
+\bm \varphi( {\bm u}, {\bm v})\, .
$$
Equation~\nn{topmiddlebottom}
shows that $\hat {\bm v}^+=\bm v^+$ and $\hat {\bm u}^+=\bm u^+$, since these are defined independently from $\bm \tau$. Now Equation~\nn{botslot} gives $\hat {\bm v}^-$, and so also $\hat{\bm u}^-$, in terms of their unhatted counterparts. Inserting these results in the above display and then making the ansatz
$$
\hat{\bm v} = \bm a \bm v^++b \bm v + c \bm v^{-}\, ,
$$
(and similarly for $\hat {\bm u}$
for some $\bm a\in \Gamma(\xi[-2])$, $ b\in C^\infty Z$, $ c\in  \Gamma(\xi)$, gives a system of equations. These can be uniquely solved for the unknowns $(\bm a,b,c)$, and the solution is uniquely given by $(-\bm \Upsilon^\sharp,1,0)$ 
where $\bm \Upsilon^\sharp=[\alpha,(d^\alpha \log \Omega)^\sharp]$.
In summary we have found
$$
\hat {\bm v}^+=\bm v^+\, ,
$$
$$
\hat {\bm v} = \bm v - \bm \Upsilon^\sharp\, \bm v^+\, ,
$$
\vspace{-3mm}
$$
\bm{\hat  v}^-= 
\bm v^-+\bm \Upsilon(\bm v)-\tfrac12\bm v^+\bm \chi
\, .
$$
Comparing with Definition~\ref{def:contractor},
this establishes that the contractor and ambient tractor bundle are isomorphic. This concludes the proof of Theorem \ref{thm:equivalence}.\hfill$\square$

Note that the ambient vector field $X$ generating the $\R^+$-action $x$ lies in the (ambient) contact distribution $\Xi$ along the cone ${\mathcal C}$, 
 and it has homogeneity zero with respect to ${\mathcal L}_X$. Hence, it corresponds to a section of $\FZ[1]$.
It follows from Theorem~\ref{thm:equivalence},
that this is precisely the canonical tractor. Now suppose that $f\in C^\infty(M)$ has homogeneity $w$. 
Then the map $$
Df:=d f-\big(A-\frac12  X^\flat\big)  {\mathcal L}_R f\, ,
$$
obeys the tangentiality condition 
$$
D A(X)={\mathcal O}\big({ A}(X)\big)\, .
$$
This is the ambient analog of the contact $D$-operator defined in Equation~\nn{DDef}. Note that in these terms $Y_\tau = D\log\tau$.

\begin{example}[Continued from Example~\ref{getmaxy}] Consider the functions $$f=e^{w\mu} F(z)\mbox{ and } \sigma = e^\mu\, ,$$
	for some $F\in C^\infty(Z)$. We have ${\mathcal L}_X f = wf$ and ${\mathcal L}_X \sigma = \sigma$, and also $X^\flat=2e^{2\mu} \alpha$
and
$R=\frac{\partial}{\partial \theta}$. 
If we denote $Y:=d\sigma=e^\mu d\mu$, it follows that 
$$
D f= e^{w\mu}\big (w F
d\mu +
dF\big)
=
e^{w\mu}\big(
we^{-\mu} YF 
+\bar d F
+ \frac12 e^{-2\mu} X^\flat {\mathcal L}_\rho F\big)\, ,
$$
where $\bar d F := d F -\alpha {\mathcal L}_\rho F$ and $\rho$ is the Reeb vector of $\alpha$.
Using that the choice of ambient scale~$\tau = \sigma$ determines the contact 1-form $\alpha$, and Theorem~\ref{thm:equivalence}, we see that the above matches the contact $D$-operator given in Equation~\nn{DDef}.\hfill$\blacksquare$
\end{example}

\subsection{Ambient Connections}\label{ssec:ambientconnection} We have  now developed the contractor bundle $\FZ\longrightarrow Z$, the corresponding ambient construction $\Xi/\!\sim$, and established their equivalence in Theorem~\ref{thm:equivalence}. In Subsection~\ref{ssec:ContractConnections} we additionally introduced linear contractor connections $\nabla^\FZ$ for the contractor bundle~$\FZ$. In this subsection we  develop the ambient construction for contractor connections. Interestingly, the ambient contractor connections we construct will intertwine projective geometry---the parabolic geometry associated to $\mbox{SL}(n,\R)$---with contact geometry.

Let us first introduce the central object of interest in this subsection.

\begin{definition}\label{def:ambientconnection}
Let $(M,A,x)$ be an ambient contact manifold. An ambient contractor connection is a connection $\bar \nabla:\Xi_{}\longrightarrow\Xi_{}\otimes T^*\!M$ on the distribution $\Xi$ such that 
\begin{enumerate}[(i)]
	\item \label{FIRST}$\bar \nabla X=\Id - R\otimes A + {\mathcal O}(A(X))$,
	\item \label{SECOND} $\bar\nabla \Phi={\mathcal O}(A(X))$, 
\end{enumerate}
and such that, for any representative ambient contractor $V$, the following conditions are satisfied:
\begin{enumerate}[(a)]
	\item \label{thesecondone} $\bar \nabla \big(V+\order(A(X))\big)-\bar \nabla V
	=\order(A(X))\, ,$
	\item  $\bar \nabla_U V\in \Gamma(\Xi)$,
	\item $\Lie_X (\bar \nabla_U V)=-\bar \nabla_U V+\order(A(X))$,
\end{enumerate}
\noindent where $U\in\Gamma(TM)$ is any vector field subject to the condition $\Lie_U(A(X))=\order(A(X))=\Lie_X U$.\hfill$\blacksquare$
\end{definition}

The last three conditions in Definition \ref{def:ambientconnection} ensure that, along vectors tangent to the cone $\Cone$ which descend to vectors along~$Z$, the connection $\bar \nabla$ is a map from ambient contractors to ambient contractors. The first two conditions  guarantee that ambient contractor connections are contractor connections, as defined in Subsection \ref{ssec:ContractConnections}.

\noindent The main result in this subsection is the existence of an ambient contractor connection as introduced in Definition \ref{def:ambientconnection}. We construct this ambient contractor connection by intertwining contact geometry with projective structures. 
We discuss these below, but already state the  precise relation:

\begin{theorem}\label{existenceambientconnection}
Let $(M,A,x)$ be an ambient contact manifold and $\widetilde \nabla$ an ambient projective connection $($see Definition \ref{def:ambientprojective}$)$ that preserves the Levi 2-form along $\Xi$, {\it i.e.} $\widetilde\nabla \Phi= \order(A(X))$. Then, the connection $\bar \nabla$, defined for $U\in \Gamma(TM)$ by
	$$
	\bar\nabla_{U}=\widetilde \nabla_{U} - (R\otimes A)\circ \widetilde \nabla_{U}\, ,
	$$
is an ambient contractor connection.\hfill$\blacksquare$
\end{theorem}

In short, the strategy is to use the existence of tractor connections for projective structures in order to build ambient contractor connections. The proof of Theorem \ref{existenceambientconnection} will be provided in the next subsection, as we first need to develop the projective structure ingredients, which we now do.

\noindent Given an ambient contact manifold $(M,A,x)$, we consider the volume form $\Vol_A$ associated to the contact form $A$,
$$
\tilde \mu = A \wedge \Phi^{\wedge n}=:\Vol_A\, .
$$
The homogeneity condition $\Lie_X A=2A$ used in Definition \ref{def:ambientcontractor} implies
$$
\Lie_X \Vol_A = 2(n+1) \Vol_A\, .
$$
Thus, starting in the contact geometric setting, it is not directly possible to view the quotient of~$M$ by the $\R^+$-action as a projective manifold and $M$ as its projective ambient space. Indeed, according to Definition \ref{def:ambientprojective}), the homogeneity weight $2(n+1)$ of $\Vol_A$ with respect to $X$ should have been $\dim(M)=2n+3$. The solution is to construct a new $\R^+$-action correcting this. 

For that, we first observe that
\begin{equation}\label{recall}
\Lie_{X-\frac12 A(X) R}  A= 2A -\frac12 A(X) \Lie_R A -\frac12 \big(d(A(X))\big) \iota_R A
= A+\frac12 X^\flat\, ,
\end{equation}
while we also have the following Lie derivative:
$$
\Lie_{X-\frac12 A(X) R}  \Phi
=\Phi+\frac12  d X^\flat = 2\Phi\, . 
$$
The two equalities above suggest that the infinitesimal generator $X$ for the original $\R^+$-action ought to be corrected to $X':=X-\frac12 A(X) R$. Indeed, the following result shows that the vector field $X'$ is the desired generator:

\begin{lemma}\label{lem:Xprime} The vector field $X':=X-\frac12 A(X) R$ obeys $$\Lie_{X'}\Vol_A = (2n+3) \Vol_A.$$
In addition, $[X,X']=0.$
\end{lemma}

\begin{proof}
The homogeneity condition follows directly from the two computations above. Let us show that $X$ and $X'$ commute. For that, first note that
	$$
	0=\Lie_X 1 = \Lie_X (A(R))=2 A(R)+ A(\Lie_X R)\, ,
	$$
	and thus $A(\Lie_X R)=-2$. Given that
	$$
	0=\Lie_X (\Phi(R,\pdot))
	=2 \Phi(R,\pdot) +\Phi(\Lie_X R,\pdot) =\Phi(\Lie_X R,\pdot)\, , 
	$$
	we conclude that $\Lie_X R \propto R$ and hence
	$$
	\Lie_X R = -2 R\, .
	$$
	Now we can use
	$
	\Lie_X(A(X))=2A(X)
	$
	and obtain
	$$
	\Lie_X X'=\frac12 R\,  \Lie_X(A(X)) +\frac12 A(X)\, \Lie_X R = 0\, .
	$$
\end{proof}

\noindent By construction, the flow of $X'$ is given by the composition of the flows of $X$ and $\frac12 A(X) R$ and thus coincides with the $\R^+$-action $x$ along a neighborhood of the cone $\Cone$, as required. The first identity in Lemma \ref{lem:Xprime} shows that the quotient by the $X'$-action of 
an open neighborhood $\Op(\Cone)$ of the cone $\Cone$ 
is
an $\R^+$-principal bundle over the projective manifold $S:=\Op(\Cone)/\!\sim_{X'}$.

\noindent 
We now consider $\Op(\Cone)/\!\!\sim_{X'}$ as a projective manifold,
and thus may analyze this as a
$G=\SL(2n+3,\R)$ parabolic structure~\cite{Thomas}. Let us remind the reader that a $d$-dimensional projective manifold~$({S}^d, \bm \nabla)$ consists, by definition, of a smooth $d$-dimensional manifold $S$ and an equivalence class of torsion-free affine connections $\bm \nabla$, where the equivalence relation is given by $\nabla\sim {\nabla}'$ if and only if~$\nabla,{\nabla}'\in \bm \nabla$ determine the same unparameterized geodesics on $\mathcal{S}$. (Note that we 
use the symbol~$\bm \nabla^\hbar$ 
to denote the quantization of a contact structure, which ought not be confused with a projective class of connections.)
The ambient construction of projective structures starts with the following:

\begin{definition}\label{def:ambientprojective}
A projective ambient space for a projective $d$-dimensional manifold $({S},\bm \nabla)$ consists of a triple $(M^{d+1},\pi,\tilde \mu)$ where $\pi:M\longrightarrow {S}$ is an $\R^+$-principal bundle over ${S}$, and $\tilde \mu\in \Omega^{d+1} M$ is a positive volume form satisfying
	$$
	{\mathcal L}_{X'} \tilde \mu = (d+1) \hh \tilde \mu\, ,
	$$
	where $X'$ is the vector field generating the $\R^+$-action on $M$.\hfill$\blacksquare$
\end{definition}

In line with Definition \ref{def:ambientprojective}, the {\it ambient projective tractor bundle} $\mathfrak{P}$ is defined as the tangent bundle~$TM$ to the projective ambient space quotiented by the $\R^+$-action $x'_\lambda$, generated by $X'$, according to 
$$
V_P\sim \big(\lambda^{-1} dx'_\lambda(V)\big)_{x_\lambda(P)}\, ,
$$
for each $P\in M$. The {\it canonical protractor} corresponds to the ambient vector field $X'\in \Gamma(TM)$ that generates the $\R^+$-action. The ambient projective tractor bundle $\mathfrak{P}$ carries an {\it ambient projective connection}, which by definition is a torsion-free affine connection $\tilde \nabla$ subject to
$$\tilde \nabla X' = \Id\, ,\qquad \tilde \nabla \tilde \mu = 0\, ,$$
and a {\it Killing vector condition}
\begin{equation}\label{Kill}
{\mathcal L}_{X'} \tilde \nabla V = 
\tilde \nabla \, {\mathcal L}_{X'} V\, ,
\end{equation}
for any $V\in \Gamma(TM)$. The latter condition holds iff
$$
\tilde R(X',\cdot) V = 0\, ,
$$
where $\tilde R$ is the Riemann curvature of $\tilde \nabla$, see {\it e.g.}~\cite{Harrison}. The existence of an ambient projective connection is established in~\cite{Thomas}.

Let us now consider the ambient projective space to be the $(2n+3)$-dimensional ambient contact manifold $(M,A)$ endowed with the vector field $X'$, which we view as the infinitesimal generator of an $\R^+$-action.

\noindent Note that the projection of a weight $w$ ambient projective tractor to the distribution $\Xi|_{\Cone}$, given by~$\Xi=\ker A$ along the cone, yields a weight $w$ ambient contractor. Indeed, the 
conditions
$$\Lie_{X'} V=(w-1) V
\:\mbox{ and }\: A(V)={\mathcal O}(A(X))$$
together imply
$$
\Lie_{X} V = (w-1) V +\tfrac12 \Lie_{A(X) R} V
=(w-1) V+
\tfrac12 \big(d(A(X))\big) \iota_R V
+{\mathcal O}(A(X))
=(w-1) V+{\mathcal O}(A(X))\, .
$$
These ingredients
yield the hypothesis of Theorem~\ref{existenceambientconnection}, and we can now proceed to its proof.

\subsection{Proof of Theorem \ref{existenceambientconnection}}\label{proofy}
	Given an ambient contractor $V$, we have for any $U,B\in \Gamma(TM)$ with $\Lie_U (A(X))={\mathcal O}\big(A(X)\big)$ that  
	\begin{equation}\label{project}
	\bar \nabla_U(V + A(X) B)-\bar \nabla_U V
	=
	(B-RA(B))\, \iota_U d(A(X))
	+\order(A(X))
	=\order(A(X))\, ,
	\end{equation}
	because $\iota_U d(A(X))=\Lie_U (A(X))$.  
	This establishes Condition~(\ref{thesecondone}) in  Definition~\ref{def:ambientconnection}.
	Note that $\bar \nabla_U V\in \Gamma(\Xi)$ because $\Id - R\otimes A$ is a projector; this establishes (b).
	Alas, Condition (c)
requires some	computation:
	$$
	\Lie_X (\bar \nabla_U V) 
	=\Lie_X (\tilde \nabla_U V - R A(\tilde \nabla_U V))
	=(\Id -R\otimes A)\big(\Lie_X(\tilde \nabla_U V)\big)
	\, .
	$$
	Here we used that $\Lie_X A -2A = 0 = \Lie_X R+2R$.
	Now
	$$
	\Lie_X(\tilde \nabla_U V)=
	\tilde \nabla_X \tilde \nabla_U V - \tilde\nabla_{\tilde \nabla_U V}
	X\, ,
	$$
	because $\tilde\nabla$ is torsion free.
	Thus
	$$
	\Lie_X(\tilde \nabla_U V)=\tilde R(X,U) V +\tilde \nabla_{[X,U]} V+\tilde\nabla_U \tilde \nabla_X V
	-\tilde\nabla_{\tilde \nabla_U V}(
	X'+\tfrac12 A(X) R)\, ,
	$$
	where $\tilde R$ is the Riemann tensor of $\tilde\nabla$. 

By using the Killing condition~\nn{Kill}, it follows that $0=\tilde R(X',\pdot) 
	=\tilde R(X,\pdot)+\order(A(X))$. Note as well that
	$$\tilde \nabla_U \tilde \nabla_X V
	= \tilde\nabla_U (\Lie_X V + \tilde \nabla_V X)=-\tilde\nabla_U V
	+\tilde \nabla_U V +\frac12 \tilde\nabla_U 
	\big(A(X)\tilde\nabla_R V\big)=\order(A(X))\, ,
	$$
	because $\Lie_U (A(X))=\order(A(X))$.
	Recalling that $\Lie_X U=\order(A(X))$ and $\tilde \nabla X'=\Id$, we then have
	$$
	\Lie_X(\tilde \nabla_U V)=-\tilde\nabla_U V 
	-\frac12 R\, \tilde\nabla_{\tilde \nabla_U V}
	A(X)
	+\order(A(X))\, .
	$$
	Applying the projector $\Id - R\otimes A$ to the above expression and inserting this in Equation~\nn{project} yields
	$$
	\Lie_X (\bar \nabla_U V) = -(\Id - R\otimes A)(\tilde \nabla_U V)+\order(A(X))=-\bar \nabla_U V+\order(A(X))\, .
	$$
	This establishes that $\bar\nabla$ maps ambient contractors to ambient contractors. 
	
	It remains to verify the first two Conditions~(\ref{FIRST},\ref{SECOND}), of Definition~\ref{def:ambientconnection}. 
First, we compute
	\begin{multline*}
	\bar \nabla X =(\Id - R\otimes A) \tilde \nabla_U (X'+\tfrac12 A(X) R)=(\Id - R\otimes A)\big(U+\tfrac 12 R \Lie_U (A(X))\big)+\order(A(X))\\[1mm]=(\Id - R\otimes A)U +\order(A(X))\, ,
	\end{multline*}
	which verifies Condition~(\ref{FIRST}).
	Finally, that the 2-form $\Phi$ is parallel with respect to $\bar \nabla$ along $\Cone$
	follows directly from $\tilde \nabla \tilde\mu=\order(A(X))$.\hfill$\Box$

\section{Dynamical Quantization of $(Z,\xi)$}\label{sec:Quant}
\newcommand{\Sp}{\operatorname{Sp}}
\newcommand{\cev}[1]{\reflectbox{\ensuremath{\vec{\reflectbox{\ensuremath{#1}}}}}}

The quantization of a strict contact structure $(Z,\alpha)$ was given in Section~\ref{sec:StrictQuant}, while a  contractor calculus for handling the associated contact manifold $(Z,\xi)=(Z,\ker{\alpha})$ was developed in Section~\ref{sec:Contract}. We now meld these ingredients to present the dynamical contact quantization of $(Z,\xi)$.

\subsection{Hilbert Contractor Bundle}
\label{HBC}
 Let $(Z,\xi)$ be a contact manifold and $\FZ\longrightarrow Z$ its contractor bundle, as constructed in Section \ref{ssec:contractorbundle}. Let ~${\mathscr P}$ be the  $P$-principal bundle corresponding to $\FZ$, where~$P$ is the parabolic subgroup of $\operatorname{Sp}(2n+2)$ stabilizing a ray in $(\R^{2n+2},\omega_\st)$. In particular, the contractor bundle $\FZ$ is an associated vector bundle for this $P$-principal bundle. Given a Hilbert space $\mathcal{H}$,  we consider the metaplectic projective representation $$\mathcal{M}:\mbox{Sp}(2n+2,\mathbb{R})\lr\Aut(\P(\mathcal{H})),$$
of $\mbox{Sp}(2n+2,\mathbb{R})$ acting on the projective Hilbert space $\P(\mathcal{H})$. The metaplectic representation is discussed in the following section. For now, we use it to provide the following definition.

\begin{definition}\label{def:Hilbertcontract}
Let $(Z,\xi)$ be a contact manifold, $\FZ\rightarrow Z$ a standard contractor bundle and $\mathcal H$ a Hilbert space equipped with a metaplectic projective representation of  $\mbox{Sp}(2n+2,\mathbb{R})$. The {\it Hilbert Contractor Bundle} $\mathcal{H}(\FZ)$ of $\FZ\rightarrow Z$ is the vector bundle
${\mathscr P}\times_{ P}\P(\mathcal{H})$
associated to the metaplectic  projective representation $\mathcal{M}:\mbox{Sp}(2n+2,\mathbb{R})\rightarrow\Aut(\P(\mathcal{H}))$. \phantom{junk}
\hfill$\blacksquare$
\end{definition}

In fact, for our purposes it is convenient to use a different but equivalent vector bundle $\SH(\FZ)$ defined by considering the bundle of orthonormal frames $\mathscr P$ on $\FZ$ with respect to the Hilbert space inner product. Then $\SH(\FZ):={\mathscr P}\times_{U(\mathcal H)} \P({\mathcal H})$. We shall slightly abuse notation and also call this bundle the Hilbert contractor bundle.

On each fiber ${\mathbb P}({\mathcal H})$, the metaplectic representation comes equipped with a representation of the Heisenberg algebra ${\mathfrak{ sp}}(2n+2)$. Hence, 
for each standard contractor $U\in \Gamma(\FZ)$, the Hilbert contractor bundle comes with a map 
$$
s(U):\Gamma\big({\mathscr H}(\FZ)\big)\to\Gamma\big({\mathscr H}(\FZ)\big)\, ,
$$
where
$$
s(U)\circ s(V) -s(V)\circ s(U)=J(U,V)\, .
$$
In the above, $V$ is also a standard contractor and $J$ denotes the contractor symplectic form. We recycle the name  {\it Heisenberg map} for the above map $s$, which is canonically defined by extending the fiberwise unitary representation of the Heisenberg algebra (discussed in the next section) to act on sections of~${\mathscr H}(\FZ)$. It can  be further extended to act on weighted sections of ${\mathscr H}(\FZ)[w]$ or take as inputs weighted contractors. In particular, the map 
$$
s^+:=s(X):\Gamma({\mathscr H})\to\Gamma({\mathscr H})
$$
for the Hilbert contractor bundle plays a {\it r\^ole} analogous to that of the canonical  tractor for the standard contractor bundle.
It is useful to develop a more explicit realization of the Hilbert contractor bundle akin to that given for the standard contractor bundle in Section~\ref{cbundles}.
To that end,
we next give an explicit realization of the metaplectic projective representation of $\mbox{Sp}(2n+2,\mathbb{R})$. 

\subsection{Metaplectic Representation}\label{sec:metaplectic} 
The material in this section is standard but included for completeness, good references include~\cite{Gelfand,Gosson}.
The metaplectic representation  is a unitary representation of the connected double covering $\operatorname{Mp}(2n):=\widetilde{\Sp}(2n)$ of $\operatorname{Sp}(2n)$, that can be constructed
in terms of  generalized Fourier transformations acting on $L^2({\mathbb R}^{n})$. (Strictly, this is defined on the space of Schwartz functions ${\mathcal S}(\mathbb R^{n})$ on   
$\mathbb R^{n}$, which are dense in $L^2({\mathbb R}^{n})$.) 

Let $\mathscr{F}\sse\operatorname{Sp}(2n)$ be the set of free symplectic matrices in the symplectic linear group given by
$$\mathscr{F}=\left\{\left( \begin{array}{cc}
A & B \\
C & D \end{array} \right)\in\operatorname{Sp}(2n): A,B,C,D\in\mbox{Mat}_{n}(\R),\det(B)\neq0\right\}\sse\operatorname{Sp}(2n)\,  .$$
This set of matrices generates $\operatorname{Sp}(2n)$; in particular any $\operatorname{Sp}(2n)$ matrix is expressible as the product of a pair of free symplectic matrices. Now consider the {\it generating function} given by the bilinear form $a_F:\R^{n}\times\R^{n}\to {\mathbb R}$ associated to a linear symplectomorphism $F\in\mathscr{F}$ according to
$$a_F(x,\tilde{x})=\frac{1}{2}(B^{-1}A\tilde{x})\cdot \tilde{x}+\frac{1}{2}(DB^{-1}x)\cdot x-(B^{-1}x)\cdot\tilde{x}.$$
The generating function $a_F$  allows us to define the generalized Fourier transform $\mathscr L_F$ induced by $F$
$$\mathscr L_F:L^2(\R^{n})\lr L^2(\R^{n}),\quad \mathscr L_F\psi(x)=\left(\frac{1}{2\pi}\right)^{\frac n2}i^{\nu_F}\sqrt{|\det B^{-1}|}\int_{\R^{n}} e^{ia_F(x,\tilde{x})}\psi(\tilde{x})d\tilde{x}\, .$$
The system of phases, encoded by $i^{\nu_F}$, in the definition of  $\mathscr L_F$ is
given in terms of the so-called 
Conley-Zehnder index $\nu_F$ modulo~$4$ of a path joining the free symplectic matrix $F$ to the identity~$\Id\in\operatorname{Sp}(2n)$. This is  given by $\nu_F=m-n$ where
$m$ is a choice of argument for the branch of $\det B^{-1}$.

Let $\mathscr{H}=L^2(\R^{n})$, then the formul\ae\ above define the metaplectic representation 
$$\mathscr{M}:{\mbox{Sp}}(2n,\mathbb{R})\lr\Aut(\P(\mathscr{H}))$$
by defining it on free symplectic matrices and extending by virtue of producing a group automorphism.
Happily, since we work with $\P(\mathscr{H})$, we no longer need 
keep track of the system of phases controlled by the Conley-Zehnder index.

\begin{remark}\label{rmk:JFourier}
The operator $\widehat J$  
associated to the standard almost complex structure $J$, with generating function $a_{J}(x,\tilde{x})=-x\cdot\tilde{x}$, is the standard Fourier transform. It is readily seen that the operator $\widehat R$ associated to a diagonal rescaling of scale $r\in\GL(n,\R)$ is 
$$\widehat R\psi(x)=i^m\sqrt{\det(r)}\cdot \psi(r\cdot x)\, ,$$
 where $m$ is the Maslov index associated to a choice of branching argument for $\det(r)$. \hfill$\blacksquare$
\end{remark}

\begin{example}\label{Iamthexample}
Consider $\operatorname{Sp}(2)$ and denote $$
T:=\begin{pmatrix}1&t\\0&1\end{pmatrix}\, ,\qquad S:=\begin{pmatrix}0&1\\-1&0\end{pmatrix}\quad\Rightarrow
\quad-STS=
\begin{pmatrix}1&0\\-t&1\end{pmatrix}\,.
$$
The matrices $S$ and $T$ 
are free symplectic, while $STS$ is not.  Denoting the vectors $(x,\tilde x)$ discussed above by some~$(y,\tilde y)\in {\mathbb R}^2$, then $a_S=-y\tilde y$ and $a_T=\frac1{2t}
(y-\tilde y)^2$. Thus, if $\psi(y)\in L^2({\mathbb R})$ is a Schwartz function,
$$
\mathscr{M}_S \psi(y)=\frac1{\sqrt{2\pi i} } \int_{\mathbb R}
e^{-iy \tilde y}\psi(\tilde y)d\tilde y\, ,\qquad
\mathscr{M}_T \psi(y)=\frac1{\sqrt{2\pi it}} \int_{\mathbb R}
e^{\frac{i}{2t}{(y -\tilde y)^2}}\psi(\tilde y)d\tilde y\, .
$$
Moreover, a simple calculation shows
$$
\mathscr{M}_{\scalebox{.5}{$\begin{pmatrix}1\!&\!0\\-t\!&\!1\end{pmatrix}$}}=
(\mathscr{M}_{S^{-1}}\circ \mathscr{M}_T \circ \mathscr{M}_S)\psi(y)=e^{ -\frac i2 t y^2} \psi(y)\, .
$$ 
Defining $s_{\pm}\in \operatorname{End}({\mathcal S}(\mathbb R))$ by
$$
s_{+}  \psi(y)=y \psi(y)\, \qquad 
s_{-} \psi(y)=\frac1i\frac{\partial}{\partial y} \psi(y) \, ,
$$
the linearization of  $\mathscr{M}_T$ and 
$(\mathscr{M}_{S^{-1}}\circ \mathscr{M}_T \circ \mathscr{M}_S)$ in $t$ yields the operators  
$$
\frac{d}{dt} (\mathscr{M}_S\circ \mathscr{M}_T \circ \mathscr{M}_S) \psi(y)|_{t=0}
=-\frac i2 (s_{+}){}^2\,  \psi(y)\,, \quad
\frac{d}{dt} \mathscr{M}_T \psi(y)|_{t=0}=\frac i2(s_{-}){}^2\,  \psi(y)\, . 
$$
The operators $m_{++}=s_+^2$, $m_{--}=s_-^2$ and $m_{+-}=s_+ s_-+s_-s_+$ obey 
$$
[m_{--}, m_{++}]= 4 y\frac{\partial}{\partial y} +2=2m_{+-}\, ,
$$
and thus
generate the Lie algebra $\mathfrak{sp}(2)$ acting on Schwartz functions.
\hfill$\blacksquare$\end{example}

The pair of operators $s_+$ and $s_-$ in Example \ref{Iamthexample} are particularly important, since they give a representation of the one-dimensional Heisenberg Lie algebra. In fact, by the Stone--von Neumann theorem, once we fix how the central element of the Heisenberg algebra acts, we obtain a unique unitary representation of the Heisenberg algebra. 
This gives a map
$
s(u)\in \Aut(\P(\mathscr{H}))
$
 which, for any~$u,v\in {\mathbb R}^{2n+2}$, obeys
$$
s(u)\circ s(v) -s(v)\circ s(u)=J(u,v)\, .
$$
This is a unitary representation of the Heisenberg algebra where the central element acts as the identity. Moreover, this interacts with the projective metaplectic representation according to 
\begin{equation}\label{fun2meta}
{\mathscr M}(g) \circ s(u)
\circ  {\mathscr M}(g^{-1})=
s(g\cdot u)\, .
\end{equation}
Here $g\cdot u$ denotes the fundamental representation of ${\rm Sp}(2n)$ on 
the symplectic vector space
$({\mathbb R}^{2n},J)$.

\smallskip

Armed with the projective metaplectic representation, we now give an  explicit description of the Hilbert contractor bundle $\mathscr{H}(\FZ)$ of Definition~\ref{def:Hilbertcontract}.

\subsection{Hilbert Contractor Bundle II}

Consider the standard contractor bundle $\FZ\lr (Z,\xi)$. By Definition~\ref{def:contractor}, the equivalence relation defining the bundle $\FZ$ is given by the linear transformation
\begin{equation}\label{transf}
\begin{pmatrix}
v^+ \\
v^a \\
v^-
\end{pmatrix}
\sim 
\left(
\begin{array}{c|c|c}
\Omega & {0} & 0 \\
\hline
0&\delta^a_b
&0\\
\hline
0 &0 & \Omega^{-1}
\end{array}
\right)
\left(
\begin{array}{c|c|c}
1 & {0} & 0 \\
\hline
0&M^b_c
&0\\
\hline
0 &0 & 1
\end{array}
\right)
\left(
\begin{array}{c|c|c}
1 & {0} & 0 \\
\hline
-{\Upsilon}^c &
\delta^c_d & 0 \\
\hline
-\frac{1}{2}\chi & {\Upsilon}_d &1
\end{array}
\right)
\begin{pmatrix}
v^+ \\
v^d \\
v^-
\end{pmatrix}\, .
\end{equation}
The above display deserves some explanation:. Since the densities $\bm v^{\pm}$ of Definition~\ref{def:contractor}
are defined as equivalence classes $[\alpha,v^{\pm}]=[\Omega^2\alpha ,\Omega^{\pm1} v^{\pm}]$, we have labeled them by the functions $v^\pm$.
We have also expanded the distribution-valued vector $v=v^a f_a$ in a basis $\{f_a\}$ for the distribution such that~$\varphi(f_a, f_b)=j_{ab}$.
Indices are raised and lowered with the symplectic bilinear form $j$ in the usual way.
Global existence of such a basis is not important for the discussion that follows. The ${\rm Sp}(2n)$-valued matrix~$M^b_c$ encodes equivalence under different basis choices. 
The function $\Omega$ encodes
 rescaling of the contact form $\alpha$ to 
$\Omega^2\alpha$ with $\Omega\in C^\infty(Z)$, $\Omega$ positive, where
$
d\log \Omega = \Upsilon + \alpha \chi
$, $\Upsilon(\rho^\alpha)=0$, and~$\Upsilon=\Upsilon_a e^a$,
for a dual basis  $e^a$ determined by the Reeb vector of $\alpha$.

The matrices in the above display are lower triangular symplectic matrices that belong to
 the parabolic group $P$ discussed  in Subsection~\ref{ssec:contractorbundle}. 
 Understanding how they are realized by the metaplectic representation is key to giving an explicit description of the Hilbert contractor bundle.
 Let us focus on a single fiber of ${\mathscr H}(\FZ)$, given by the
  Hilbert space  $\mathcal{H}=L^2(\R^{n+1})$,
 which comes naturally equipped with the metaplectic representation of ${\rm Sp}(2n+2)$. 
 Following the 5-grading 
 in Equation~\nn{5grading}, 
 we may label the generators of $\mathfrak {sp}(2n+2)$ by
 $
 \{m_{++},m_{+a},m_{ab},m_{+-},m_{-a},m_{--}\}\, ,
 $
where the grading is computed by counting the subscripts $\pm$ with weights $\pm 1$.
The three generators $\{m_{++},m_{+-},m_{--}\}$  
 span 
 the Lie algebra $\mathfrak{sp}(2)$.
 The Lie algebra of the parabolic subgroup $P$  is generated by the non-negatively graded elements $ 
 \{m_{++},m_{+a},m_{ab},m_{+-}\}$.
Now we write the $(n+1)$-dimensional Heisenberg Lie algebra as $\operatorname{span}\{s_+,s_a,s_-,1\}$, where
 $$
 [s_-,s_+]=1\,, \quad [s_a,s_b]=j_{ab}\, ,
 $$
 and we have denoted the generator of the center by $1$.
  The generators of the Lie algebra of the parabolic subgroup $P$ are given by bilinears in the enveloping algebra of this Heisenberg algebra according to
  $$
  m_{++}=(s_+)^2\,,\quad
  m_{+a}=s_+ s_a\, ,\quad
  m_{ab}=s_a s_b + s_b s_a\, ,\quad
  m_{+-}=s_+s_-+s_-s_+\, . 
  $$ 
The map $s_+$ on ${\mathcal H}=L^2({\mathbb R}^{n+1})$ extends to the map $s_+$ in ${\rm Aut}(\Gamma({\mathscr H}))$ discussed above; this recycling of notation should cause no confusion. 
   
   Let us denote elements of $\R^{n+1}$ by $(y,x)$ where $x\in {\mathbb R}^n$ and $y\in {\mathbb R}$. In an index notation, we write~$x=x^a$. We now  construct an explicit  realization of the metaplectic representation
   $$\mathscr{M}:\mbox{Sp}(2n+2,\mathbb{R})\lr\Aut\big({\mathbb P}(L^2({\mathbb R}^{n+1}))\big)$$
   in which the generator $m_{++}$ acts on functions of $(y,x)$ by multiplication by the variable $y^2$ and  the corresponding Heisenberg generator $s_+$ is multiplication by $y$. 
   From Example~\ref{Iamthexample}, the Heisenberg generator $s_-$ acts by differentiation with respect to $y$. 
    For that, we examine each of the parabolic-valued matrices appearing in the transformation rule in Equation~\nn{transf}. To begin with, we factor the third matrix appearing there as
$$
\left(
\begin{array}{c|c|c}
1 & 0 & 0 \\
\hline
-{\Upsilon}^a & \delta^a_b & 0 \\
\hline
-\frac{1}{2}\chi & \Upsilon_b & 1
\end{array}
\right)=
\left(
\begin{array}{c|c|c}
1 & 0 & 0 \\
\hline
0 & \delta^a_c & 0 \\
\hline
-\frac{1}{2}\chi & 0 & 1
\end{array}
\right)
\left(
\begin{array}{c|c|c}
1 & 0 & 0 \\
\hline
-{\Upsilon}^c & \delta^c_b & 0 \\
\hline
0 & {\Upsilon}_b & 1
\end{array}
\right).
$$
The two multiplicative symplectic factors on the right are not free symplectic matrices. Nevertheless, the symplectic shear $M_{++}$ in the first factor can be conjugated by the standard complex structure~$J$ to an upper triangular matrix. By Remark \ref{rmk:JFourier}, the operator $\widehat{J}$ associated to $J$ is the standard Fourier transform. Hence, for $[\Psi]\in {\mathbb P}(L^2({\mathbb R}^{n+1}))$, the metaplectic action associated to the first factor is: 

$$M_{++}=
\left(
\begin{array}{c|c|c}
1 & 0 & 0 \\
\hline
0 & \delta^a_b & 0 \\
\hline
-\frac{1}{2}\chi &0 & 1
\end{array}
\right)
\stackrel{\mathscr M}\longmapsto\left\{\Psi(y,x)\longmapsto e^{\frac i4\chi y^2}\Psi(y,x)\right\}\, .$$
As promised, the linearization  of $M_{++}$ is proportional to $m_{++}=y^2$.
We now
need to understand the metaplectic action of the matrix $M_+$ in the second factor above:
$$M_+=\left(
\begin{array}{c|c|c}
1 & 0 & 0 \\
\hline
-{\Upsilon}^a & \delta^a_b & 0 \\
\hline
0 & {\Upsilon}_b & 1
\end{array}
\right)\, .$$
For that, we consider a Lagrangian polarization $\R^{2n}=\R^n\oplus\R^n$ and denote by $\pi_1$ and $\pi_2$ the projections onto the two factors. In particular, $\pi_2({\Upsilon})=:\Upsilon_2^{\bar a}$ undergoes a phase shift under the metaplectic representation while $\pi_1({\Upsilon})=:\Upsilon^1_{\bar a}$  acts as a translation. Explicitly,  the image under   the first projection obeys:
$$M_+^1=\left(
\begin{array}{c|c|c|c}
1 & 0& {0} & 0\\
\hline
0 & \delta_{\bar a}^{\bar b} & 0 & 0 \\
\hline
-\frac{1}{2}\Upsilon^{\bar a}_1 & 0 & \delta_{\bar b}^{\bar a} & 0\\
\hline
0 & \frac{1}{2}\Upsilon^1_{\bar b}& 0  & 1
\end{array}
\right)
\stackrel{\mathscr M}\longmapsto\left\{ \Psi(y,x)\longmapsto \Psi\left(y,x+\frac{y}{2}\pi_1({\Upsilon})^\sharp\right)\right\}\,  ,$$
and the second projection is represented by:
$$M_+^2=
\left(
\begin{array}{c|c|c|c}
1 & 0& 0& 0\\
\hline
\frac{1}{2}\Upsilon_{\bar a}^2  & \delta_{\bar a}^{\bar b} & 0 & 0 \\
\hline
0 & 0 & \delta_{\bar b}^{\bar a} & 0\\
\hline
0 & 0&\frac{1}{2}\Upsilon_2^{\bar b}   & 1
\end{array}
\right)
\stackrel{\mathscr M}\longmapsto
\left\{ \Psi(y,x)\longmapsto e^{-i y(\pi_2({\Upsilon})^\sharp\cdot x)/2 }\Psi(y,x)\right\}\, .$$
Linearization of these formul\ae\  combined with the fact that $m_{+a}=s_+ s_a$,  gives results for the Heisenberg generators:
$$
\pi_1(s_a)=:s_{\bar a} = \frac{\partial }{\partial x^{\bar a}}
\, ,\quad
\pi_2(s_a)=:s^{\bar a} = x^{\bar a}\, .
$$
Noticing that
$$
M_+=M_+^1 (M_+^2)^2 M_+^1\, ,
$$
we have
$$
M_+\stackrel{\mathscr M}\longmapsto
\Big\{\Psi(y,x)\longmapsto e^{- i y(\pi_2({\Upsilon})^\sharp\cdot (x+y\pi_1({\Upsilon})^\sharp/2)) }\Psi\left(y,x+y\pi_1({\Upsilon})^\sharp\right)\Big\}\, .
$$
By  composition with the symplectic shear $M_{++}$, we obtain
$$
\left(
\begin{array}{c|c|c}
1 & 0 & 0 \\
\hline
-{\Upsilon}^a & \delta^a_b & 0 \\
\hline
-\frac{1}{2}\chi & \Upsilon_b & 1
\end{array}
\right)
\stackrel{\mathscr M}\longmapsto
\left\{\Psi(y,x)\longmapsto e^{- i y\left(\frac{\chi}{4}y+\pi_2({\Upsilon})^\sharp\cdot (x+y\pi_1({\Upsilon})^\sharp/2)\right) }\Psi\left(y,x+y\pi_1({\Upsilon})^\sharp\right)\right\}.$$

The central block of the second matrix in the transformation rule~\nn{transf} is
 ${\rm Sp}(2n)$-valued. Its metaplectic action is described in the previous section. Using Remark~\ref{rmk:JFourier}, we see that the first matrix in Equation~\nn{transf} 
 gives the action
 $$
 M_{+-}=\left(
\begin{array}{c|c|c}
\Omega & {0} & 0 \\
\hline
0&\delta^a_b
&0\\
\hline
0 &0 & \Omega^{-1}
\end{array}
\right)
\stackrel{\mathscr M}\longmapsto
\Big\{
\Psi(y,x)\mapsto \Psi(\Omega y,x)
\Big\}\, .
 $$

 Now, given a section $\Psi(z;y,x)$ of ${\mathscr H}(\FZ)$ corresponding to some choice of $\alpha\in  \bm \alpha$, and  where $z\in Z$ labels a point in base manifold, moving to a conformally related contact form $\Omega^2 \alpha$,
  we have the equivalence 
  $$
  \Psi(z;y,x)\sim 
  e^{- i \Omega y\left(\frac{\chi}{4}\Omega y+\pi_2({\Upsilon})^\sharp\cdot (x+\Omega y\pi_1({\Upsilon})^\sharp/2)\right) }\Psi\left(z;\Omega y,x+\Omega y\hh\hh\pi_1({\Upsilon})^\sharp\right)\, .
    $$
    In the above, both the function $\Omega\in C^\infty(Z)$ and covector $\Upsilon$ now depend on $z$.

  In order to conclude our contact quantization of $(Z,\xi)$, our next step is to introduce  the BRST charge. In Section \ref{sec:StrictQuant} we  constructed a BRST charge given by the connection form $-i\hbar d^{\hat A}$ for the Hilbert bundle associated to a strict contact manifold~$(Z,\alpha)$. The ambient contractor construction developed in Section~\ref{sec:Contract} allows us to employ this construction to make a connection  on the Hilbert contractor bundle $\mathscr{H}(\FZ)$.

\subsection{Contact Quantization} \label{contquant}
The quantization of a contact structure $(Z,\xi)$ can be achieved by performing an $\R^+$-equivariant strict quantization of a corresponding  ambient contact manifold $(M,A,x)$---see Definition~\ref{strictequi}---and then restricting the quantum connection form  to the symplectic cone~${\mathcal C}$ in a manner analogous to that used to obtain a contractor connection 
from an ambient one.

For that, first let $(Z,\xi)$ be a $(2n+1)$-dimensional contact manifold and $(M,A,x)$ an ambient contact manifold  for $(Z,\xi)$, and recall that  by Theorem \ref{thm:equivalence} the contractor bundle $\FZ\lr(Z,\xi)$ is isomorphic to the restriction of the ambient contractor bundle $\FZ\cong\Xi|_{\Cone}/\!\sim$ to the symplectization of $(Z,\xi)$, which corresponds to the symplectic cone $\Cone$. In addition, the ambient construction of Section \ref{sec:Contract}  produces by restriction,  a  contractor connection $(\FZ,\nabla^{\FZ})$ from an ambient contractor connection $(\Xi,\nabla^{\Xi})$.
In Section~\ref{ssec:ambientconnection} we explained how to obtain an ambient contractor connection from an ambient projective structure suitably compatible with the strict contact structure $(M,A)$.

\noindent An ambient tractor connection can also be extracted from an $\R^+$-equivariant quantum connection~$\nabla$ on $(M,A,x)$ as follows.
First, note that the quantum connection form $\nabla$ determines a connection $\nabla^\Xi$ on the distribution $\Xi$ of $(M,A)$ via Equation~\nn{extract}. At grade $-1$, Equation~\nn{eqgr} imposes
\begin{equation*}
s(\nabla^\Xi_X U)=s(U)\, ,
\end{equation*}
where $U$ has homogeneity zero, so that $[X,U]=0$. Thus, using Equation~\nn{torsionlike}, we have
$$
\hat \kappa(U)=
s\big(\nabla^\Xi_{U}\bar X+{\mathcal L}_{\bar X} U +A(X) \nabla_R^\Xi U
\big)
=
s(\nabla^\Xi_{\bar X} U
)
+A(X)\,  s\big(\nabla_R^\Xi U
-
[R,U]
\big)
\, ,
$$
where $\bar X = X- R A(X)$.
Hence, for any $U'$---regardless of its homogeneity, as there exists a local basis of $TM$ spanned by homogeneous vector fields---we have
\begin{equation}\label{nablaX}
\nabla^\Xi_{U'} \bar X =\overline{ U'}
+A(X)\big( \nabla^\Xi_R U'-\overline{[R,U']}\big)\, .
\end{equation}
Here a bar is used to denote the projection of a vector to the contact distribution~$\Xi$.
Moreover, flatness of $\nabla$ imposes that
$\{\nabla^{(0)},\hat \kappa(\cdot)\}=0$. Thus, it follows that 
$$
\nabla^\Xi \Phi=0\, ,
$$
where $\Phi=dA$. This establishes that $\nabla^\Xi$ obeys the first two requirements of Definition~\ref{def:ambientconnection}\,. The next  requirement follows using the same arguments employed in Section~\ref{proofy}, 
while the subsequent one is a consequence of the definition given in Equation~\nn{extract}.
For the last condition, let $U$ be a homogeneity zero vector and $V\in \Gamma(\Xi)$ have homogeneity $-1$. 
Then the grade zero part of Equation~\nn{eqgr} reads
$$
[\nabla^{(0)}_X,\nabla^{(0)}_U]=0\, .
$$
Thus, using the Jacobi identity and Equation~\nn{extract} we have
\begin{multline*}
0=\big[[\nabla_X^{(0)},\nabla_U^{(0)}],\hat \kappa(V)\big]
=-\big[
[\nabla_U^{(0)},\hat \kappa(V)], \nabla^{(0)}_X
\big]
-\big[
[\hat \kappa(V),\nabla_X^{(0)}], \nabla^{(0)}_U
\big]
\\
=\big[\nabla_X^{(0)},\hat \kappa(\nabla^\Xi_U V)\big]
+\big[\hat \kappa(\nabla_X^\Xi V),\hat \nabla_U^{(0)}\big]
=\hat \kappa(\nabla_X^\Xi \nabla^\Xi_U V + \nabla^\Xi_U  \nabla^\Xi_XV)\, .
\end{multline*}
But  $\nabla^\Xi$ obeys the torsion condition~\nn{torsionlike}, so it follows that
$$
\nabla_{\bar X}^\Xi V =  \nabla^\Xi_V \bar X +\overline{[\bar X, V]}=
{\mathcal O}\big(A(X)\big)\, .
$$
Here we used Equation~\nn{nablaX} and the homogeneity condition $[X,V]=-V$.  Hence, imposing the tangential  condition $\Lie_U (A(X))={\mathcal O}\big(A(X)\big)$, 
we have established that
$$
\nabla_X^\Xi \nabla^\Xi_U V 
={\mathcal O}\big(A(X)\big)\, .
$$
Finally, note that if $V\in \Gamma(\Xi)$, then 
$$
\Lie_X V = 
\nabla^\Xi_{\bar X} V - \nabla^\Xi_V \bar X
+ R A([X,V])
+\Lie_{R A(X)} V\, .
$$
Remembering that ${\mathcal L}_X R=-2R$, we can project both sides of the above equation along the distribution which gives
$$
{\mathcal L}_X V
=\nabla^\Xi_{\bar X} V -V
+{\mathcal O}\big(A(X)\big)\, .
$$

This shows that $\nabla^{(0)}$ determines an ambient contractor connection.
To summarize, we have just established the following result.
\begin{proposition} Given an ambient contact manifold, $(M,A,x)$ for a contact manifold $(Z,\xi)$, then 
an $\R^+$-equivariant  quantum connection form $\nabla$ determines a contractor connection $\nabla^\FZ$.
\end{proposition} 

The contractor bundle $\FZ$ corresponding to the ambient contact manifold $(M,A,x)$, together with a Hilbert space ${\mathcal H}$, determines a
contractor Hilbert bundle as an associated vector bundle as demanded in Definition~\ref{def:Hilbertcontract}\,.

\begin{remark}
	Ambient contractors are sections of the ambient distribution $\Xi$ subject to the homogeneity condition in Equation~\nn{homog} and the equivalence~\nn{supp}. The ambient connection $\nabla^{(0)}$ in the direction of $X$ plays the {\it r\^ole} of the Lie derivative ${\mathcal L}_X$ on sections of the Hilbert bundle~${\mathscr H}$.
\hfill$\blacksquare$
\end{remark}

Alternatively, and equivalently,
in addition to specifying a contractor connection,  an $\R^+$-equivariant quantum connection form $\nabla$ determines the 
contractor Hilbert bundle from the ambient Hilbert bundle.
This result is an analog of Theorem~\ref{thm:equivalence}
relating the ambient distribution restricted to the symplectization of $Z$ to the standard contractor bundle~$\FZ$.

\begin{theorem}\label{sympspin}
Let $\nabla$ be an $\R^+$-equivariant quantum connection form acting on sections of a Hilbert bundle  $\mathscr{H}$ over $M$. Then the section space $\Gamma(\mathscr{H}(\FZ))$ is isomorphic to $ \ker \nabla^{(0)}_X|_{\mathcal C}\subset \Gamma(\mathscr{H})$.
\end{theorem}

\begin{proof}[Proof of Theorem~\ref{sympspin}]
In Theorem~\ref{thm:equivalence} we established an isomorphism between 
contractors and sections of the ambient distribution $\Xi_{\mathcal C}/\!\sim$ along the  projectivized cone. These obey the homogeneity condition ${\mathcal L}_X U = -U + {\mathcal O}\big(A(X)\big)$. As shown above, this homogeneity condition can  equivalently be written in terms of the ambient connection $\nabla^\Xi$ as
$$
\nabla^\Xi_X U = {\mathcal O}\big(A(X)\big)\, .
$$
Now consider $\psi, \chi \in \Gamma\big({\mathcal H}(\FZ)\big)$ that obey
$$
\nabla_X^{(0)} \psi = {\mathcal O}\big(A(X)\big) = \nabla_X^{(0)} \chi\, .
$$
Then 
$${\mathcal L}_X \langle\chi| s(U)|\psi \rangle 
=\nabla_X^{(0)} \langle\chi| s(U)|\psi \rangle 
\in \Gamma(\Xi^*)= {\mathcal O}\big(A(X)\big)\, .$$
We already proved that ambient vectors $U$ in the kernel of $\nabla^\Xi_X$ along the cone define standard contractors which are sections of an associated vector bundle transforming under the fundamental representation of the parabolic $P\subset \Sp(2n+2)$. It thus suffices to establish that the same holds for $\psi$ and $\chi$, which is true because  Equation~\nn{fun2meta}  shows that a fundamental $P$-action  on  a vector $U$ induces a projective metaplectic transformation on $\psi$ and $\chi$.
\end{proof}

We have now gathered the ingredients necessary to establish Theorem~\ref{biggy}, relating $\R^+$-equivariant quantizations of ambient contact manifolds to the contact quantization of a contact manifold as laid out in Definition~\nn{QMcontdef}.

\begin{proof}
[Proof of Theorem~\ref{biggy}]
First, the symplectization of $(Z,\xi)$ 
is obtained 
from $(M,A,x)$ by restriction to the hypersurface given by the zero locus of $A(X)$, so in the following we identify these spaces. We first need to check that the  connection $\nabla^\hbar$  constructed for a strict quantization of $(M,A)$ pulls back to a symplectic quantization of $(Z,\xi)$. This follows directly from the definition of an ambient contact manifold, which in particular implies 
that $A\big|_{A(X)=0}=\lambda_\xi$, so indeed
$$\lim_{\hbar \to 0} i\hbar {\nabla}^\hbar \big|_{A(X)=0,\hbar = 0} =
\lambda_\xi\, .$$
For this to give a contact quantization of $(Z,\xi)$, we still need to verify the 
 $\R^+$-equivariance Condition~\nn{RPE}. This is a direct consequence  of the Definition~\ref{strictequi} for a strict quantization.
\end{proof}

In  Section~\ref{S3} we spell out this quantization procedure in detail for the example of the standard contact 3-sphere. The final task in this Section \ref{sec:Quant} is explaining how to recover a strict quantization from the quantization of a contact structure. This will rely on contractor connections that admit parallel contractors.

\subsection{Parallel Contractors and Reeb Quantization}\label{parallel} 

Suppose that $\bm \nabla$ is a contact quantum connection determining a contractor connection~$\nabla^\FZ$, and moreover, $\nabla^\FZ$ admits a parallel standard contractor $I\in \Gamma(\FZ)$, 
$$
\nabla^\FZ I = 0\, .
$$
Then, via Lemma~\ref{parascale}, this determines a scale $\bm \sigma=X_A I^A$. Let us suppose that $\bm \sigma$ is a true scale and therefore picks out a  distinguished contact form $\alpha=\bm\alpha/\bm\sigma\in \bm \alpha$. The question we want to address is the relationship between the contact quantization given by $\bm \nabla$ and the quantization of the Reeb dynamics of $\alpha$. In other words, we now explain how to reduce contact quantization to Reeb quantization.

For that, we first 
feed the parallel tractor~$I$ to  the Heisenberg map $s$, which gives a map 
$$s^-:=s(I):\Gamma\big({\mathscr H}(\FZ)\big)\to\Gamma\big({\mathscr H}(\FZ)\big)\, .$$ 
We would like to impose the constraint
$$
s^- \Psi = 0
$$
on sections $\Psi\in
{\mathscr H}(\FZ)$ to obtain a new Hilbert bundle on which the Reeb dynamics of the contact 1-form $\alpha=\bm\alpha/\bm \sigma$ are quantized. The scale $\bm \sigma$  determines a distinguished section of the symplectization~${\mathcal C}$ with contact form $\alpha$.  Indeed, we now have a  hypersurface $Z_I$ in ${\mathcal C}$ that is isomorphic to $Z$ and 
equipped with a contact 1-form $\alpha$.

Let $U$ be any vector tangent to $Z_I$. Given a contact quantization~$(\SH(\FZ),\bm\nabla^\hbar)$ of $\xi$ and underlying symplectic quantization $\nabla^\hbar$, 
 we require that the maps~$s^-$
 and $\nabla^\hbar_U$, restricted to the hypersurface,  commute
 \begin{equation}\label{commuter}
[\nabla^\hbar_U,s^-]=0\, .
\end{equation}
We then define a new Hilbert bundle $\SH_I$ over $Z_I$ as the subbundle of the restriction 
$\SH(\FZ)|_{Z_I}$ given by  imposing that the operator $s^-$ annihilates sections. The above-displayed condition ensures that we also get a quantum connection 
on $\SH_I$ this way.
Moreover, acting with $\big(\frac12\frac{d^2}{d\hbar^2}\circ \hbar^2 \cdot \big)\big|_{\hbar=0}$ extracts the grade zero piece of 
Condition~\nn{commuter} and this implies the parallel condition $
\nabla^\FZ I =0
$ introduced above.

\begin{remark}
Employing Equations~\nn{nablaV} and~\nn{scaletractor},  the parallel contractor condition 
implies that the tensors $P$ and $Q$ vanish  in the scale determined by $\bm \sigma$. This is a direct analog
of a result for conformal geometries which states that a conformal class of metrics admits an Einstein metric precisely when the corresponding tractor bundle admits a parallel scale tractor~\cite{BEG}. This gives a pleasing geometric analogy between general relativity  and quantum mechanics.\hfill$\blacksquare$
\end{remark}

This Section \ref{sec:Quant} is now concluded by working through the above construction explicitly in our running example.

\begin{example}[Continued from Example~\ref{getmaxy}]\label{maxy} Let $(Z,\alpha)$ be a strict contact structure and consider
its symplectization ${\mathcal C}\cong \R_\mu\times Z$; 
the coordinate along rays is~$\mu\in\R$, and the Liouville  form 
is
$$
\lambda^\xi = e^{2\mu}\alpha\, .
$$
This Liouville form is obtained by pulling back the 1-form $A$ in Equation~\nn{hereisA} to the hypersurface~${\mathcal C}$ where~$A(X)=2\theta$ vanishes. Also, pulled-back to~${\mathcal C}$, the frames $E^A$ in Equation~\nn{whoframedrogerrabbit} give an isomorphism between the tractor bundle $\FZ$ over $Z$ and tangent bundle $T{\mathcal C}$ of the symplectization  ${\mathcal C}$. 

\noindent Let us denote $\sigma=e^\mu$ and note that $d \sigma = \iota_I d \lambda^\xi$ where $I=-\frac1{2\sigma}\rho$. Thus $I^A=E^A(I)=(\sigma,0,0)\,$ and
$$d I^A + \Omega^A{}_B I^B = 
(0,-P^a, e^{-\mu}Q)\, .
$$
In conclusion, $I^A$ is parallel when $P^a=0=Q$. In turn,  the system of equations~\nn{flatter} imposes $d^\omega e^a = 0$. Expanding the Heisenberg map  in terms of the frames $E^A$ and then pulling-back to the hypersurface~$Z_I$ given by  $\sigma=1$ $($so that $\mu=0)$, we find
$$
E^A \hat S_A= e^a \hat S_a + 2 \alpha  \hat S^-\, ,
$$
where $\hat S^-=s(I)=E^A(I)\hat S_A = I^A \hat S_A$. Hence, the connection form in Equation~\nn{nab0} pulled-back first to the symplectization, then to the hypersurface $Z_I$, 
and acting on the kernel of $s(I)$, becomes
$$
\nabla^0|_{Z_I,\ker s(I)}= \frac{\alpha}{i\hbar} + \frac{e^a \hat S_a}{i\hbar} + d + \frac{1}{2i\hbar} \omega^{ab} \hat S_a \hat S_b\, .
$$
Note that there is no dependence on the pair of operators $\hat S_{\pm}$ in the above display. The condition~\nn{commuter} implies the same property for the higher order terms 
$\hat \Omega^1=\nabla - \nabla^0$ acting on the kernel of~$s(I)$, since the corresponding totally symmetric contractors $\Omega^{A_1\ldots A_k}$ defined by Equation~\nn{notme} must then obey
$$
I_A\Omega^{A A_2\ldots A_k}=0\, .
$$
Finally, flatness of the restriction $\nabla^\hbar:=\nabla|_{Z_I,\ker s(I)}$ is guaranteed by that of $\nabla$, and thus $\nabla^\hbar$  defines a strict quantization of~$(Z,\alpha)$, as required.
\hfill$\blacksquare$\end{example}

\section{Dynamical Quantization of the contact 3-sphere $(S^3,\xi_\st)$}\label{S3}
In this section,  we apply our methods to the quantization of the standard contact 3-sphere $(S^3,\xi_\st)$. A useful model for the three-sphere is a pair of solid tori identified along their common (Clifford) torus boundary. In coordinates, let $\theta_1,\theta_2\in [0,2\pi)$ be periodic coordinates for this torus $T^2$. Then, introducing $\psi\in [0,\frac \pi2]$, the standard embedding of the three sphere   in ${\mathbb R}^4$ is given by
$$
S^3\ni (\theta_1,\theta_2,\psi)\hookrightarrow (\cos\psi \cos\theta_1,\cos\psi \sin\theta_1, \sin\psi \cos\theta_2 ,\sin \psi \sin\theta_2)\in {\mathbb R}^4\, .
$$
The Clifford torus is at $\psi=\frac \pi4$. A contact form $\alpha$ for the unique tight contact structure $\xi_\st$ on $S^3$ is
$$
\lambda_i:= 2(\cos^2\psi d\theta_1 + \sin^2 \psi d\theta_2)\, ,
$$
and the corresponding Reeb vector field is $\rho=\frac{\partial}{\partial \theta_1}+\frac{\partial}{\partial \theta_2}\,$. Note that all Reeb orbits are periodic.

The three-sphere $S^3$ is parallelizable and, moreover, the codistribution is globally spanned by the pair of coframes $e_1,e_2$ given by
$$ \lambda_j:=e^1 = 2\cos(\theta_1 + \theta_2)d\psi + \sin2\psi\sin(\theta_1 + \theta_2)\left(d\theta_1 - d\theta_2 \right) ,
$$
$$
 \lambda_k:=e^2 = 2\sin(\theta_1 + \theta_2)\hh d\psi - \sin2\psi\cos(\theta_1 + \theta_2)\left(d\theta_1 - d\theta_2 \right),$$  
where $
d\alpha= e^1\wedge e^2$. Indeed, the (quaternionic) triplet of 1-forms
$
(\lambda_i,\lambda_j,\lambda_k)
$
obey
\begin{equation}
\label{foldback}
d\lambda_i = \lambda_j\wedge\lambda_k\, ,\qquad
d\lambda_j = \lambda_k\wedge\lambda_i\, ,\qquad
d\lambda_k = \lambda_i\wedge\lambda_j\, .
\end{equation}
Hence
$$
d\begin{pmatrix}e^1 \\ e^2\end{pmatrix}
+\begin{pmatrix}0& \alpha\\
-\alpha & 0 
\end{pmatrix}\wedge\begin{pmatrix}e^1 \\ e^2\end{pmatrix}=0\, .
$$
Thus, having solved $de^a + \omega^a{}_b \wedge e^b = 0$ with $\omega$ determined by the matrix displayed above,  we may search for a quantum connection form. Let ${\mathcal H}=L^2({\mathbb R})$ with Heisenberg algebra
$$
[\hat s_1,\hat s_2]=-i\hbar \, ,
$$
Note, that for this example, in order to connect to former results, we prefer to use standard physics conventions for the dependence on $\hbar$.

Now consider the quantum connection form
$$
\nabla = \frac\alpha{i\hbar}
+ \frac{e^1\hat s_1 + e^2 \hat s_2}
{i\hbar}
+ \nabla_0
+\hat a\, .
$$
Here $\nabla_0 = d - \frac\alpha{2i\hbar}  (\hat s_1^2+\hat s_2^2)$,
and we must compute $\hat a$ by requiring $\nabla^2=0$. For that, note that~$(\lambda_i,\lambda_j,\lambda_k)$ are a basis for $T^*S^3$ so we may expand
$$
\nabla = d + \frac1{i\hbar}\sum_{m\in\{i,j,k\}} \lambda_m \hat L_{m}\, ,
$$
where the operators $\hat L\in \operatorname{End}L^2({\mathbb R})$ are hermitean, and by virtue of Equation~\nn{foldback} must obey
$$
{i\hbar} \hat L_i
+[\hat L_j,\hat L_k]=0=
{i\hbar} \hat L_j
+[\hat L_k,\hat L_i]
={i\hbar} \hat L_k
+[\hat L_i,\hat L_j]\, .$$
In other words, we must find a representation of $\mathfrak{su}(2)$ on $L^2({\mathbb R})$ such that
$$
\hat L_i=1-\frac12 (\hat s_1^2+\hat s_2^2)+ \cdots\, ,\quad
\hat L_j=\hat s_1+\cdots
\, ,\quad
\hat L_k =\hat s_2 +\cdots\, .
$$
Here $\cdots$ denotes higher order terms in the grading, where $\hat s^1$, $\hat s^2$ carry grade one and $\hbar$ is grade two. To find this representation, we make the ansatz $\hat L_i=1-\frac12 (\hat s_1^2+\hat s_2^2)$ without further correction and then solve for the higher order terms in $\hat L_j$ and $\hat L_k$.

Before proceeding, it is interesting to note that the above problem amounts to starting with the Heisenberg algebra $[\hat s_1,\hat s_2]=i\hbar$, and then searching for corrections to the generators $\{\hat s_1, \hat s_2, 1\}$ that lift this to the $\mathfrak{su}(2)$ Lie algebra. Also, one can assure oneself of the likely existence of a solution by noting that given the Poisson bracket $\{q,p\}=1$, the three Hamiltonians $$1-\frac12(p^2+q^2)\, ,\quad q\, \sqrt{1-\frac14(p^2+q^2)}\, ,\quad p\, \sqrt{1-\frac14(p^2+q^2)}\, ,
$$ obey the $\mathfrak{su}(2)$ Lie algebra with respect to the bracket $\{\cdot\hh,\hh\cdot\}$.

Returning to the quantum solution, it is useful to introduce operators
$$
a:=(\hat s_2+i\hat s_1)/\surd 2
,\qquad a^\dagger:=(\hat s_2-i\hat s_1)/\surd 2$$ subject to $[a,a^\dagger]=\hbar$, as well as
$$\hat E:=(\hat L_k+i\hat L_j)/\surd 2,\qquad \hat E^\dagger:=(\hat L_k-i\hat L_j)/\surd 2\, .$$ Then, we must now solve
$$
[\hat L_i,\hat E]=\hbar \hat E\, ,\quad
[\hat L_i,\hat E^\dagger]= -\hbar \hat E^\dagger\, ,\quad
[\hat E,\hat E^\dagger]=\hbar \hat L_i\, ,
$$
where
$$
\hat L_i = 1 - \hat N - \frac\hbar 2\, ,
$$
with $\hat N:=a^\dagger a$. Inspired by the classical solution displayed above, we posit
$$\hat E= \sqrt{1+\frac {\hat N+\hbar}2} \, a\, ,
\quad
\hat E^\dagger= a^\dagger \, \sqrt{1+\frac {\hat N+\hbar}2}\, .$$
Indeed, the above operators obey the required $\mathfrak{su}(2)$ algebra. Even better, they are an example of a construction found in 1940 by Holstein and Primakoff. 

At this stage, we have determined a rather interesting one-parameter family of quantum dynamical systems on the three-sphere with its standard  Reeb dynamics. These are given by the connections
\begin{equation}\label{nablaS3strict}
\nabla^\hbar= d + \frac{\lambda_i }{i\hbar} \Big(1-\hat N-\frac\hbar2\Big)
+\frac{\lambda}{i\hbar} a^\dagger \, \sqrt{1-\frac {\hat N+\hbar}2}
+\frac{\bar\lambda}{i\hbar} \sqrt{1-\frac {\hat N+\hbar}2} \, a\, ,
\end{equation}
where $\lambda:=(\lambda_k+i\lambda_j)/\surd 2$.
In particular, while the limit $\hbar\to 0$ connects back to the underlying classical system, other values of $\hbar$ are interesting too.
Indeed, a realization of a phenomenon originally observed by Holstein and Primakoff occurs in this quantum mechanical system whenever
$$
1-\frac\hbar2 \in {\mathbb Z}_{\leq 0}\, . 
$$
Namely, while the operator $E$ always annihilates the Fock vacuum $|0\rangle$, because $a|0\rangle=0$, the number operator eigenstate $|\frac\hbar 2-1\rangle$ is annihilated by $E^\dagger$. Hence in that case, the Hilbert space truncates to the spin $\frac{\hbar-2} 4$ unitary irreducible representation of $\mathfrak{su}(2)$. This concludes the first part of the discussion regarding the dynamical quantization of the Reeb dynamics of the contact form $\lambda_i$ for $(S^3,\xi_\st)$.

\medskip

Let us now turn to the ambient quantization of $(S^3,\xi_\st)$. For that, let us consider Cartesian coordinates $(x,y,z,w,u)\in\R^5$ and an  ambient space
$(M,A,\rm x)$ given by $(\R^5,A)$, where the contact form is $A=\lambda_I+du$ with $\lambda_I:=2(xdy-ydx+zdw-wdz)$. Thus the Levi form reads
$$\Phi:= dA = 4\big(dx\wedge dy + dz\wedge dw\big).$$
The associated Reeb vector is $R=\frac\partial{\partial u}$,
and the $\R^+$-action reads
$$(x,y,z,w,u)\stackrel{\rm x}\longmapsto (\lambda x, \lambda y, \lambda z,\lambda w,\lambda^2 z).$$
It is generated by the vector field
$$
X=x\frac{\partial}{\partial x}+
y\frac{\partial}{\partial y}+
z\frac{\partial}{\partial z}+
w\frac{\partial}{\partial w}+
2u\frac{\partial}{\partial u},
$$
which satisfies ${\mathcal L}_X A = 2A$. Thus $X^\flat = \Phi(X,\cdot) = 2\lambda_I\, .$ Since $A(X)=2u$, the symplectic cone is the  $\{u=0\}\sse(\R^5,A)$ hyperplane. This is the space ${\mathbb R}^4$, with coordinates $(x,y,z,w)$, equipped with its standard $\R^+$-action and Liouville form. In other words, the symplectization of~$S^3$ with its standard contact structure and the unique Stein filling of $S^3$ added to the concave end. Indeed, the pullback of~$\lambda_I$ to the   embedded 3-sphere at~$r^2:=x^2+y^2+z^2+w^2=1$
is exactly the contact 1-form $\lambda_i$ defined above.

Next, we introduce  symplectic frames for the ambient codistribution $\Xi^*$. For that, first note that 
$$
\lambda_I:=2(xdy-ydx+zdw-wdz)\, ,\quad\!\!
\lambda_J:= 2(x dz - z dx -ydw+wdy)\, ,\quad\!\!
\lambda_K:= 2(x dw - w dx +ydz-zdy)\, ,
$$
obey the relations
$$
d\lambda_I = \frac{\lambda_J\wedge \lambda_K}{r^2}+2 Y\wedge \lambda_I\, ,\qquad
d\lambda_J = \frac{\lambda_K\wedge \lambda_I}{r^2}+2 Y\wedge \lambda_J\, ,\qquad
d\lambda_K = \frac{\lambda_I\wedge \lambda_J}{r^2}+2 Y\wedge \lambda_K\, .
$$
Hence $\Xi^*=\operatorname{span}\{E^+,E^1,E^2,E^-
\}$ where $E^+=X^\flat$, $E^1=\lambda_J/r$, $E^2=\lambda_K/r$, $E^-=Y$, and this is a symplectic framing because
$$
\Phi=E^-\wedge E^+ + E^1\wedge E^2\, .
$$
This allows us to introduce the connection $\nabla^\Xi$ subject to
\begin{equation}\label{Ideterminenabla}\nabla^\Xi E^A=(dE^A + \Omega^A{}_B E^B):=
d\begin{pmatrix}
E^+\\[1mm]E^1\\[1mm]E^2\\[1mm]E^-
\end{pmatrix}
+
\begin{pmatrix}
-Y&\frac{\lambda_K}r&
-\frac{\lambda_J}r&
2\lambda_I\\[1mm]
0&0&\frac{\lambda_I}{r^2}&\frac{\lambda_J}r\\[1mm]
0&-\frac{\lambda_I}{r^2}&0&\frac{\lambda_K}r
\\[1mm]
0&0&0&Y
\end{pmatrix}
\begin{pmatrix}
E^+\\[1mm]E^1\\[1mm]E^2\\[1mm]E^-
\end{pmatrix}
=0\, .
\end{equation}
Let $V\in \Gamma(\Xi)$ be a
vector in the distribution and denote $V^A:=E^A(V)$. Then the 
above display determines a connection on $\Xi$ according to
$$
\nabla^\Xi V^A := d V^A + \Omega^A{}_B V^B\, .
$$
For the special case $V=X$, the homothety $X^A=(0,0,0,1)$, we have $\nabla^\Xi X^A = E^A\,.$
In addition, if~$V$ has homogeneity $-1$, so that ${\mathcal L}_X V=-V$, then the functions $V^A=E^A(V)=(V^+,V^1,V^2,V^-)$  have homogeneities given by ${\mathcal L}_X V^+=V^+$, ${\mathcal L}_X V^{1,2}=0$, ${\mathcal L}_X V^-=-V^-$. By inspection---remember that $[d,{\mathcal L}_X]=0$,  ${\mathcal L}_X Y=0$, ${\mathcal L}_X r = r$ and ${\mathcal L}_X \lambda_{I,J,K}=2\lambda_{I,J,K}$---the functions~$\nabla^\Xi V^A$ 
have the same respective homogeneities as $V^A$, which proves that $\nabla^\Xi$ maps contractors to contractors and thus defines a contractor connection.

Let us proceed with the contact quantization of the ambient space $(M,A,x)$. Note that
the curvature $F^\Xi$ of the connection $\nabla^\Xi$ above ~is
$$
F^\Xi=
\big(\nabla^\Xi\big)^2=
\begin{pmatrix}
0&0&0&0\\
0&0&\frac{\lambda_J\wedge\lambda_K}{r^4}&0\\
0&-\frac{\lambda_J\wedge\lambda_K}{r^4}&0&0\\
0&0&0&0
\end{pmatrix}\, ,
$$
which obeys $\iota_X F^\Xi=0$. In order to achieve flatness, we now search for an $\R^+$-equivariant quantum connection form~$d^{\hat A}$ for $(M,A,x)$ where
$$
d^{\hat A} = \frac{A}{i\hbar} + \frac {E^A \hat S_A}{i\hbar} + d + \frac{\Omega^{AB} \hat S_A \hat S_B }{2i\hbar} + \cdots \, .
$$
Again the terms denoted $\cdots$ are solved for  by requiring flatness; these are certainly non-vanishing because $F^\Xi\neq 0$.
This connection form acts on a Hilbert bundle over $M$ whose associated Hilbert space operators $\hat S_A=(\hat S_+,\hat S_1,\hat S_2,\hat S_-)$ obey
$$
[\hat S_-,\hat S_+]=-i\hbar
\, ,\qquad 
[\hat S_1,\hat S_2]=-i\hbar\, .
$$
In the current setting, $T^*{\mathbb R}^5\backslash\{0\}$ is framed by the 1-forms $\{du, \lambda_I, \lambda_J,\lambda_K, Y\}$, and thus we can write 
\begin{multline*}
\nabla =  e^{-\frac{u}{i\hbar}}d\hh\circ\hh e^{\frac{u}{i\hbar}} + \frac{\lambda_I}{i\hbar} \Big((1+\hat S_+)^2 
-\frac{H}{r^2}\Big)  
+  \frac{\lambda_J\hat S_1}{i\hbar}  (1+\hat S_+)
+  \frac{\lambda_K\hat S_2}{i\hbar} \hat (1+\hat S_+)
\\
+
 \frac{Y}{i\hbar}\,  \Big(\hat S_-+\frac12 (\hat S_+ \hat S_-+\hat S_-\hat S_+)\Big)+\, \cdots\, .
\end{multline*}
The terms labeled $\cdots$ required for $\nabla$ to be nilpotent can be found by modifying the Holstein--Primakoff solution above. For that, let us denote
$$\Lambda:=(\lambda_K+i\lambda_J)/\surd 2,\qquad A:=(\hat S_2+i \hat S_1)/\surd 2,\qquad \hat H:=\frac12 (\hat S_1^2+\hat S_2^2)\, .$$ The desired connection is then
\begin{multline*}
\nabla:= e^{-\frac{u}{i\hbar}}\, d\hh\circ\hh e^{\frac{u}{i\hbar}} +
\frac{\lambda_I}{i\hbar} \Big((1+\hat S_+)^2 
-\frac{\hat H}{r^2}\Big) 
\\[1mm]+
\frac{\Lambda}{i\hbar r}\,  A^\dagger\,  \sqrt{(1+\hat S_+)^2 -\frac{\hat H+\frac\hbar 2}{2 r^2}}
+
\frac{\bar \Lambda}{i\hbar r}\,   \sqrt{(1+\hat S_+)^2 -\frac{\hat H+\frac\hbar 2}{2 r^2}}\: A
\\[1mm]
+
 \frac{Y}{i\hbar}\,  \Big(\hat S_-+\frac12 (\hat S_+ \hat S_-+\hat S_-\hat S_+)\Big)\, .
\end{multline*}
It is not difficult to verify that the above connection  obeys $\nabla^2=0$ and agrees at grades $-2,-1,0$ with the display before last. Also, the operator
$$
\nabla_X^{(0)}+{\sf gr}=
d_X
+\frac{1}{2i\hbar} (\hat S_+ \hat S_-+\hat S_-\hat S_+)+{\sf gr}
$$
commutes with $\hat S_+$. Using this, it is easy to see that for any homogeneous vector field $U$, Equation~\nn{eqgr} holds, and thus the connection $\nabla$ given above is $\R^+$-equivariant. By Theorem~\ref{biggy}, we have now determined a contact quantization of $S^3$. For each contact form $\alpha\in \bm \alpha$ this determines a nilpotent connection $\bm \nabla^\hbar$ acting on sections of   the contractor Hilbert bundle~${\mathscr H}(\FZ)$. This can be computed by pulling-back the ambient quantum connection form $\nabla$ to the codimension~2 submanifold given by $\{u=0,\tau=0\}\sse\R^5$, where $\tau$ is any (positive) scale obeying ${\mathcal L}_X \tau = \tau$. For this example, a valid choice~is the radius coordinate $\tau=r$. For this choice, the connection $\nabla$ pulls back to 
\begin{multline*}
\nabla_{S^3}:= d +
\frac{\lambda_i}{i\hbar} \Big((1+\hat S_+)^2 
-\frac{\hat H}{r^2}\Big) 
+
\frac{\lambda}{i\hbar }\,  A^\dagger\,  \sqrt{(1+\hat S_+)^2 -\frac{\hat H+\frac\hbar 2}{2 }}
+
\frac{\bar \lambda}{i\hbar r}\,   \sqrt{(1+\hat S_+)^2 -\frac{\hat H+\frac\hbar 2}{2 }}\: A
\, ,
\end{multline*}
where $\lambda_i=\alpha$ is the standard contact form on $S^3$ given above. 

Finally, to recover the strict quantization of the standard (periodic) Reeb dynamics discussed above, we  search for a parallel scale tractor $I$. For that, we must solve
$
\nabla^\Xi I = {\mathcal O}(A(X))
$
for $I\in \Gamma(\Xi)$ subject to ${\mathcal L}_X I=-I$. If we denote $I^A=E^A(I)=(I^+,I^1,I^2,I^-)$, we must study---see Equation~\nn{Ideterminenabla}---the system of equations:
$$
\left\{
\begin{array}{c}
d I^+ - Y I^+ +\frac{\lambda_K I^1+\lambda_J I^2}{r}+2 \lambda_I I^- = 0\, ,\\[2mm]
d I^1 +\frac{\lambda_I I^1+r\lambda_J I^-}{r^2} = 0\, ,\\[2mm] 
d I^2 +\frac{-\lambda_I I^2+r\lambda_K I^-}{r^2} = 0\, ,\\[2mm] 
d I^-+YI^-=0 \, . 
\end{array}
\right.
$$
This is readily solved by taking
$$
I^A=\begin{pmatrix}
r\\0\\0\\0
\end{pmatrix}\, ,
$$
so that we obtain
$$I=rY^\sharp=d r^\sharp=\frac 14\big( 
\frac xr\frac{\partial}{\partial y}-
\frac yr\frac{\partial}{\partial x}+
\frac zr\frac{\partial}{\partial w}-
\frac wr\frac{\partial}{\partial z}\big).$$
Given the parallel scale tractor $I$ and corresponding scale $\sigma=r$, we can compute the constraint
$$
s^-=s(I) = \hat S_A I^A=\hat S_+\, .
$$
This operator clearly commutes with the above connection $\nabla_{S^3}$.
Imposing the condition $\hat S_+=0$ on sections of the Hilbert bundle ${\mathscr H}(\FZ)$, this gives the section space  of the Hilbert bundle $\mathscr H$. Finally, acting on the kernel of $\hat S_+$, the connection $\nabla_{S^3}$ precisely reproduces the quantum connection~\nn{nablaS3strict}, upon identifying $\hat S_{1,2}$ with $\hat s_{1,2}$ and noting that $\hat H= \hat N+\frac\hbar 2$.\hfill$\Box$

\bibliographystyle{unsrt}
\bibliography{QCont_CHW}


\end{document}